\documentclass[a4paper,12pt]{amsart}

\usepackage{kantlipsum} 
\calclayout{}

\usepackage[ruled,vlined]{algorithm2e}
\usepackage{amsfonts}
\usepackage{amsmath}    
\usepackage{amssymb}    
\usepackage{bm}         
\usepackage{bbm}
\usepackage{color}
\usepackage{enumitem}
\usepackage{graphicx}   
\usepackage{natbib}    
\usepackage{epsf}
\usepackage{lscape}
\usepackage{enumitem}
\usepackage{longtable}
\bibpunct{(}{)}{;}{a}{,}{,}

\newtheorem{thm}{Theorem}
\newtheorem{lemma}{Lemma}
\newtheorem{cor}{Corollary}
\newtheorem{prop}{Proposition}

\newtheorem{defn}{Definition}

\newcommand{\beginsupplement}{%
        \setcounter{lemma}{0}
        \renewcommand{\thelemma}{S\arabic{lemma}}%
        \setcounter{cor}{0}
        \renewcommand{\thecor}{S\arabic{cor}}%
        \setcounter{table}{0}
        \renewcommand{\thetable}{S\arabic{table}}%
        \setcounter{figure}{0}
        \renewcommand{\thefigure}{S\arabic{figure}}%
        \setcounter{section}{0}
        \renewcommand{\thesection}{S\arabic{section}}%
        \setcounter{equation}{0}
        \renewcommand{\theequation}{S.\arabic{equation}}%
     }

\newcommand{\bbeta}{{\bm{\beta}}}

\newcommand{\tbeta}{{\tilde{\beta}}}

\newcommand{\tgamma}{{\tilde{\gamma}}}

\newcommand{\msf}[1]{\mathsf{#1}}

\usepackage{blindtext}
\usepackage[textwidth=22mm,textsize=tiny]{todonotes}

\thanks{DR was supported by Spanish Government grants Europa Excelencia EUR2020-112096, RYC-2015-18544, PGC2018-101643-B-I00. DR and AB acknowledge NIH grant R01 CA158113-01}

\begin{document}

\title{Approximate Laplace approximations for scalable model selection}
\author{David Rossell, Oriol Abril, Anirban Bhattacharya}
\date{}  


\begin{abstract}
We propose the approximate Laplace approximation (ALA) to evaluate integrated likelihoods, a bottleneck in Bayesian model selection. The Laplace approximation (LA) is a popular tool that speeds up such computation and equips strong model selection properties. However, when the sample size is large or one considers many models the cost of the required optimizations becomes impractical. ALA reduces the cost to that of solving a least-squares problem for each model. Further, it enables efficient computation across models such as sharing pre-computed sufficient statistics and certain operations in matrix decompositions.
We prove that in generalized (possibly non-linear) models ALA achieves a strong form of model selection consistency for a suitably-defined optimal model, at the same functional rates as exact computation.  We consider fixed- and high-dimensional problems, group and hierarchical constraints, and the possibility that all models are misspecified. We also obtain ALA rates for Gaussian regression under non-local priors, an important example where the LA can be costly and does not consistently estimate the integrated likelihood. Our examples include non-linear regression, logistic, Poisson and survival models. We implement the methodology in the R package \texttt{mombf}.
\end{abstract}

\maketitle

A main computational bottleneck in Bayesian model selection is evaluating integrated likelihoods, either when the sample size $n$ is large or there are many models to consider.
If said integrals can be obtained quickly, one can often use relatively simple algorithms to explore effectively the model space. For example, one may rely on the fast convergence of 
Metropolis--Hastings moves when posterior model probabilities concentrate \citep{yangyun:2016},
 sequential Monte Carlo methods to lower the cost of model search \citep{schafer:2013}, 
tempering strategies to explore model spaces with strong multi-modalities \citep{zanella:2019}, or adaptive Markov Chain Monte Carlo to reduce the effort in exploring low posterior probability models  \citep{griffin:2020}.
Unfortunately, except for very specific settings such as Gaussian regression under conjugate priors, the integrated likelihood has no closed-form, which seriously hampers scaling computations to even moderate dimensions.

Our main contribution is proposing a simple yet powerful approximate inference technique, the Approximate Laplace Approximation (ALA). Analogously to the classical Laplace approximation (LA) to an integral, ALA uses a second-order Taylor expansion, the difference being that the expansion is done at a point that simplifies calculations.
Also, there is a particular version of the ALA for which, within the exponential family, one may pre-compute statistics to obtain the ALA in all models. After said pre-computation, the computational cost does not depend on $n$.

We outline the idea.
Let $y=(y_1,\ldots,y_n)$ be an observed outcome of interest and suppose that one considers several models $\gamma \in \Gamma$ within some set of models $\Gamma$.
Given prior model probabilities $p(\gamma)$, Bayesian model selection assigns posterior probabilities $p(\gamma \mid y)= p(y \mid \gamma) p(\gamma)/p(y)$, where
\begin{align}
p(y \mid \gamma)= \int p(y \mid \eta_\gamma, \gamma) p(\eta_\gamma \mid \gamma) d\eta_\gamma
\label{eq:intlhood}
\end{align}
is the integrated likelihood, $p(y \mid \eta_\gamma, \gamma)$ the likelihood-function under model $\gamma$,
$\eta_\gamma \in \mathbb{R}^{p_\gamma}$ the model parameters, $p(\eta_\gamma \mid \gamma)$ their prior density, and $p(y)= \sum_{\gamma \in \Gamma} p(y \mid \gamma) p(\gamma)$.
The LA provides an approximation $\hat{p}(y \mid \gamma)$ using a Taylor expansion of the log-integrand in \eqref{eq:intlhood}
at the posterior mode $\hat{\eta}_\gamma$, giving
\begin{align}
\hat{p}(y \mid \gamma)=  p(y \mid \hat{\eta}_\gamma, \gamma) p(\hat{\eta}_\gamma \mid \gamma) (2 \pi)^{p_\gamma/2} |\hat{H}_\gamma|^{-\frac{1}{2}},
\label{eq:intlhood_la}
\end{align}
where $\hat{H}_\gamma$ is the log-integrand's negative hessian at $\hat{\eta}_\gamma$.
Although $\hat{p}(y \mid \gamma)$ is typically accurate, the optimization to obtain $\hat{\eta}_\gamma$ can be costly when $p_\gamma= \mbox{dim}(\eta_\gamma)$ is large, especially when one repeats such a calculation for many models. It is also costly when the sample size $n$ is large, since for most common models evaluating the likelihood and derivatives has a linear cost in $n$,  and sometimes higher (e.g. high-dimensional models where $p_\gamma$ grows with $n$).

ALA avoids the need to obtain $\hat{\eta}_\gamma$ by expanding the log-likelihood at a suitably-chosen initial value $\eta_{\gamma 0}$,  saving the associated optimization time to compute $\hat{\eta}_\gamma$. 
Let $\tilde{\eta}_\gamma= \eta_{\gamma 0} - H_{\gamma 0}^{-1} g_{\gamma 0}$ be a guess at $\hat{\eta}_\gamma$ given by a Newton--Raphson iteration from $\eta_{\gamma 0}$, where $g_{\gamma 0}$ and $H_{\gamma 0}$ are the gradient and hessian of the negative log-likelihood at $\eta_{\gamma 0}$.
A quadratic log-likelihood expansion at $\eta_{\gamma 0}$ (see Section \ref{sec:derivation_ala}) gives the ALA to the integrated likelihood
\begin{align}
\tilde{p}(y \mid \gamma)=  p(y \mid \eta_{\gamma 0}, \gamma) p(\tilde{\eta}_\gamma \mid \gamma) (2 \pi)^{p_\gamma/2} |H_{\gamma 0}|^{-\frac{1}{2}}
\exp\{ \frac{1}{2} g_{\gamma 0}^T H_{\gamma 0}^{-1} g_{\gamma 0} \},
\label{eq:intlhood_ala}
\end{align}
leading to ALA posterior probabilities $\tilde{p}(\gamma \mid y)= \tilde{p}(y \mid \gamma) p(\gamma)/\sum_{\gamma' \in \Gamma} \tilde{p}(y \mid \gamma') p(\gamma')$.
Figure \ref{fig:logistic} offers a simple illustration in a univariate logistic regression example.
See Section \ref{ssec:derivation_ala_general} for an alternative ALA based on expanding the full integrand, which attains the same rates as \eqref{eq:intlhood_ala} under mild conditions and performed similarly in our examples.

We focus attention in regression problems where setting coefficients in $\eta_{\gamma 0}$ to 0  results in further simplifications, particularly in exponential family models where it allows pre-computing sufficient statistics. 
The computational savings are substantial, see Figure \ref{fig:logistic_sims} and Figure \ref{fig:simpoisson} for logistic and Poisson regression examples.
 Even when sufficient statistics are not available the savings from avoiding the optimization exercise can still be significant, see our survival model examples in Table \ref{tab:cputime_survival}. 


A caveat is that, unlike LA, in general $\tilde{p}(y \mid \gamma)$ does not consistently estimate $p(y \mid \gamma)$ as $n \rightarrow \infty$.
For concave log-likelihoods the LA has relative error converging to 1 in probability, under the minimal condition that $\hat{H}_\gamma/n$ converges in probability to a positive-definite matrix (\cite{rossell:2019t}, Proposition 8). Under further conditions, the LA estimates Bayes factors with relative error of order $1/n^2$ \citep{kass:1990}, see also \cite{ruli:2016} for higher-order approximations to high-dimensional integrals.
The ALA does not equip such properties.
 Figure \ref{fig:logistic} (right) illustrates a situation where, due to the posterior distribution concentrating far from the expansion point, ALA significantly underestimates the integral.  
Nevertheless ALA attains a strong type of model selection consistency, even
when (inevitably) models are misspecified, that is the data are truly generated by a distribution $F^*$ outside the considered models.
Specifically, we prove that the ALA posterior probability $\tilde{p}(\tgamma^* \mid y)$ converges to 1 in the $L_1$ sense for a suitably-defined optimal model $\tgamma^*$. 
Said $\tgamma^*$ is in general different from the model $\gamma^*$ asymptotically recovered by exact calculations. 
Under misspecification, neither $\gamma^*$ nor $\tgamma^*$ in general recover the set of 
covariates associated to the mean of $y$ under the data-generating $F^*$,
but the optimal covariates under two different implicit loss functions. 
However we show via examples that $\tgamma^*$ and $\gamma^*$ often coincide, and provide a sufficient condition for $\tgamma^*$ to discard truly spurious parameters.
 Intuitively, the reason why situations like that in Figure \ref{fig:logistic} (right) need not be problematic is that Bayes factors target ratios of integrated likelihoods. Despite the integral being under-estimated, it is still very large relative to the likelihood at 0 is very large, signaling that the parameter should be included.

Relative to approximate inference methods primarily designed for estimation or prediction such as variational Bayes \citep{jordan:1999} or expectation propagation \citep{minka:2001}, ALA focuses on model selection problems where the goal is structural learning; see, however, \cite{carbonetto:2012} and \cite{huang_xichen:2016} for variational Bayes approaches to variable selection in Gaussian regression with conjugate priors. 
We focus our study on a wide model class within the exponential family, which includes generalized linear, generalized additive models and other (possibly non-additive) generalized structured regression.
We also illustrate the use of ALA with concave log-likelihoods outside the exponential family, via a non-linear additive accelerated failure time model \citep{rossell:2019t}.
We incorporate two aspects where state-of-the-art methods encounter difficulties.
First, we consider that the model selection exercise may combine group and hierarchical constraints, a case where penalized likelihood and shrinkage prior methods can face difficulties in terms of computational complexity.
Said constraints are relevant when one considers categorical covariates, interaction terms, and semi- and non-parametric covariate effects, for example.
Second, we use ALA to facilitate computation for non-local priors \citep{johnson:2010,johnson:2012}.
Non-local priors attain some of the strongest theoretical properties among Bayesian methods in high dimensions,
see \cite{shin:2018} and \cite{rossell:2018}.
However, exact calculations are unfeasible, the LA is costly and it does not consistently estimate $p(y \mid \gamma)$ for any model $\gamma$ that includes truly spurious parameters \citep{rossell:2017}.
Interestingly, although we generally view the ALA as fast approximate inference that may perform slightly worse than the LA, for non-local priors the ALA often attains better inference.

The paper is structured as follows.
Section \ref{sec:likelihood} reviews exponential family models and discusses computational savings associated to the ALA.
Although the ALA applies to a wide set of priors, for concreteness Section \ref{sec:prior} outlines specific local and a non-local priors on parameters that we use in our examples, and a group hierarchical prior on models $p(\gamma)$.
The development of the non-local prior is in fact a secondary contribution of this paper:
it is a novel class combining additive penalties \citep{johnson:2010} suitable for group constraints with product-type penalties required for high-dimensional consistency \citep{johnson:2012}.
Section \ref{sec:ala} gives specific ALA expressions for local priors, and subsequently for the more challenging non-local prior case.
Section \ref{sec:theory} gives model selection consistency results for ALA, specifically rates that hold for fixed $p$ under minimal conditions and high-dimensional rates where $p$ grows with $n$, under slightly stronger conditions.
We distinguish cases where the exponential family has a known dispersion parameter (e.g. logistic and Poisson regression) and cases where it is unknown. In particular our high-dimensional theory focuses on the known case, to alleviate the technical exposition, but our results also apply to Gaussian outcomes with unknown error variance. 
Section \ref{sec:results} shows examples assessing the numerical accuracy of ALA, the computational time, and the quality of its associated model selection.
We consider logistic, Poisson and survival examples, as well as non-linear Gaussian regression under non-local priors where $p(y \mid \gamma)$ are hard to approximate.  We also briefly illustrate the use of ALA in combination with importance sampling, variable screening and adding optimization iterations to improve the expansion point $\eta_{\gamma 0}$ in \eqref{eq:intlhood_ala}. 
Section \ref{sec:discussion} concludes.
The supplementary material contains proofs, derivations and supplementary results.
 R code and data to reproduce our examples are available at 
\small{https://github.com/davidrusi/paper\_examples/tree/main/2020\_Rossell\_Abril\_Bhattacharya\_ALA}. 


\begin{figure}
\begin{center}
\begin{tabular}{cc}
 $n=100$, true $\beta_\gamma^*=0.405$ &  $n=200$, true $\beta_\gamma^*=1.099$ \\
\includegraphics[width=0.5\textwidth]{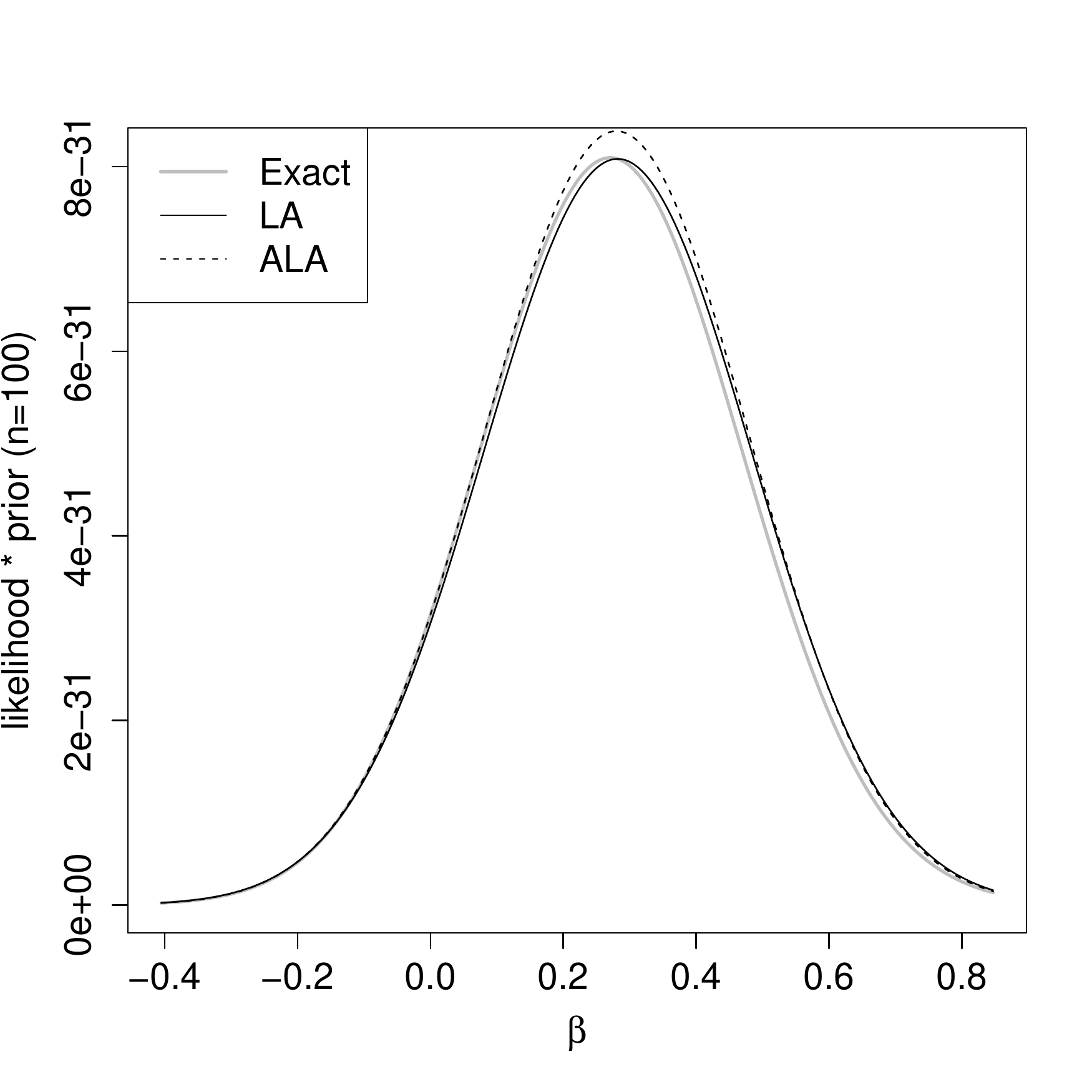}  &
\includegraphics[width=0.5\textwidth]{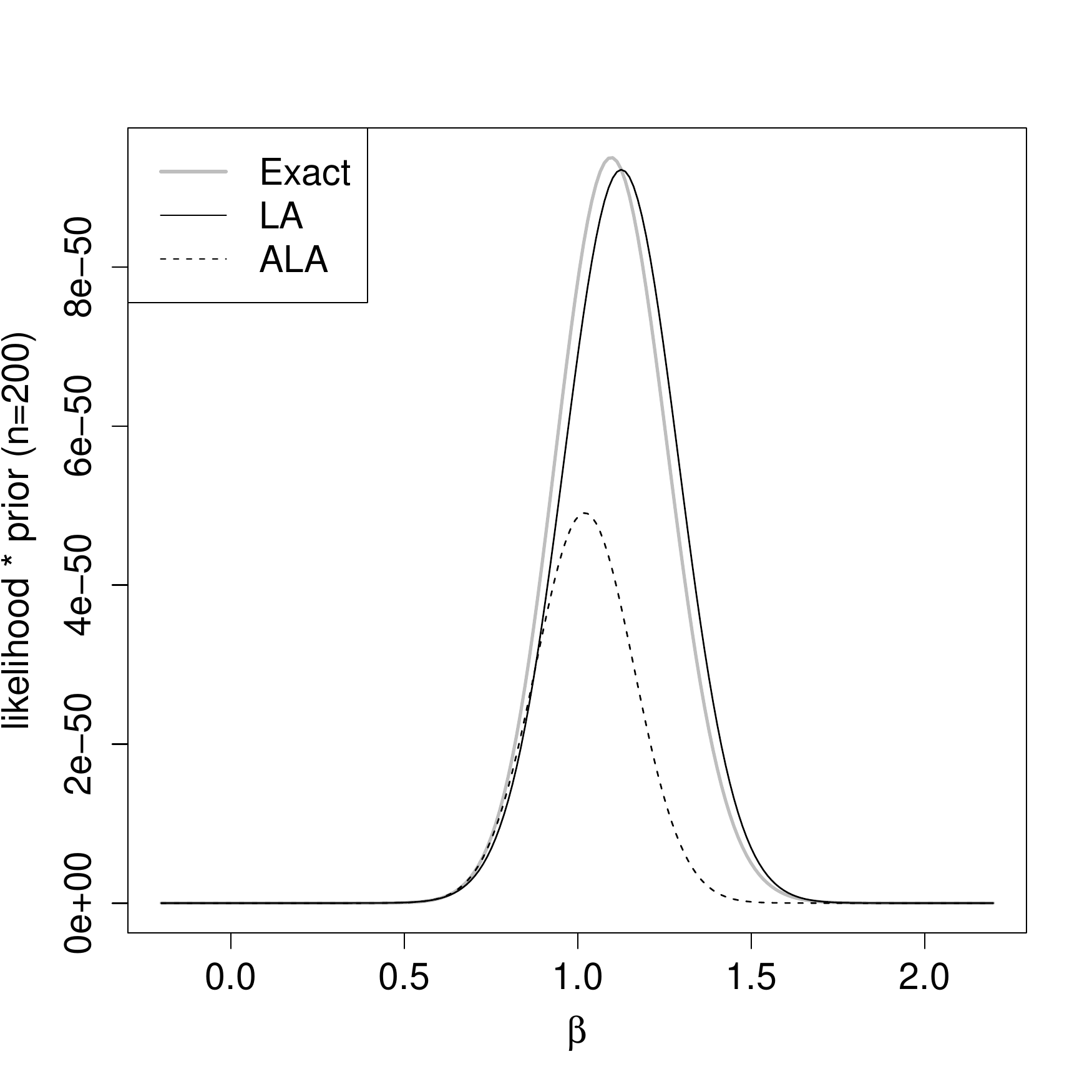}
\end{tabular}
\end{center}
\caption{Logistic regression simulation in one dimension with a standard Gaussian prior on the coefficient. 
The likelihood multiplied by the prior $p(y \mid \beta_\gamma,\gamma) N(\beta_\gamma;0,1)$ is plotted in grey (Exact). The solid black line (LA) plots an approximation by replacing the log-likelihood with a second-order Taylor expansion at the MLE. The dashed line (ALA) does the same with a quadratic expansion around zero.
}
\label{fig:logistic}
\end{figure}

\section{Likelihood}\label{sec:likelihood}

We lay out notation. Let $x_i \in \mathcal{X}$ be covariates taking values in some domain $\mathcal{X}$.
Consider a generalized structured regression with predictor
\begin{align}
h(E(y_i \mid x_i)) = \sum_{j=1}^J z_{ij}^T \beta_j,
\label{eq:glm_mean}
\end{align}
where $h()$ is the canonical link function and
$z_i= (z_{i1}^T,\ldots,z_{iJ}^T)^T$ a basis for the effect of $x_i$ with coefficients $\beta=(\beta_1^T,\ldots,\beta_J^T)^T$.
 For example a standard generalized linear model corresponds to $z_i=x_i$.
We also consider situations where  each $z_{ij} \in \mathbb{R}^{p_j}$ defines a group with $p_j$ elements,
e.g. multiple binary indicators for a categorical covariate or a non-linear basis expansion for a continuous covariate.
That is, \eqref{eq:glm_mean} includes 
additive regression on functions of $x_i$, non-linear interactions between elements of $x_i$, for example. 
Let $p= \sum_{j=1}^J p_j$ be the total number of parameters and $Z=(z_1^T,\ldots,z_n^T)^T$ the $n \times p$ design matrix.

Our goal is to determine which $\beta_j \in \mathbb{R}^{p_j}$ should be set to zero. 
Let $\gamma_j= \mbox{I}(\beta_j \neq 0)$ for $j=1,\ldots,J$ be group inclusion indicators, so that $\gamma= (\gamma_1,\ldots,\gamma_J)$ indexes the model.
We denote by $Z_\gamma$ the $n \times p_\gamma$ submatrix of $Z$ with (blocks of) columns selected by $\gamma$ where $p_\gamma = \sum_{j: \gamma_j = 1} p_j$, and by $z_{\gamma i}$ its $i^{th}$ row.
For any given model $\gamma$, the distribution of $y$ is assumed to be in the exponential family with canonical link and likelihood function
\begin{align}
  p(y \mid \beta, \phi, \gamma)= \exp \left\{ [y^T Z_\gamma \beta_\gamma - \sum_{i=1}^n b(z_{\gamma i}^T \beta_\gamma)]/\phi + \sum_{i=1}^n c(y_i, \phi) \right\},
\label{eq:glm_pdf}
\end{align}
where $\phi>0$ is an optional dispersion parameter and $b()$ an infinitely differentiable function.
The gradient and hessian of the negative log-likelihood $-\log p(y \mid \beta, \phi,\gamma)$ are
\begin{align}
g_\gamma(\beta_\gamma, \phi)
  &= -\frac{1}{\phi }\begin{pmatrix}
    Z_\gamma^T y - \sum_{i=1}^n b'(z_{\gamma i}^T\beta_\gamma)  z_{\gamma i} \\
    - [y^T Z_\gamma \beta_\gamma - \sum_{i=1}^n b(z_{\gamma i}^T \beta_\gamma)]/\phi + \phi \sum_{i=1}^n \nabla_\phi c(y_i,\phi)
  \end{pmatrix}
\nonumber \\
H_\gamma(\beta_\gamma, \phi)
&= \frac{1}{\phi} \begin{pmatrix}
    Z_\gamma^T D_\gamma Z_\gamma   & g_\beta(\beta_\gamma,\phi) \\
    g_\beta(\beta_\gamma,\phi)^T & -2[y^T Z_\gamma\beta_\gamma - \sum_{i=1}^n b(z_{\gamma i}^T \beta_\gamma)]/\phi^2 - \phi \sum_{i=1}^n \nabla^2_{\phi\phi} c(y_i,\phi)
 \end{pmatrix}
\nonumber
\end{align}
where $D_\gamma$ is an $n \times n$ diagonal matrix with diagonal entry $b''(z_{\gamma i}^T \beta_\gamma)>0$.
For completeness, Section \ref{sec:logl_popularmodels} provides expressions for logistic and Poisson models.

A computationally-convenient choice for the ALA in \eqref{eq:intlhood_ala} is to set 
 a global $\beta_0 \in \mathbb{R}^p$ and let $\eta_{\gamma 0}=(\beta_{\gamma 0},\phi_0)$,
where $\beta_{\gamma 0}$ contains the entries of $\beta_0$ selected by $\gamma$ 
and, if $\phi$ is a unknown parameter, 
\begin{align}
\phi_0= \arg\max_\phi p(y \mid \beta=\beta_0, \phi)
\end{align}
is the maximum likelihood estimator conditional on $\beta=\beta_0$. 
Since $\phi_0$ does not depend on $\gamma$ it can be computed upfront and shared across all models. 
 By basing ALA on such a global choice, one avoids the model-specific optimization costs that would be required by a LA. 

 The choice $\beta_0=0$ gives further computational simplifications (one may also set the intercept to a non-zero value at essentially no cost). 
To ease notation let $\tilde{y}= (y - b'(0) \mathbbm{1})/ b''(0)$ denote a shifted and scaled version of $y$, $\mathbbm{1}=(1,\ldots,1)^T$ being the $n \times 1$ unit vector. The gradient and hessian at $(\beta_{\gamma 0},\phi)=(0,\phi_0)$ are
\begin{align}
  g_{\gamma 0}&=
-\frac{b''(0)}{\phi_0} \begin{pmatrix} Z_\gamma^T \tilde{y} \\ 0 \end{pmatrix} \nonumber \\
  H_{\gamma 0}&= \frac{b''(0)}{\phi_0} \begin{pmatrix} Z_\gamma^TZ_\gamma & -Z_\gamma^T \tilde{y}/\phi_0 \\
-\tilde{y}^T Z_\gamma/\phi_0 & s(\phi_0) \end{pmatrix},
\label{eq:logl_gradhess0}
\end{align}
where $s(\phi_0)= [2n b(0)/\phi_0^2 + \phi_0 \sum_{i=1}^n \nabla_{\phi \phi}^2 c(y_i,\phi_0)]/b''(0)$.
To interpret these expressions, the exponential family predicted variance for $\beta_{\gamma 0}=0$ is
$V(y_i \mid z_{\gamma i},\beta_{\gamma 0}=0, \phi)= \phi b''(0)$, hence $V(\tilde{y}_i \mid z_{\gamma i},\beta_{\gamma 0}=0,\phi)= \phi/b''(0)$.
Thus, $(g_{\gamma 0},H_{\gamma 0})$ are analogous to the gradient and hessian in a least-squares regression of $\tilde{y}$ on $Z$, with model-based variance $\phi/b''(0)$.

Sections \ref{sec:ala}-\ref{sec:theory} show that $\eta_{\gamma 0}=(0,\phi_0)^T$ leads to desirable model selection rates.
The ALA then basically requires least-squares type computations
where $(Z^T\tilde{y}, Z^TZ)$ play the role of sufficient statistics that can be computed upfront and shared across all models.
To further save memory and computational requirements,
in our implementation we store $Z^TZ$ in a sparse matrix that is incrementally filled the first time that any given entry is required.
 That is, when searching models typically many elements in $Z^TZ$ are never used, hence there is no need to compute nor to allocate them to memory beforehand. 
One may also consider alternative $\eta_{\gamma 0}$, say obtained after a few Newton--Raphson iterations, in an attempt to obtain an ALA that is closer to the LA in \eqref{eq:intlhood_la}.  See Figure \ref{fig:logistic_sims} and Sections \ref{ssec:binary_example} and \ref{ssec:ala_importancesampling} for some examples. Such alternatives can lead to improved inference, at a higher computational cost. Their theoretical study requires a separate treatment, however, and is left for future work.  

\section{Prior}\label{sec:prior}

Most of our results apply to  a wide class of priors  $p(\gamma)$. 
For concreteness we outline a structure that assigns the same probability to all models with the same number of active groups $|\gamma|=\sum_{j=1}^J \gamma_j$, and an arbitrary distribution $p(|\gamma|)$ on $|\gamma|$.

 Although unnecessary in canonical regression problems, we also consider 
that in certain situations one may want to  impose hierarchical constraints, in the sense that $\beta_l=0$ implies $\beta_j=0$ for some $l \neq j$.
For instance, one may exclude interaction terms unless the corresponding main effects are present,
or decompose non-linear effects as a linear plus a non-linear term, and only include the latter if the linear term is present
\citep{scheipl:2012,rossell:2019t}.
Such (optional) constraints can be added as follows. 

Let $C \subseteq \Gamma$ be the models satisfying the constraints.
These are easily incorporated by assigning $\pi(\gamma)=0$ to any $\gamma \not\in C$. 
Specifically,
\begin{align}
  p(\gamma)=
\begin{cases} 
 K p(|\gamma|) {J \choose |\gamma|}^{-1} \mbox{ , if } \gamma \in C \mbox{ and }|\gamma| \leq \bar{J}
\\
 0 \mbox{ , otherwise}
\end{cases}
\label{eq:prior_models}
\end{align}
where $\bar{J}$ is the maximum model size one wishes to consider
and $K$ a prior normalization constant that does not need to be evaluated explicitly.
The formulation allows both a number of parameters $p \gg n$ and groups $J \gg n$, but restricts the model space to using combinations of at most $\bar{J}$ groups. 
Given that models with $p_\gamma \geq n$ parameters result in data interpolation, typically one sets both $\bar{J} \ll n$ and $p_\gamma \ll n$ for any allowed $\gamma \in C$, see Section \ref{ssec:bf_highdim} for further discussion.
 In a standard generalized linear model without groups nor constraints; $J=p$, the constraint $\gamma \in C$ is removed, and $K=1$. 

In Section \ref{sec:theory} we provide pairwise Bayes factor rates for general $p(|\gamma|)$, whereas to ease exposition for posterior model probability rates we focus on
\begin{align}
 p(|\gamma|) \propto p^{-c |\gamma|},
\nonumber
\end{align}
where $c \geq 0$ is a user-specified constant and $\propto$ denotes ``proportional to''. 
For $c=0$ one obtains uniform $p(|\gamma|)= 1/(\bar{J}+1)$, 
which generalizes the Beta-Binomial(1,1) distribution advocated by \cite{scott:2006} to a setting where there may be groups and hierarchical constraints. 
In our experience $c=0$ strikes a good balance between sparsity and retaining power to detect truly non-zero coefficients, hence in all our examples we used $c=0$.
One may also set $c>0$, which is motivated by the so-called Complexity priors of \cite{castillo:2015}. 
These set a stronger prior penalty on the model size that leads to faster rates to discard spurious parameters, at the cost of slower rates to detect active parameters. See Section \ref{sec:theory} for further details.

We remark that adding hierarchical constraints to penalized likelihood and Bayesian shrinkage frameworks lead to computational difficulties. For instance, hierarchical constraints for LASSO penalties lead to a challenging optimization problem, and while one can devise relaxed constraints \citep{bien:2013}, our examples indicate the computation can be prohibitive.

\subsection{Group product priors}\label{ssec:prior_parameters}

Regarding the prior on parameters, our examples use Normal priors and a novel group moment (gMOM) prior family factorizing over groups
\begin{align}
  p^L(\beta_\gamma \mid \phi, \gamma)&= \prod_{\gamma_j=1} N\left(\beta_j; 0, \frac{\phi g_L n}{p_j} (Z_j^T Z_j)^{-1} \right)
                               \nonumber \\
p^N(\beta_\gamma \mid \phi, \gamma)&= \prod_{\gamma_j=1} \frac{\beta_j^T Z_j^T Z_j
\beta_j}{\phi g_N  n p_j/(p_j+2)} N\left(\beta_j; 0, \frac{\phi g_N n}{p_j+2} (Z_j^T Z_j)^{-1} \right)
\label{eq:prior_parameters}
\end{align}
and, for models where $\phi$ is unknown, we set $p(\phi)= \mbox{IG}(\phi; a,b)$ for given prior parameters $g_L,g_N,a,b > 0$.
 All other parameters ($\beta_j$ such that $\gamma_j=0$) are zero with probability 1. 
Both priors feature a Normal kernel with a group-Zellner precision matrix given by $Z_j^T Z_j$.
Other covariances may be used, but our choice leads to inference that is robust to affine
within-group reparameterizations of $\beta_j$ (for example, changing the reference category for discrete predictors), and to simple default parameter values $g_L=g_N=1$ (Section~\ref{ssec:prior_elicitation}).

 The group Zellner $p^L()$ is a local prior, in the nomenclature of \cite{johnson:2010}, whereas the gMOM $p^N()$ is a non-local prior. The defining property of non-local priors is that the density vanishes as $\beta_\gamma$ approaches any value that lies in the parameter space of a submodel of $\gamma$, i.e. $\beta_j=0$ in our setting. Their interest is that they help discard spurious parameters, by inducing a data-dependent penalty that has little asymptotic effect on power \citep{rossell:2017}.
Earlier proposals \citep{johnson:2010,johnson:2012} did not account for group structure, however. 
The intuition is simple, the gMOM penalizes groups with small contributions $\beta_j^T Z_j^T Z_j \beta_j$ relative to its size $p_j=\mbox{dim}(\beta_j)$, which helps induce sparsity.

\subsection{Prior elicitation}\label{ssec:prior_elicitation}

Although our focus is computational and our theory applies to any prior parameters $(g_L,g_N)$ (under minimal conditions), 
we outline a simple strategy to obtain default $(g_L,g_N)$ that we used in our examples.
The strategy builds upon the unit information prior, a popular default leading to the Bayesian information criterion \citep{schwarz:1978},
the difference being that we account for the presence of groups in $Z_\gamma$.

Suppose that there were no groups in $Z_\gamma$.
The unit information prior can be interpreted as containing as much information as a single observation.
Another (perhaps more natural) interpretation is its specifying the prior belief that
$E(\beta_\gamma^T Z_\gamma^T Z_\gamma \beta_\gamma/[n\phi])= p_\gamma$.
The expected contribution
$\beta_\gamma^T Z_\gamma^T Z_\gamma \beta_\gamma/n$
relative to the dispersion $\phi$, which is a measure of the predictive ability contained in $Z_\gamma$, is given by the number of variables $p_\gamma$.

Suppose now that a variable defines a group of columns in $Z_\gamma$, for instance a non-linear basis expansion. Then $p_\gamma$ depends on the basis dimension, which is often chosen arbitrarily.
The unit information prior would imply the belief that the predictive power of $Z_\gamma$
increases with the arbitrary basis dimension, rather than the number of variables $\sum_{j=1}^J \gamma_j$.
Instead, we set prior parameters
such that the prior expected predictive power depends on $\sum_{j=1}^J \gamma_j$ and is unaffected by the basis dimension.
That is, we set $(g_L,g_N)$ such that
$$
E\left( \frac{\beta_\gamma^T Z_\gamma^T Z_\gamma \beta_\gamma}{n\phi} \right)=
E\left( \sum_{\gamma_j=1} \frac{\beta_j^T Z_j^T Z_j \beta_j}{n\phi} \right)=
\sum_{j=1}^J \gamma_j.
$$

For the Normal and gMOM priors in~\eqref{eq:prior_parameters} this rule gives $g_L=g_N=1$, see Section \ref{sec:default_priordispersion}.

\section{Approximate Laplace approximation}\label{sec:ala}

We first discuss the ALA under a local prior $p^L()$, and subsequently that for the gMOM prior $p^N()$.
Recall that for the latter the integrated likelihood has no computationally-convenient closed-form, even in Gaussian regression.

\subsection{Local priors}
\label{sec:ala_localprior}

The ALA to the Bayes factor between any pair of models $(\gamma,\gamma')$ is
\begin{align}
\tilde{B}^L_{\gamma \gamma'}=
  \frac{\tilde{p}^L(y \mid \gamma)}{\tilde{p}^L(y \mid \gamma')}=
\exp \left\{\frac{1}{2} (g_{\gamma 0}^T H_{\gamma 0}^{-1} g_{\gamma 0} -  g_{\gamma' 0}^T H_{\gamma' 0}^{-1} g_{\gamma' 0}) \right\} (2 \pi)^{\frac{p_\gamma - p_{\gamma'}}{2}}
\frac{|H_{\gamma' 0}|^{\frac{1}{2}} p^L(\tilde{\eta}_\gamma \mid \gamma)}{|H_{\gamma 0}|^{\frac{1}{2}} p^L(\tilde{\eta}_{\gamma'} \mid \gamma')}.
\label{eq:bfala}
\end{align}

Expression \eqref{eq:bfala} can be used beyond the exponential family, provided the log-likelihood is concave.
As an example, Section \ref{ssec:survival_example} illustrates the Gaussian accelerated failure time model; see Section \ref{ssec:logl_aft} for the corresponding log-likelihood, gradient and derivatives, and the conditions for log-likelihood concavity.

We now provide specific expressions for exponential family models \eqref{eq:glm_pdf}, and discuss a curvature adjustment designed to improve finite $n$ performance.
Consider first the case where $\phi$ is a known constant, as in logistic and Poisson models.
Taking $\eta_{\gamma 0}=\beta_{\gamma 0}= 0$ gives
\begin{align}
  \tilde{B}^L_{\gamma \gamma'}=
    \exp \left\{\frac{b''(0)}{2\phi} (\tbeta_\gamma^T Z_\gamma^T Z_\gamma \tbeta_\gamma - \tbeta_{\gamma'}^T Z_{\gamma'}^T Z_{\gamma'} \tbeta_{\gamma'}) \right\}
  \left( \frac{2 \pi \phi}{b''(0)} \right)^{\frac{p_\gamma - p_{\gamma'}}{2}}
  \frac{|Z_{\gamma'}^T Z_{\gamma'}|^{\frac{1}{2}} p^L(\tbeta_\gamma \mid \phi, \gamma)}
  {|Z_{\gamma}^T Z_{\gamma}|^{\frac{1}{2}} p^L(\tbeta_{\gamma'} \mid \phi, \gamma')}
  \label{eq:bfala_expfamily_knownphi}
\end{align}
where $\tbeta_\gamma= (Z_\gamma^T Z_\gamma)^{-1} Z_\gamma^T \tilde{y}$ and $\tilde{y}= (y - b'(0) \mathbbm{1})/ b''(0)$.
See Section \ref{ssec:derivation_ala_knownphi} for the derivation.

Consider now the case where $\phi$ is an unknown model parameter. Then, taking $\eta_{\gamma 0}= (0,\phi_0)$ as in Section \ref{sec:likelihood}, one obtains
\begin{align}
\tilde{B}^L_{\gamma,\gamma'}=
\exp \left\{ \frac{b''(0)}{2\phi_0}[t_\gamma \tbeta_\gamma^T Z_\gamma^T Z_\gamma \tbeta_\gamma- t_{\gamma'} \tbeta_{\gamma'}^T Z_{\gamma'}^T Z_{\gamma'} \tbeta_{\gamma'} ] \right\}
(2 \pi)^{\frac{p_\gamma-p_{\gamma'}}{2}}
\frac{|H_{\gamma' 0}|^{\frac{1}{2}} p^L(\tbeta_\gamma, \tilde{\phi}_\gamma \mid \gamma)}{|H_{\gamma 0}|^{\frac{1}{2}} p^L(\tbeta_{\gamma'}, \tilde{\phi}_{\gamma'} \mid \gamma')}
\label{eq:bfala_expfamily_unknownphi}
\end{align}
where
$$
t_\gamma= 1 + \frac{\tbeta_\gamma^T Z_\gamma^TZ_\gamma \tbeta_\gamma}{\phi_0^2 (s(\phi_0) - \tbeta_\gamma^T Z_\gamma^TZ_\gamma \tbeta_\gamma)},
$$
and $s(\phi_0)$ is as in \eqref{eq:logl_gradhess0},
see Section \ref{ssec:derivation_ala_unknownphi} for the derivation.

\subsection{Curvature adjustment}
\label{ssec:curvature_adjustment}

The Bayes factor in \eqref{eq:bfala_expfamily_knownphi} for models where $\phi$ is known attains desirable theoretical properties as $n \rightarrow \infty$, see Section \ref{sec:theory}.
There is however an important practical remark, which makes us recommend a curvature-adjusted ALA to improve finite $n$ performance.
We outline the idea and refer the reader to Section \ref{ssec:ala_curvature_adjustment} for a full description.
Expression \eqref{eq:bfala} can be given in terms of the model-predicted covariance $\textnormal{Cov}(y \mid Z_\gamma, \beta=\beta_{\gamma 0},\phi)=E((y-\mu_{\gamma 0})^T(y-\mu_{\gamma 0}) \mid Z_\gamma, \beta=\beta_{\gamma 0},\phi)$, where $\mu_{\gamma 0}= E(y \mid Z_\gamma, \beta=\beta_{\gamma 0}, \phi)$.
Even when the data are truly generated from a distribution $F^*$ included in the assumed model \eqref{eq:glm_pdf} for some $\beta_\gamma^*$, there is a mismatch between said covariance and $E_{F^*}((y-\mu_{\gamma 0})^T(y-\mu_{\gamma 0}) \mid Z_\gamma )$, due to $\mu_{\gamma 0}$ being different from the true mean $E(y \mid Z_\gamma, \beta=\beta_\gamma^*, \phi)$.
That is, the data may be either over- or under-dispersed relative to the model prediction at $\beta=\beta_{\gamma 0}$, which can adversely affect inference.

The curvature-adjusted ALA is obtained by replacing $\textnormal{Cov}(y \mid Z_\gamma, \beta=\beta_{\gamma 0},\phi)$ for
$\hat{\rho} \, \textnormal{Cov}(y \mid Z_\gamma, \beta=\beta_{\gamma 0},\phi)= \hat{\rho} \, \phi \, \textnormal{diag}(b''(Z_\gamma\beta_{\gamma 0}))$,
where 
$
\hat{\rho}= \sum_{i=1}^n (y_i - \bar{y})^2/[\phi b''(h(\bar{y})) (n-1)]
$
is a Pearson residual-based estimate of over-dispersion, $h()$ is the link function in \eqref{eq:glm_mean} and $\bar{y}= \sum_{i=1}^n y_i /n$ the sample mean.
The curvature-adjusted Bayes factor is
\begin{align}
\tilde{B}_{\gamma \gamma'}=
\frac{p(\tbeta_\gamma \mid \phi, \gamma) |Z_{\gamma'}^T Z_{\gamma'}|^{1/2}}{p(\tbeta_{\gamma'} \mid \phi, \gamma') |Z_\gamma^T Z_\gamma|^{1/2}}
\left(\frac{2\pi \phi}{\hat{\rho} b''(h(\bar{y}))}\right)^{(p_\gamma-p_{\gamma'})/2}
e^{\frac{b''(h(\bar{y}))}{2 \hat{\rho} \phi}  [\tbeta_\gamma^T  (Z_\gamma^T Z_\gamma) \tbeta_\gamma - \tbeta_{\gamma'}^T (Z_{\gamma'}^T Z_{\gamma'}) \tbeta_{\gamma'}]},
 \label{eq:bfala_adj_expfamily}
\end{align}
where $\tbeta_\gamma= (Z_\gamma^T Z_\gamma)^{-1} Z_\gamma^T [y - b'(h(\bar{y})) \mathbbm{1}]/b''(h(\bar{y}))$. Note that $\tbeta_\gamma$ here denotes the parameter estimate after one Newton--Raphson iteration from the maximum likelihood estimator under the intercept-only model. It is possible to use alternative over-dispersion estimators that are specific for each model, at a slightly higher computational cost, see Section \ref{ssec:ala_curvature_adjustment} for a discussion. In all our logistic and Poisson regression examples we used $\hat{\rho}$ as outlined above, since in our experience this simple choice performs fairly well in practice.
See Section \ref{ssec:binary_example} and Figure \ref{fig:simpoisson} for a Poisson example where the curvature adjustment significantly improves inference.

\subsection{Non-local priors}
\label{sec:ala_nonlocalprior}

The ALA Bayes factors $\tilde{B}^N_{\gamma \gamma'}= \tilde{p}^N(y \mid \gamma) / \tilde{p}^N(y \mid \gamma')$ for the gMOM prior in \eqref{eq:prior_parameters} require an alternative strategy. 
Let $\pi(\beta_\gamma,\phi \mid \gamma)= \mathcal{N}(\beta_\gamma;0,\phi V_j^{-1})$
where $V_j= Z_j^T Z_j (p_j+2)/ (n p_j g_N)$, so that the gMOM prior equals
$$p^N(\beta_\gamma \mid \phi, \gamma)= \pi(\beta_\gamma \mid \phi, \gamma) \prod_{\gamma_j=1} \frac{\beta_j^T V_j \beta_j}{\phi p_j}.$$

Denote by $\pi(\beta_\gamma,\phi \mid y,\gamma)$ the posterior density and $\pi(y \mid \gamma)$ the integrated likelihood associated to the prior $\pi(\beta_\gamma,\phi \mid \gamma)$.
By Proposition 1 in \cite{rossell:2017}, the identity
\begin{align}
p^N(y \mid \gamma)= \pi(y \mid \gamma)
\int \left( \prod_{\gamma_j=1} \frac{\beta_j^T V_j \beta_j}{\phi p_j} \right)
\pi(\beta_\gamma,\phi \mid y,\gamma) d \beta_\gamma d\phi
\label{eq:intlhood_nlp}
\end{align}
holds exactly.
$\pi(y \mid \gamma)$ is the integrated likelihood under a local prior, hence one may obtain an ALA $\tilde{\pi}(y \mid \gamma)$ as described in Section \ref{sec:ala_localprior}.
The second term in \eqref{eq:intlhood_nlp} is the posterior expectation of a product,
and its computation requires a number of operations that grow exponentially with the model dimension $p_\gamma$.
As an alternative, in \eqref{eq:intlhood_nlp} we replace $\pi(\beta_\gamma \mid \phi,y)$ by its ALA-based normal approximation
and we also replace the integral by a product of expectations.
Specifically,
\begin{align}
 \tilde{p}^N(y \mid \gamma)= \tilde{\pi}(y \mid \gamma)
\left[ \prod_{\gamma_j=1} \int \frac{\beta_j^T V_j \beta_j}{\phi p_j}
\mathcal{N}(\beta_j; m_j, \phi S_j)
\tilde{\pi}(\phi \mid y,\gamma) d \beta_j d\phi \right]
\label{eq:ala_nlp}
\end{align}
where
$m_j$ and $S_j$ are the sub-vector of $\tbeta_\gamma= (Z_\gamma^T Z_\gamma)^{-1} Z_\gamma^T \tilde{y}$
and sub-matrix of $(Z_\gamma^T D_\gamma Z_\gamma)^{-1}$ associated to $\beta_j$,
and recall that $D_\gamma$ is diagonal with $(i,i)$ entry $z_{i \gamma}^T \tbeta_\gamma$.
The following lemma is useful.

\begin{lemma}
  Let $A$ and $S$ be $l \times l$ full-rank matrices and $a,b>0$ be constants.
  Then
  \begin{align}
  \int \frac{\xi^T A \xi}{\phi}  \mathcal{N}(\xi; m, \phi S) d\xi
  &=   \mbox{tr}(A S) + \frac{m^T A m}{\phi}
    \nonumber \\
  \int \int \frac{\xi^T A \xi}{\phi}  \mathcal{N}(\xi; m, \phi S) \mathrm{IG}(\phi; a,b) d\xi d\phi
  &=
  \mbox{tr}(A S) + \frac{a}{b} m^T A m.
    \nonumber
  \end{align}
  \label{lem:mean_pmompenalty}
\end{lemma}

By Lemma \ref{lem:mean_pmompenalty} when $\phi$ is known the integral in \eqref{eq:ala_nlp} has the simple expression
\begin{align}
\tilde{p}^N(y \mid \phi,\gamma)= \tilde{\pi}(y \mid \phi,\gamma)
\left[ \prod_{\gamma_j=1} \mbox{tr}(V_j S_j) / p_j + \frac{m_j^T V_j m_j}{\phi p_j} \right].
\label{eq:ala_nlp_knownphi}
\end{align}
As a remark, in linear regression with known error variance $\phi$ where the groups $\beta_j$ are independent a posteriori ($Z^TZ$ is block-diagonal), then $p^N(y \mid \phi, \gamma)= \tilde{p}^N(y \mid \phi, \gamma)$, i.e. Expression~\eqref{eq:ala_nlp_knownphi} is exact.
In contrast, Laplace approximations $\hat{p}^N(y \mid \phi, \gamma)$ are not consistent even in this simplest setting.
See Section \ref{ssec:alavsla_nlps} for examples.

Lemma \ref{lem:mean_pmompenalty} is also useful in Gaussian regression with unknown error variance $\phi$. 
Suppose that one sets the prior $\phi \sim \mbox{IG}(\phi; a', b')$, then
\begin{align}
\tilde{p}^N(y \mid \phi,\gamma)=
   \pi(y \mid \phi,\gamma)
\left[ \prod_{\gamma_j=1} \mbox{tr}(V_j S_j) / p_j +
m_j^T V_j m_j E(1/\phi \mid y,\gamma) / p_j \right],
\nonumber
\end{align}
where $E(1/\phi \mid y, \gamma)= (a' + n) / (b' + y^T y - \tbeta_\gamma^T Z_\gamma^T Z_\gamma \tbeta_\gamma)$ and $\pi(y \mid \phi,\gamma)$ has closed-form expression.
Finally, for non-Gaussian regression case and unknown $\phi$ we propose
\begin{align}
\tilde{p}^N(y \mid \gamma)=
   \tilde{\pi}(y \mid \gamma)
\left[ \prod_{\gamma_j=1} \mbox{tr}(V_j S_j)/p_j + \frac{m_j^T V_j m_j}{\tilde{\phi}_\gamma p_j} \right].
\label{eq:ala_nlp_unknownphi}
\end{align}

\section{Theory}\label{sec:theory}

We consider a general setting where $(y,Z)$ arise from a data-generating $F^*$ that may be outside the assumed model class~\eqref{eq:glm_pdf}.
We prove that as $n$ grows the ALA-based $\tilde{p}(\gamma \mid y)$ assign probability increasing to 1 to an optimal $\tgamma^*$.
For the particular choice $\eta_{\gamma 0}=(0,\phi_0)$, such $\tgamma^*$ is the smallest model minimizing a mean squared loss associated to linear projections.
We first explain that exact $p(\gamma \mid y)$ asymptotically select a $\gamma^*$, in general different from $\tgamma^*$, defined by a log-likelihood loss and Kullback-Leibler (KL) projections to $F^*$.
We then provide Theorem \ref{thm:sparse_linproj}, characterizing certain situations where $\gamma^*$ coincides with $\tgamma^*$.
Subsequently, in Section \ref{ssec:bf_finitep} we consider fixed $p$ settings where one may characterize 
$\tilde{p}(\tgamma^* \mid y)$
via the rate at which pairwise Bayes factors $\tilde{B}_{\gamma, \tgamma^*}= \tilde{p}(y \mid \gamma) / \tilde{p}(y \mid \tgamma^*)$ converge to 0 in probability.
Section \ref{ssec:bf_highdim} considers high-dimensional settings where $p$ may grow with $n$. 

For any model $\gamma$, denote the KL-optimal $\eta_\gamma$ under $F^*$ by
\begin{align}
\eta_\gamma^*= \arg\max_{\eta_\gamma}
E_{F^*}[\log p(y \mid \eta_\gamma, \gamma)],
\label{eq:optimalbeta_kl}
\end{align}
and by $\eta^*$ that under the full model $p(y \mid \eta)$ including all parameters.
Multiple models may attain the global maximum
$E_{F^*}(\log p(y \mid \eta_\gamma^*, \gamma)) = E_{F^*}(\log p(y \mid \eta^*))$,
and we define the optimal $\gamma^*$ as that with smallest dimension.
If the model is well-specified, that is $F^*$ is truly contained in~\eqref{eq:glm_pdf} for some $\eta^*$, then $\gamma^*$ drops any parameters such that $\eta_j^*=0$.
Under misspecification, then $\gamma^*$ is such that adding any parameter to $\gamma^*$ cannot improve the fit,
as measured by the expected log-likelihood in \eqref{eq:optimalbeta_kl}.

In contrast, the ALA optimal model $\tgamma^*$ is based on mean squared error or, equivalently, on linear projections of $E_{F^*}(\tilde{y} \mid Z)$, where $\tilde{y}= (y - b'(0) \mathbbm{1})/b''(0)$ as in \eqref{eq:logl_gradhess0}. Let
\begin{align}
 \tbeta_\gamma^*&=   \arg\min_{\beta_\gamma} E_{F^*} \|\tilde{y} - Z_\gamma \beta_\gamma\|_2^2=
  [E_{F^*}(Z_\gamma^T Z_\gamma)]^{-1} E_{F^*}[ Z_\gamma^T \tilde{y} ]
\label{eq:optimalbeta_ala}
\end{align}
be the parameters giving the linear projection of $E_{F^*}(\tilde{y} \mid Z)$ on $Z_\gamma$, where we assume $E_{F^*}(Z_\gamma^T Z_\gamma)$ to be a finite positive-definite matrix for all $\gamma \in \Gamma$ (see Condition (C1) below).
Then $\tgamma^*$ is the smallest model minimizing mean squared error,
that is $\tgamma^*= \arg\min_{\gamma \in \widetilde{\Gamma}^*} p_{\gamma}$ where
$$\widetilde{\Gamma}^*= \{\gamma: E_{F^*} \|\tilde{y} - Z_\gamma \tbeta^*_\gamma\|_2^2= \min_{\beta} E_{F^*} \|\tilde{y} - Z \beta\|_2^2 \}$$
For simplicity we assume $\tgamma^*$ to be unique, but our results generalize when there are multiple such models by defining $\tgamma^*$ to be their union. 

It is important to note that when~\eqref{eq:glm_pdf} is misspecified neither $\gamma^*$ nor $\tgamma^*$ recover the truth, but the simplest model according to their implicit loss functions. 
That said, Theorem \ref{thm:sparse_linproj} below delineates an interesting robustness property of the linear projection $\tbeta^* = [E_{F^*}(Z^T Z)]^{-1} E_{F^*}(Z^T \tilde{y} )$, under which terms that do not affect $E_{F^*}(y \mid Z)$ are discarded by $\tgamma^*$.
\begin{thm}\label{thm:sparse_linproj}
Suppose $(y_i, z_i) \stackrel{\text{i.i.d.}} \sim F^*$ for $i = 1, \ldots, n$, where $F^*$ is a probability distribution on $\mathbb{R} \otimes \mathbb{R}^p$ with a finite positive-definite covariance matrix. Also, assume $E_{F^*}(z_i) = 0$, i.e., the covariate distribution is centered. Let $\delta \subseteq \{1, \ldots, p\}$ denote the true regression model so that $E_{F^*}(y_i \mid z_i) = f(z_{i\delta})$ almost-everywhere $F^*$, where $f: \mathbb{R}^{p_\delta} \to \mathbb{R}$ is a measurable function. Letting $\upsilon = \{1, \ldots, p\} \backslash \delta$ indicate the truly inactive parameters, assume that 
\begin{align}\label{eq:lin_exp}
E_{F^*}(z_{i \upsilon} \mid z_{i \delta}) = A z_{i \delta} \ \text{almost-everywhere} \ F^*,
\end{align}
where $A \in \mathbb{R}^{(p-p_\delta) \times p_\delta}$. Then, $\tbeta^* = (\tbeta^*_\delta; 0)$, where recall that $\tbeta^*_\delta = [E_{F^*}(Z_\delta^T Z_\delta)]^{-1} E_{F^*}[ Z_\delta^T \tilde{y} ]$. In particular, $\tbeta_j^* = 0$ whenever $j \notin \delta$. 
\end{thm}

The assumption in \eqref{eq:lin_exp} states that the conditional mean of truly spurious variables is linear in the truly active $z_{i \delta}$.
For example, the assumption is satisfied when the components of $z_i$ are independent, or when their marginal distribution under $F^*$ follows a (centered) elliptical distribution, 
such as the multivariate Gaussian or T family.
Since we assumed that $z_i$ has a finite positive-definite covariance, standard elliptical results (see, e.g., Chapter 1 of \cite{muirhead2009aspects}) give that $E_{F^*}(z_{i \upsilon} \mid z_{i \delta})$ is linear in $z_{i \delta}$. 

Under the assumptions of Theorem \ref{thm:sparse_linproj}, the ALA-optimal model $\tgamma^*$ is contained in the true model $\delta$. 
In fact, $\tgamma^* = \delta$ whenever all the entries of $\tbeta^*_\delta$ are non-zero. 
More generally, since $\tgamma^*$ is defined by zeroes in $\tbeta^*$, 
we have that $\tgamma^*$ includes any term $j$ conditionally uncorrelated with the true regression function $f(z_{i \delta})$,
that is satisfying $\mathrm{Cov}_{F^*}(z_{ij}, f(z_{i \delta}) \mid z_{i \tgamma^*})=0$.
Observe that $f(z_{i \delta})$ need not be linear for the theorem to hold. For example, for a truly generalized linear model with $E_{F^*}(y_i \mid z_i) = f(\sum_{j \in \delta} \beta_j^* z_{ij})$, we have $f = h^{-1}$ from \eqref{eq:glm_mean}. More flexible models such as mixtures of generalized linear models are also permitted. For example, consider a two-component mixture with $E(y_i \mid z_i) = \pi \, f(\sum_{j \in \delta_1} \beta_{1j}^* z_{ij}) + (1-\pi) \, f(\sum_{k \in \delta_2} \beta_{2k}^* z_{ik})$, $\pi \in [0,1]$. Then, $\tgamma^*$ discards any variable outside the true model $\delta = \delta_1 \cup \delta_2$ under the theorem assumptions.
We remark that there are simple examples where the ALA asymptotic model $\tgamma^* \neq \delta$. For instance, for $E_{F^*}(y_i \mid z_i)= z_{i1}^2$ we have that truly $\delta = \{1\}$ but $\tbeta_j^*=0$,  see also the Poisson example in Section \ref{ssec:ala_importancesampling}.  In practice, however, in most of our examples we observed that the ALA-based $\tgamma^*$ largely coincides with the model-based $\gamma^*$.

We next prove that ALA asymptotically recovers $\tgamma^*$, and give the associated rates.

\subsection{Finite-dimensional problems}
\label{ssec:bf_finitep}

The assumptions to obtain ALA Bayes factor rates are minimal. For any model $\gamma \in \Gamma$ and $\tbeta_\gamma^*$ in \eqref{eq:optimalbeta_ala}, we assume the following conditions.
\begin{enumerate}[leftmargin=*,label=(C\arabic*)]
\item $(y_i,z_i) \sim F^*$ independently for $i=1,\ldots,n$, with finite positive-definite $\Sigma_{z \gamma}= \mbox{Cov}_{F^*}(z_{i \gamma})$
and finite 
$\Sigma_{y | z, \gamma}= \mbox{diag}(\mbox{Cov}_{F^*}(y_1 \mid Z_\gamma), \ldots, \mbox{Cov}_{F^*}(y_n \mid Z_\gamma))$.

\item The matrix $L_\gamma=E_{F^*} (z_{i \gamma} z_{i \gamma}^T [E_{F^*}(y_i \mid z_{i \gamma}) - z_{i \gamma}^T \beta_\gamma^*]^2)$ has finite entries.

\item The prior density $p^L(\beta_\gamma \mid \phi, \gamma)$ is continuous and strictly positive at $\tbeta_\gamma^*$.

\item The equations $\phi^2 \sum_{i=1}^n c(y_i,\phi)/n= -b(0)$
and $\phi^2 E_{F^*}[ \nabla_\phi c(y_i,\phi)]=-b(0)$ have unique roots $\phi_0$ and $\phi_0^*$ respectively,
where $E_{F^*}[ \nabla_\phi c(y_i,\phi)] < \infty$.
\end{enumerate}

Conditions (C1)-(C2) require that $(y_i,z_i)$ have finite full-rank second-order moments.
Condition (C3) states that $p^L(\beta_\gamma \mid \phi, \gamma)$ is a local prior assigning positive density to $\tbeta_\gamma^*$, Theorem \ref{thm:bf_ala} and Corollary \ref{cor:bf_ala_unknownphi} below give Bayes factor rates for such prior and also for the non-local $p^N(\beta_\gamma \mid \phi,\gamma)$ in \eqref{eq:prior_parameters}.
Condition (C4) states that $\phi_0$ and $\phi_0^*$ are the unique maximizers
of the observed and expected log-likelihood under $F^*$, conditional on $\beta=0$.
(C4) is only used in Corollary \ref{cor:bf_ala_unknownphi} where $\phi$ is unknown to show that $\phi_0 \stackrel{P}{\longrightarrow} \phi_0^*$,
and can be easily replaced, should $(\phi,\phi_0^*)$ not be unique.
One may instead assume that $ \log p(y \mid \beta=0, \phi)$ defines a Glivenko-Cantelli class,
a sufficient condition being that the log-likelihood is dominated by an integrable function under $F^*$
(\cite{vandervaart:1998}, Theorems 5.7, 5.9 and Lemma 5.10).

Theorem~\ref{thm:bf_ala} states that for any model $\gamma$ adding spurious parameters (in the linear projection sense) to $\tgamma^*$,
then $\tilde{B}^L_{\gamma \tgamma^*}= O_p(n^{-(p_\gamma - p_{\tgamma^*})/2})$, for any local prior $p^L(\beta \mid \gamma)$.
Hence ALA Bayes factors discard spurious parameters at a polynomial rate in $n$.
The rates for the gMOM-based $\tilde{B}^N_{\gamma \tgamma^*}$ are faster, akin to results on exact Bayes factors \citep{johnson:2010,johnson:2012}.
In contrast, if $\gamma$ misses active parameters, then $\tilde{B}_{\gamma \tgamma^*}^L$ and $\tilde{B}_{\gamma \tgamma^*}^N$ decrease exponentially in $n$.
Corollary \ref{cor:bf_ala_unknownphi} extends Theorem \ref{thm:bf_ala} to the unknown $\phi$ case.

\begin{thm}
  Assume Conditions (C1)-(C3). Let
$\tilde{B}^L_{\gamma, \tgamma^*}$ and
$\tilde{B}^N_{\gamma, \tgamma^*}= \tilde{p}^N(y \mid \phi, \gamma)/\tilde{p}^N(y \mid \phi, \tgamma^*)$
be the ALA Bayes factor in \eqref{eq:bfala_expfamily_knownphi} and \eqref{eq:ala_nlp_knownphi}
when $\phi$ is known.
  \begin{enumerate}[leftmargin=*,label=(\roman*)]
    \item Suppose that $\tgamma^* \subset \gamma$. Then
      $\tilde{B}^L_{\gamma \tgamma^*}= n^{(p_{\tgamma^*} - p_{\gamma})/2} O_p(1)$
and $\tilde{B}^N_{\gamma \tgamma^*}= n^{3(p_{\tgamma^*} - p_{\gamma})/2} O_p(1)$
as $n \rightarrow \infty$.

    \item Suppose that $\tgamma^* \not\subseteq \gamma$. Then
      $\frac{1}{n} \log \tilde{B}^L_{\gamma \tgamma^*} < o_p(1) + W/n$
      and $\frac{1}{n} \log \tilde{B}^N_{\gamma \tgamma^*} < o_p(1) + W/n$
      for a random variable $W$ satisfying $W/n \stackrel{P}{\longrightarrow} c>0$, as $n \rightarrow \infty$.
\end{enumerate}
\label{thm:bf_ala}
\end{thm}

\begin{cor}
  Assume Conditions (C1)-(C4).
Let $\tilde{B}^L_{\gamma, \tgamma^*}$ and $\tilde{B}^N_{\gamma, \tgamma^*}$ be the ALA Bayes factor corresponding to the local and non-local priors in \eqref{eq:bfala_expfamily_unknownphi} and \eqref{eq:ala_nlp_unknownphi} when $\phi$ is unknown.
Then the statements in Theorem \ref{thm:bf_ala} (i)-(ii) remain valid.
\label{cor:bf_ala_unknownphi}
\end{cor}

The rates in Theorem \ref{thm:bf_ala} are of the same form, as a function of $n$, as those for standard \citep{dawid:1999,johnson:2010} and miss-specified Bayes factors \citep{rossell:2018b,rossell:2019t}.
The main difference with such standard rates is in Part (ii). The leading term in $\tilde{B}_{\gamma,\tgamma^*}$ is given by a random variable $W$ that converges to a chi-square distribution with non-centrality parameter
$l(\tgamma^*,\gamma)= n (\tilde{b}^*)^T \tilde{S} \tilde{b}^* > 0$, where $\tilde{S}$ is a positive-definite matrix and $\tilde{b}^* \neq 0$ are asymptotic partial regression coefficients for columns in $\tgamma^* \setminus \gamma$
(see the proof for details).
In contrast, the leading term in exact $B_{\gamma,\gamma^*}$ 
has a different non-centrality parameter $n (b^*)^T S b^*$, where $(b^*,S)$ now depend on KL projections.
That is, although $p(\gamma \mid y)$ and $\tilde{p}(\gamma \mid y)$ are both exponentially fast in $n$ at discarding models $\gamma$ that miss truly active parameters in $\gamma^*$ and $\tgamma^*$ respectively, the coefficients governing these rates may change.
For instance, if the exponential family model \eqref{eq:glm_pdf} is well-specified, even when $\tgamma^*=\gamma^*$ one expects 
$\tilde{B}_{\gamma, \tgamma^*}$ to have lower statistical power than $B_{\gamma \gamma^*}$ to detect active parameters.
In contrast, if \eqref{eq:glm_pdf} is misspecified and $E_{F^*}(y_i \mid z_i)$ is better approximated by a linear function of $z_i$ than by \eqref{eq:glm_mean}, then one expects $\tilde{B}_{\gamma, \tgamma^*}$ to attain higher asymptotic power.

From a technical point of view a contribution of Theorem \ref{thm:bf_ala} relative to earlier results is that, by building upon parameter estimation results for concave log-likelihoods in~\cite{hjort:2011}, it requires near-minimal technical conditions.

\subsection{High-dimensional problems}
\label{ssec:bf_highdim}

Our main result proves that $\tilde{p}(\tgamma^* \mid y) \stackrel{L_1}{\longrightarrow} 1$ as $n \rightarrow \infty$ and provides the associated convergence rates. 
Recall that $\tgamma^*$ is the ALA-optimal model, where $\tgamma_j=\mbox{I}(\tbeta_j^* \neq 0)$ indicates zeroes in the ALA-optimal parameter $\tbeta^*$ in \eqref{eq:optimalbeta_ala}.
By definition of $L_1$ convergence, this is equivalent to $E_{F^*} \sum_{\gamma \neq \tgamma^*} \tilde{p}(\gamma \mid y)$ converging to 0.
$L_1$ convergence guarantees the asymptotic control of certain frequentist model selection probabilities.
Let $\hat{\gamma}=\arg\max \tilde{p}(y \mid \gamma)$ be the highest posterior probability model, 
then 
$P_{F^*}(\hat{\gamma} \neq \tgamma^*) \leq 2( E_{F^*}[ \tilde{p}(\tgamma^* \mid y)] -1)$ (\cite{rossell:2018}, Proposition 1).
The same bound applies to the family-wise type I-II error probabilities, and when setting  $\hat{\gamma}$ to be the median probability model of \cite{barbieri:2004} (\cite{rossell:2018}, Corollary 1). 

Our main assumption is that $\tilde{y}$ has sub-Gaussian tails with variance parameter $\sigma^2$ under $F^*$, see Definition \ref{def:subgaussian} in the supplement.
We use the assumption to derive novel bounds on integrated tail probabilities of sub-Gaussian quadratic forms, see Propositions \ref{prop:tailintegral_quadform_subgaussian} and \ref{prop:tailintegral_difquadform_subgaussian}, which may have some independent interest. The assumption is satisfied for example if $F^*$ has Gaussian tails or $\tilde{y}$ is bounded as in logistic, multinomial or ordinal regression, and in fact allows for dependence in $\tilde{y}$, but is not satisfied when $\tilde{y}$ has thick tails such as the Poisson distribution.
One may extend Propositions \ref{prop:tailintegral_quadform_subgaussian} and \ref{prop:tailintegral_difquadform_subgaussian} to thicker-tailed $F^*$, then the $L_1$ rates could be slower than those presented here.
For simplicity we focus on the known dispersion parameter $\phi$ case, non-random design matrix $Z$ and Zellner's prior $p^L(\beta_\gamma \mid \phi,\gamma)= \mathcal{N}(\beta_\gamma; 0, (g_L\phi/n) (Z_\gamma^T Z_\gamma)^{-1})$. 
Our proofs can be extended to unknown $\phi$ and other priors, see the proof for a discussion, at the expense of more involved arguments and technical conditions.
By default we recommend setting constant $g_L$, say $g_L=1$ as in Section \ref{ssec:prior_elicitation}, but our results allow for $g_L$ to change with $n$. For example, \cite{narisetty:2014} proposed letting $g_L$ grow with $n$ to obtain sparser solutions, whereas proceeding analogously to the uniformly most powerful tests of \cite{johnson:2013b} one might let $g_L$ decrease with $n$ to improve power.

Theorem \ref{thm:bf_ala_highdim} below provides a first result on pairwise Bayes factors, specifically on
\begin{align}
E_{F^*} \left( \left[ 1 + \tilde{B}_{\tgamma^* \gamma} \frac{p(\tgamma^*)}{p(\gamma)} \right]^{-1} \right),
\nonumber
\end{align}
that is the posterior probability assigned to $\gamma$ if one only considered the models $\gamma$ and $\tgamma^*$. 
Bounding this quantity also bounds the rate at which $\tilde{B}_{\gamma \tgamma^*} \stackrel{P}{\longrightarrow} 0$, hence Theorem \ref{thm:bf_ala_highdim} extends Theorem \ref{thm:bf_ala} to high dimensions.
Theorem \ref{thm:pp_ala_highdim} is our main result characterizing $\tilde{p}(\tgamma^* \mid y)$.

We first interpret Theorem \ref{thm:bf_ala_highdim}, subsequently discuss the required technical conditions and finally state the theorem.
Part (i) says that models adding spurious parameters to $\tgamma^*$ are discarded at the same polynomial rate in $n$ (up to log terms) as in the fixed $p$ case,
\begin{align}
  r_\gamma= \left( n g_L \right)^{\frac{p_\gamma - p_{\tgamma^*}}{2}} \frac{p(\tgamma^*)}{p(\gamma)}.
\nonumber
\end{align}
This rate holds when the model-predicted variance $V(\tilde{y}_i \mid z_i,\beta=0, \phi)=\phi/b''(0) > \sigma^2$, that is data under $F^*$ are under-dispersed. 
Alternatively, if $\phi/b''(0) < \sigma^2$ (over-dispersion) then a (slower) rate $r_\gamma^a$ is attained, where $a=\phi/[b''(0) \sigma^2]<1$.
The intuition is that, if the model underestimates $\sigma^2$ then it becomes easier to add spurious parameters to $\tgamma^*$.
Part (ii) states that models missing active parameters are discarded at an exponential rate in the non-centrality parameter 
\begin{align}
\lambda_\gamma= (Z_{\tgamma^*} \beta_{\tgamma^*}^*)^T (I - H_\gamma) Z_{\tgamma^*} \beta_{\tgamma^*}^*
\label{eq:ncp}
\end{align}
where $H_\gamma= Z_\gamma (Z_\gamma^TZ_\gamma)^{-1} Z_\gamma^T$ is the projection matrix onto the column space of $Z_\gamma$.
For simplicity the result raises the rate at a constant power $b$ arbitrarily close to 1, but one can actually take $b=1$ and add logarithmic terms, see the proof for details.

The parameter $\lambda_\gamma$ has a simple interpretation, it is the reduction in mean squared error when one approximates $E_{F^*}(y \mid Z)$ with $Z_{\tgamma^*} \beta_{\tgamma^*}^*$, relative to $Z_\gamma \beta_\gamma^*$.
A common strategy in high-dimensional model selection theory is to assume conditions on the eigenvalues of $Z^TZ$ to lower-bound $\lambda_\gamma$ in terms of $\beta_{\tgamma^*}^T \beta_{\tgamma^*}$. Instead, here we give the result directly in terms of $\lambda_\gamma$, and state near-minimal conditions on $\lambda_\gamma$ required for the result to hold.
To build intuition, however, in a simplest case where the columns in $\tgamma^* \setminus \gamma$ are uncorrelated with those in $\gamma$,
then
$\lambda= (\beta_{\tgamma^* \setminus \gamma}^*)^T Z_{\tgamma^* \setminus \gamma}^T Z_{\tgamma^* \setminus \gamma} \beta_{\tgamma^* \setminus \gamma}^*$ and one can roughly think of $\lambda_\gamma$ as being linear in $n$.

The technical conditions required for Theorem \ref{thm:bf_ala_highdim} and any model $\gamma \in \Gamma$ are below. For two sequences $a_n,b_n$, $a_n \ll b_n$ denotes that $\lim_{n \rightarrow \infty} a_n/b_n=0$.
\begin{enumerate}[leftmargin=*,label=(D\arabic*)]
\item There exists a finite $\sigma^2>0$ such that $\tilde{y} - Z_{\tgamma^*} \tbeta_{\tgamma^*}^* \sim \mbox{SG}(0, \sigma^2)$,
where $\tbeta_{\tgamma^*}^*$ is as in \eqref{eq:optimalbeta_ala}.

\item $Z_\gamma^T Z_\gamma$ is invertible.

\item For any $\gamma \supset \tgamma^*$,
\begin{align}
\log(ng_L) + \frac{2}{p_\gamma - p_{\tgamma*}} \log \left( \frac{p(\tgamma^*)}{p(\gamma)} \right) \gg 1. 
\label{eq:cond_d3_i}
\end{align}

For any $\gamma \not\supset \tgamma^*$ of dimension $p_\gamma \geq p_{\tgamma^*}$,
\begin{align}
 (p_{\tgamma^*} - p_\gamma) \log \left( n g_L \right)
+  \log \left( \frac{p(\gamma)}{p(\tgamma^*)} \right)
 +  p_\gamma \ll \frac{\lambda_\gamma}{\log(\lambda_\gamma)}.
\label{eq:cond_d3_ii}
\end{align}

For any $\gamma \not\supset \tgamma^*$ of dimension $p_\gamma < p_{\tgamma^*}$,
\begin{align}
 (p_{\tgamma^*} - p_\gamma) \log \left( n g_L \right)
+  \log \left( \frac{p(\gamma)}{p(\tgamma^*)} \right) +  p_{\tgamma^*} 
\ll \frac{\lambda_\gamma}{\log(\lambda_\gamma)}.
\label{eq:cond_d3_iii}
\end{align}

\end{enumerate}

As discussed earlier, Condition (D1) states that the data-generating $F^*$ satisfies a tail property, specifically that the ALA-optimal errors $\tilde{y} - Z_{\tgamma}^* \tbeta_{\tgamma^*}^*$ are no thicker than sub-Gaussian.
(D2) can be relaxed but ensures that $p^L(\beta_\gamma \mid \phi,\gamma)= \mathcal{N}(\beta_\gamma; 0, (g_L \phi/n) (Z_\gamma^T Z_\gamma)^{-1})$ is proper.
(D2) requires that $p_\gamma \leq n$, as discussed in Section \ref{sec:prior} one may have $p \gg n$ but only models with up to $p_\gamma < n$ parameters receive positive prior probability $p(\gamma)>0$.
(D3) are minimal conditions on the prior parameters and the signal strength $\lambda_\gamma$.
If the number of truly active groups $|\tgamma^*| < J/2$, where $J$ is the number of total groups,
$p(|\gamma|)$ in \eqref{eq:prior_models} is non-increasing in $|\gamma|$,
and $g_L$ is non-decreasing in $n$, then \eqref{eq:cond_d3_i} and \eqref{eq:cond_d3_ii} hold.
(D3) states weaker, near-necessary assumptions for pairwise $B_{\gamma \gamma^*}$ convergence.
Also, for $p(\gamma)$ in  \eqref{eq:prior_models}, a sufficient condition for \eqref{eq:cond_d3_iii} is that
$$
p_{\tgamma^*} \log (g_L n) + |\tgamma^*| \log J  \ll \frac{\lambda_\gamma}{\log \lambda_\gamma},
$$
that is (D3) allows the number of parameters $p_{\tgamma^*}$ (and groups $|\tgamma^*|$) in $\tgamma^*$ to grow near-linearly in $\lambda_\gamma$, and the total number of groups $J$ to grow near-exponentially.

\begin{thm}
Let $p^L(\beta_\gamma \mid \phi,\gamma)= N(\beta_\gamma; 0, (g_L \phi / n) (Z_\gamma^T Z_\gamma)^{-1})$, where $\phi>0$ is fixed,
and assume Conditions D1-D3.

\begin{enumerate}[leftmargin=*,label=(\roman*)]
\item Suppose that $\tgamma^* \subset \gamma$. If $\phi/b''(0) > \sigma^2$, then
\begin{align}
E_{F^*} \left( \left[ 1 + \tilde{B}^L_{\tgamma^* \gamma} \frac{p(\tgamma^*)}{p(\gamma)} \right]^{-1} \right) \leq  \frac{2 \max \left\{ [(2/[p_\gamma - p_{\tgamma^*}]) \log r_\gamma]^{(p_\gamma - p_{\tgamma^*})/2}, \log(r_\gamma) \right\}}{r_\gamma},
\nonumber
\end{align}
for all $n \geq n_0$, and some fixed $n_0$.
Further, if $\phi/b''(0) \leq \sigma^2$, then 
\begin{align}
E_{F^*} \left( \left[ 1 + \tilde{B}^L_{\tgamma^* \gamma} \frac{p(\tgamma^*)}{p(\gamma)} \right]^{-1} \right) 
\leq  \frac{2.5 \max \left\{ \left[(2/[p_\gamma - p_{\tgamma^*}]) \log (r_\gamma^{\frac{a}{1+\epsilon}})\right]^{(p_\gamma - p_{\tgamma^*})/2}, \log\left(r_\gamma^{\frac{a}{1+\epsilon}}\right) \right\}}{r_\gamma^{\frac{a}{(1 + \epsilon \sqrt{2}(1+\sqrt{2}))}}}
\nonumber
\end{align}
for all $n \geq n_0$, $\epsilon= 1/\sqrt{a \log r_\gamma}$, $a= \phi/[b''(0) \sigma^2]$.

\item Suppose that $\tgamma^* \not\subset \gamma$.
For any $b<1$ there exists a finite $n_0$ such that, for all $n \geq n_0$,
\begin{align}
E_{F^*} \left( \left[ 1 + \tilde{B}^L_{\tgamma^* \gamma} \frac{p(\tgamma^*)}{p(\gamma)} \right]^{-1} \right) \leq \left( \frac{p(\gamma) e^{-\frac{\lambda_\gamma}{2\max\{\phi/b''(0), \sigma^2\} \log(\lambda_\gamma)}}}{p(\tgamma^*) (n g_L)^{\frac{(p_\gamma - p_{\tgamma^*})}{2}}} \right)^b
\nonumber
\end{align}


\end{enumerate}
\label{thm:bf_ala_highdim}
\end{thm}

Theorem \ref{thm:pp_ala_highdim} below states that the posterior probability $\tilde{p}(\tgamma^* \mid y)$ converges to 1. 
Separate rates are given for the set of models
$S=\{\gamma: \tgamma^* \subset \gamma\}$ containing spurious parameters and its complement $S^c=\{\gamma: \tgamma^* \not\subset \gamma\}$.
The result assumes the model prior in \eqref{eq:prior_models}, with $p(|\gamma|) \propto p^{-c |\gamma|}$, for $c \geq 0$.
Recall that in the canonical case with no groups nor hierarchical constraints, $c=0$ reduces to the Beta-Binomial(1,1) prior and $c>0$ to a Complexity prior.
Theorem \ref{thm:pp_ala_highdim} requires a technical condition.
\begin{enumerate}[leftmargin=*,label=(D\arabic*)]
\setcounter{enumi}{3}
\item Let $a=1$ if $\phi/b''(0) > \sigma^2$ and $a<\phi/[b''(0)\sigma^2]$ otherwise,
$q= \min\{p_j: \gamma_j^*=0 \}$ be the smallest spurious group size and $q'= \max \{p_j\}$.
Then, assume that
\begin{align}
 (|\tgamma^*|+1) (\bar{J} - |\tgamma^*|) \ll  (n g_L)^{aq/2} p^{ca + a -1}.
\label{eq:cond_d4_i}
\end{align}
Further, let $\underline{\lambda}= \min_{|\gamma| \leq |\tgamma^*|} \lambda_\gamma/\max\{|\tgamma^*| - |\gamma|,1\}$
and $\bar{\lambda}= \min_{|\gamma|>|\tgamma^*|, \gamma \not\subset \tgamma^*} \lambda_\gamma$. Then
\begin{align}
 (|\tgamma^*|+2) \log J \ll \bar{\lambda} + c \log p + \frac{q}{2} \log(n g_L)
\label{eq:cond_d4_ii}
\end{align}
\begin{align}
 \frac{q'}{2} \log(n g_L) + c \log p + (|\tgamma^*| +1) \log J \ll  \underline{\lambda}
\label{eq:cond_d4_iii}
\end{align}
\end{enumerate}

Condition (D4) ensures that (D3) holds uniformly across models $\gamma \in \Gamma$.
It is stated in terms of $(\underline{\lambda}, \bar{\lambda})$ bounding the non-centrality parameter $\lambda_\gamma$ for models of smaller and larger size than $\tgamma^*$ groups, respectively.
Expressions \eqref{eq:cond_d4_i}-\eqref{eq:cond_d4_ii} impose mild assumptions on the number of active groups $|\tgamma^*|$, total groups $J$ and largest number of allowed groups $\bar{J}$. These expressions guarantee that models of size $|\gamma|>|\tgamma^*|$ receive vanishing probability, and in particular they are satisfied when setting $c>0$.
Expression \eqref{eq:cond_d4_iii} ensures that models of size $|\gamma| \leq |\tgamma^*|$ receive vanishing probability, and imposes an upper-bound on $c$ in terms of the signal strength $\underline{\lambda}$.

\begin{thm}
Let $p^L(\beta_\gamma \mid \phi,\gamma)= N(\beta_\gamma; 0, (g_L \phi / n) (Z_\gamma^T Z_\gamma)^{-1})$,
for known $\phi>0$, and $p(\gamma)$ as in \eqref{eq:prior_models} with $p(|\gamma|)= p^{-c|\gamma|}$, $c \geq 0$.

\begin{enumerate}[leftmargin=*,label=(\roman*)]
\item Suppose that Conditions D1, D2 and D4 hold, and that $\lim_{n \rightarrow \infty} |\tgamma^*| / [p^c (n g_L)^{q/2}]=0$.
Then there is a finite $n_0$ such that, for all $n \geq n_0$ and all $\epsilon>0$,
if $\phi/b''(0) < \sigma^2$ then
\begin{align}
 E_{F_0} \left( \tilde{p}^L(S \mid y) \right) \leq
\frac{(|\tgamma^*| + 1) (\bar{J}-|\tgamma^*|) (\log(n g_L^{q/2} ) + [c+1]\log(p) + \epsilon) }{p^c (n g_L)^{q/2}}.
\nonumber
\end{align}
Further, if $\phi/b''(0) < \sigma^2$, then
\begin{align}
 E_{F_0} \left( \tilde{p}^L(S \mid y) \right) \leq
\frac{(|\tgamma^*|+1)(\bar{J} - |\tgamma^*|) [ \log \left((n g_L)^{q/2} \right) + (c+1)\log(p) + \epsilon]}{p^{a(c+1)-1} (ng_L)^{aq/2}},
\nonumber
\end{align}
where $a<1$ is a constant smaller than but arbitrarily close to $\phi/[b''(0)\sigma^2]$.

\item Suppose that (D1), (D2) and (D4) hold. Then there is a finite $n_0$ such that, for all $n \geq n_0$, the posterior probability assigned to $S^c=\{\gamma: \tgamma^* \not\subseteq \gamma\}$ satisfies
\begin{align}
 E_{F_0} \left( \tilde{p}^L(S^c \mid y) \right) \leq
 \frac{(|\tgamma^*|+1) e^{(|\tgamma^*| + 2)\log J}}{[e^{\bar{\lambda}} p^c (ng)^{q/2}]^b}
+ \frac{e^{|\tgamma^*| \log J}}{e^{\underline{\lambda}}} + \frac{e^{|\tgamma^*| \log J}}{e^{(|\tgamma^*| + 1) [\underline{\lambda} - c\log(p) - (q'/2)\log(ng)]}}.
\nonumber
\end{align}

\end{enumerate}

\label{thm:pp_ala_highdim}
\end{thm}

The three terms in the bound for $ E_{F_0} \left( \tilde{p}^L(S^c \mid y) \right)$ correspond to the posterior probability assigned models with $|\gamma|>|\tgamma^*|$, $|\gamma|=|\tgamma^*|$ and $|\gamma|<|\tgamma^*|$, respectively.
That is, 
Theorem \ref{thm:pp_ala_highdim} does not only characterize the posterior probability $\tilde{p}(\tgamma^* \mid y)$ assigned to the ALA-optimal model $\tgamma^*$, but also to other interesting model space subsets: those adding spurious parameters to $\tgamma^*$, and those missing parameters with either smaller/larger size than $|\tgamma^*|$.
These rates reflect that sparse priors, for example Complexity priors ($c>0$) or diffuse priors ($g_L$ grows with $n$), are faster at discarding models larger than $|\tgamma^*|$. The trade-off is that they attain slower rates for models of size $|\gamma|<|\tgamma^*|$, that is they have lower statistical power to discard small models missing truly active parameters.

\section{Results}\label{sec:results}

We illustrate the performance of ALA in terms of numerical accuracy, computation time and quality of the model selection inference in simulated and empirical data.
Section \ref{ssec:binary_example} shows the computational scalability in logistic and Poisson regression, as either $n$ or $p$ grow.
 We consider the default ALA at $\beta_0=0$ and refined versions where $\beta_{\gamma 0}$ is obtained by 1 and 2 Newton-Raphson iterations starting at zero, respectively. We also consider the combined use of ALA with importance sampling to obtain samples from the exact posterior. 
Section \ref{ssec:sim_gaussian} studies the accuracy of ALA under the non-local gMOM prior in \eqref{eq:ala_nlp}.
For this prior, ALA are faster and often more precise than LA, particularly for models that include spurious covariates. 
We also compare the gMOM ALA model selection performance in a non-linear regression example to exact gZellner calculations, the group LASSO \citep{bakin:1999}, group SCAD \citep{fan:2001} and group MCP \citep{zhang:2010}.
In Section \ref{ssec:salary_data} we analyze a poverty dataset that has a binary outcome and large $n$ and $p$. 
Finally, Section \ref{ssec:survival_example} shows survival examples where the likelihood lies outside the exponential family, but is nevertheless log-concave, making it amenable to the ALA.

In all examples we used the gMOM and gZellner priors with default $g_N=1$ and $g_L=1$ parameters and the Beta-Binomial prior on models, truncated to models satisfying the hierarchical constraints when required (Section \ref{sec:prior}).
To ensure that run times between LA and ALA are comparable, we implemented both in C++ in R package \texttt{mombf}.
For problems with $>10^6$ models where full enumeration was unfeasible, we used the augmented Gibbs sampling algorithm from \cite{rossell:2019t}. 
 We used the software defaults producing 10,000 full Gibbs scans and no parallel computing, hence our run times are a conservative figure relative to those potentially attainable with more advanced model search strategies.  
Packages \texttt{grplasso} and \texttt{grpreg} \citep{breheny:2015} were used to implement group LASSO, group SCAD and group MCP. 


\subsection{Simulations for logistic and Poisson regression}
\label{ssec:binary_example}

\begin{figure}
\begin{center}
\begin{tabular}{cc}
$p=10$, $n \in \{100,500,1000,5000\}$ &
$n=500$, $p \in \{5,10,25,50\}$ \\
\includegraphics[width=0.5\textwidth]{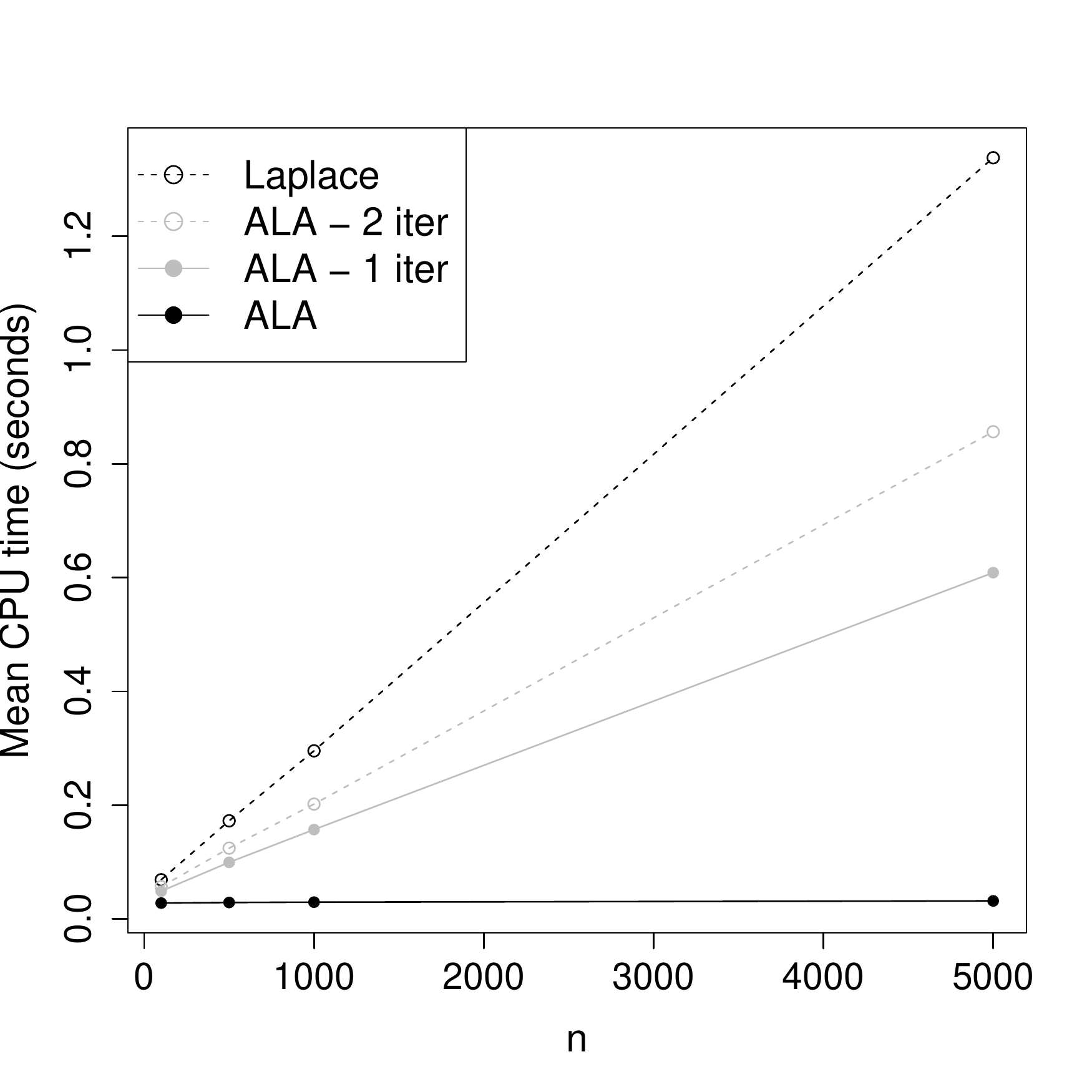} &
\includegraphics[width=0.5\textwidth]{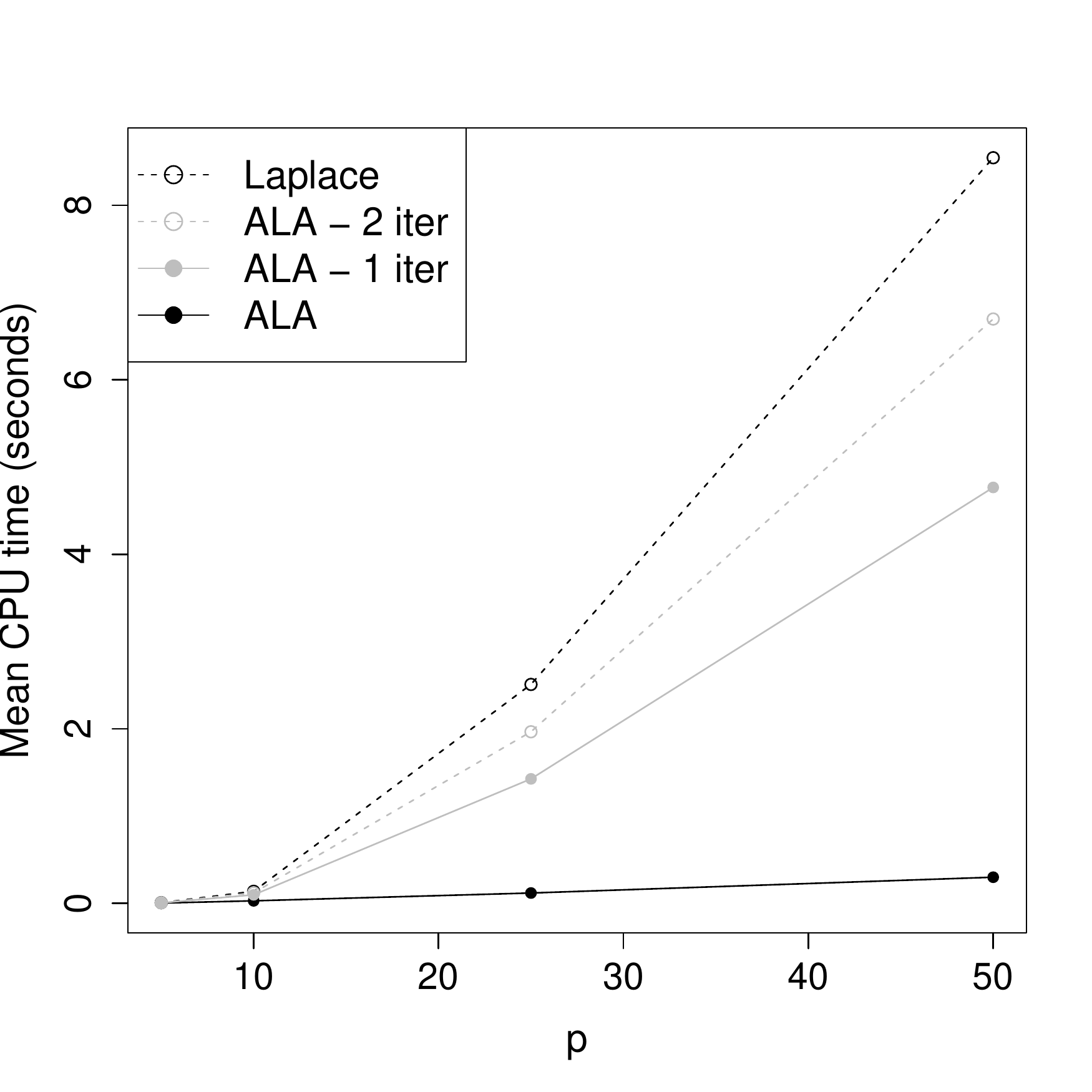} \\
\includegraphics[width=0.5\textwidth]{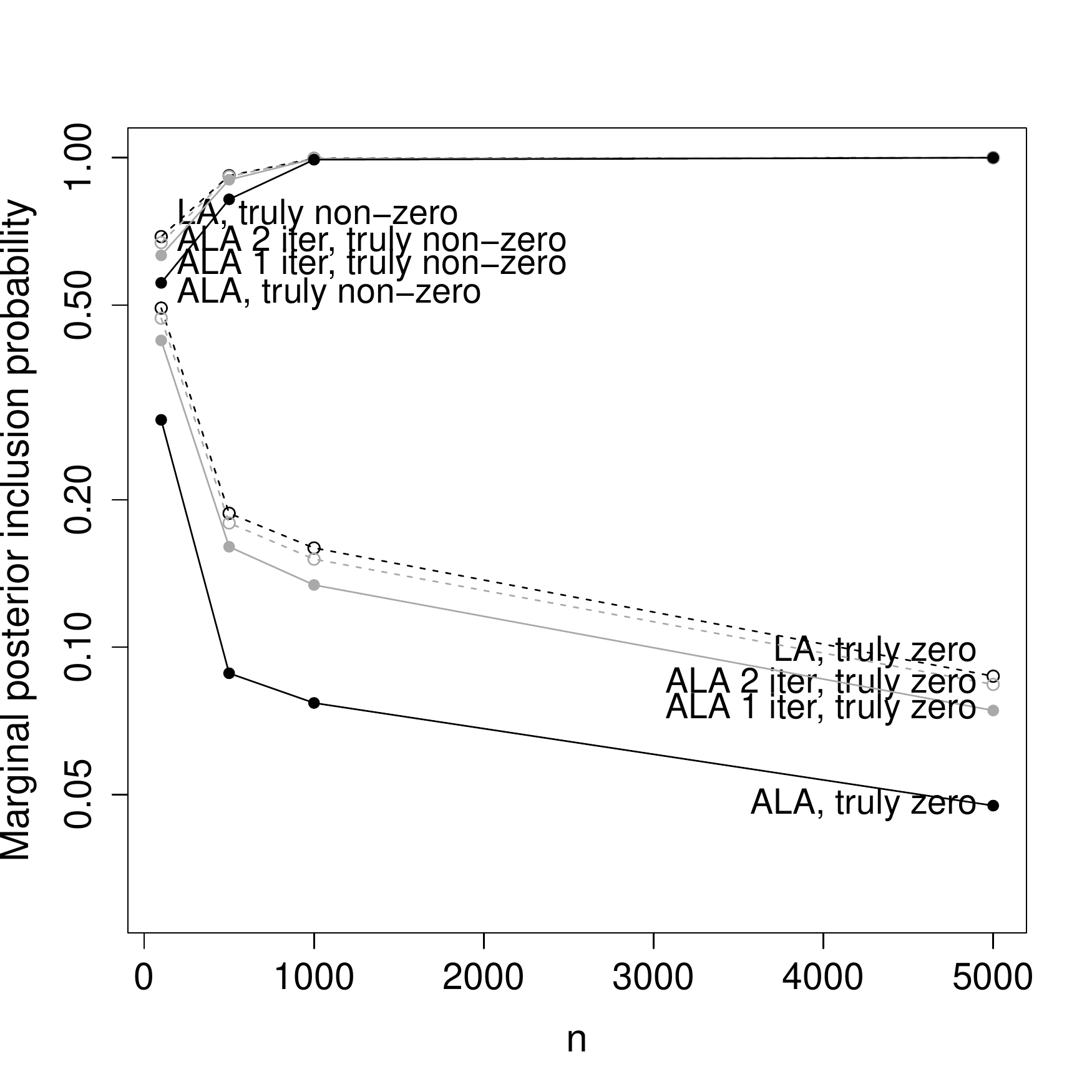} & 
\includegraphics[width=0.5\textwidth]{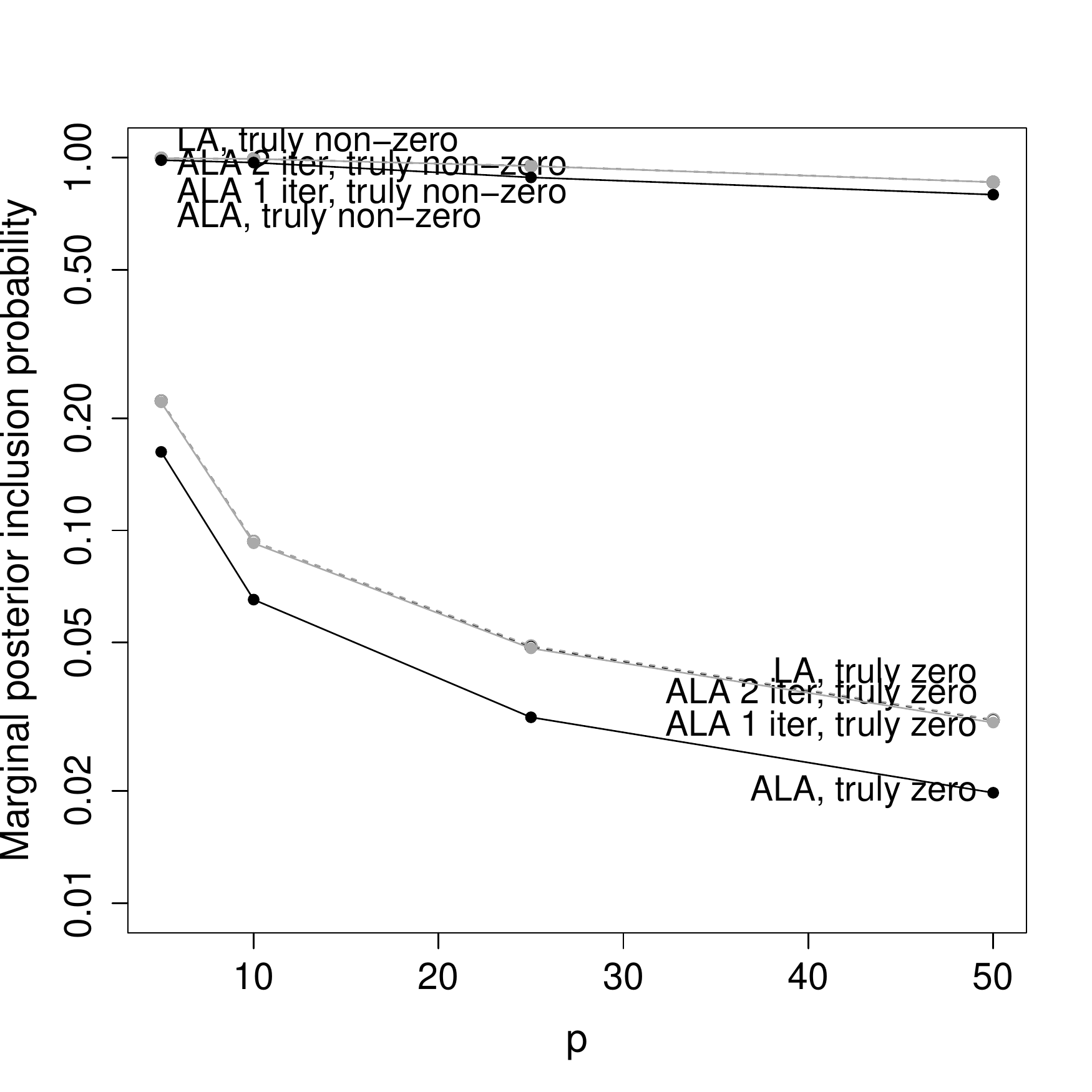} \\
\end{tabular}
\end{center}
\caption{Logistic regression simulation. 
Top: average run time (seconds) in a single-core i7 processor running Ubuntu 20.04.
Bottom: average posterior inclusion probabilities for truly active and inactive variables
}
\label{fig:logistic_sims}
\end{figure}

We considered simulated examples where the data are truly generated from logistic and Poisson models with linear predictor
$\beta^*=(0,\ldots,0,0.5,1)$ and $z_i \sim \mathcal{N}(0,\Sigma)$, $\Sigma_{ii}=1$, $\Sigma_{ij}=0.5$,
and we set the group Zellner prior in \eqref{eq:prior_parameters}. 
We consider a first setting with fixed $p=10$ and $n \in \{100,500,1000,5000\}$, and a second setting with fixed $n=500$ and $p \in \{5,10,25,50\}$.

Figure \ref{fig:logistic_sims} summarizes the logistic regression results, and Figure \ref{fig:simpoisson} those for Poisson regression.
ALA at $\beta_0=0$ significantly reduced run times for larger $n$ or $p$.
In terms of the resulting inference, ALA and LA attained consistency as $n$ grows (bottom left) and discriminateed truly active versus inactive variables, even for larger $p$ (bottom right). 
Figure \ref{fig:simpoisson} also shows that ALA Bayes factors that do not incorporate the over-dispersion curvature adjustment in \eqref{eq:bfala_adj_expfamily} led to assigning significantly higher inclusion probabilities to truly spurious parameters, even for fairly large $n$.

We also applied ALA with $\beta_{\gamma 0}$ given by 1 and 2 Newton-Raphson iterations from zero (Figure \ref{fig:logistic_sims}).
This refinement gave a closer approximation to the LA posterior, at a non-negligible computational cost, particularly in scalability as $n$ grows.

Finally, we studied the use of ALA as a tool to identify promising models that can be subsequently refined with exact methods. Specifically, we used importance sampling to re-weight models sampled from the ALA posterior. Section \ref{ssec:ala_importancesampling} offers a full description. Briefly, from our theory,  the ALA and LA posteriors in general concentrate on two different models $\tilde{\gamma}^*$ and $\gamma^*$ respectively, resulting in degenerate importance weights as $n \rightarrow \infty$.
However, in practice for finite $n$ there are situations where importance sampling is effective; see the logistic regression example in Section \ref{ssec:ala_importancesampling}. In other cases, such as the Poisson example in Section \ref{ssec:ala_importancesampling}, ALA is useful to screen out certain truly inactive covariates, but cannot be directly combined with importance sampling.
Recall that Theorem \ref{thm:sparse_linproj} gives conditions where ALA asymptotically recovers or screens out the correct parameters. More generally, combining ALA with exact strategies is an interesting avenue for future research that deserves a separate treatment elsewhere.


\subsection{Simulations under Gaussian outcomes}
\label{ssec:sim_gaussian}

Even for Gaussian outcomes the marginal likelihood under the gMOM prior requires a number of operations growing exponentially with model size \citep{kan:2008}. 
Section \ref{ssec:alavsla_nlps} illustrates that ALA not only provides faster integrated likelihoods than the LA, but that it can also be more precise. Section \ref{ssec:linreg} compares ALA-based gMOM  model selection versus exact calculations under the group-Zellner prior and three penalized likelihood methods.

\subsubsection{Numerical accuracy under non-local priors}\label{ssec:alavsla_nlps}

\begin{figure}
\begin{center}
\begin{tabular}{cc}
\includegraphics[width=0.5\textwidth]{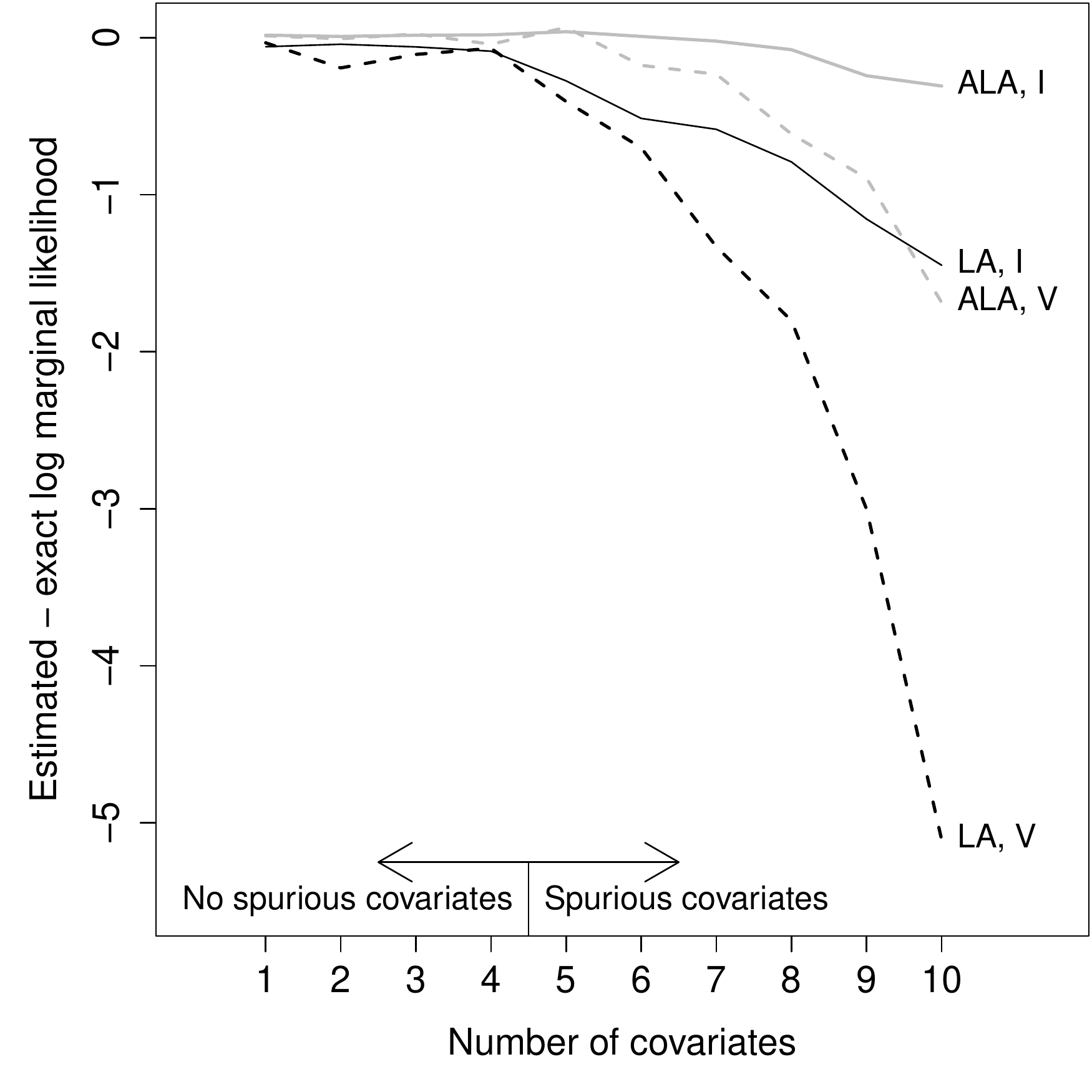} &
\includegraphics[width=0.5\textwidth]{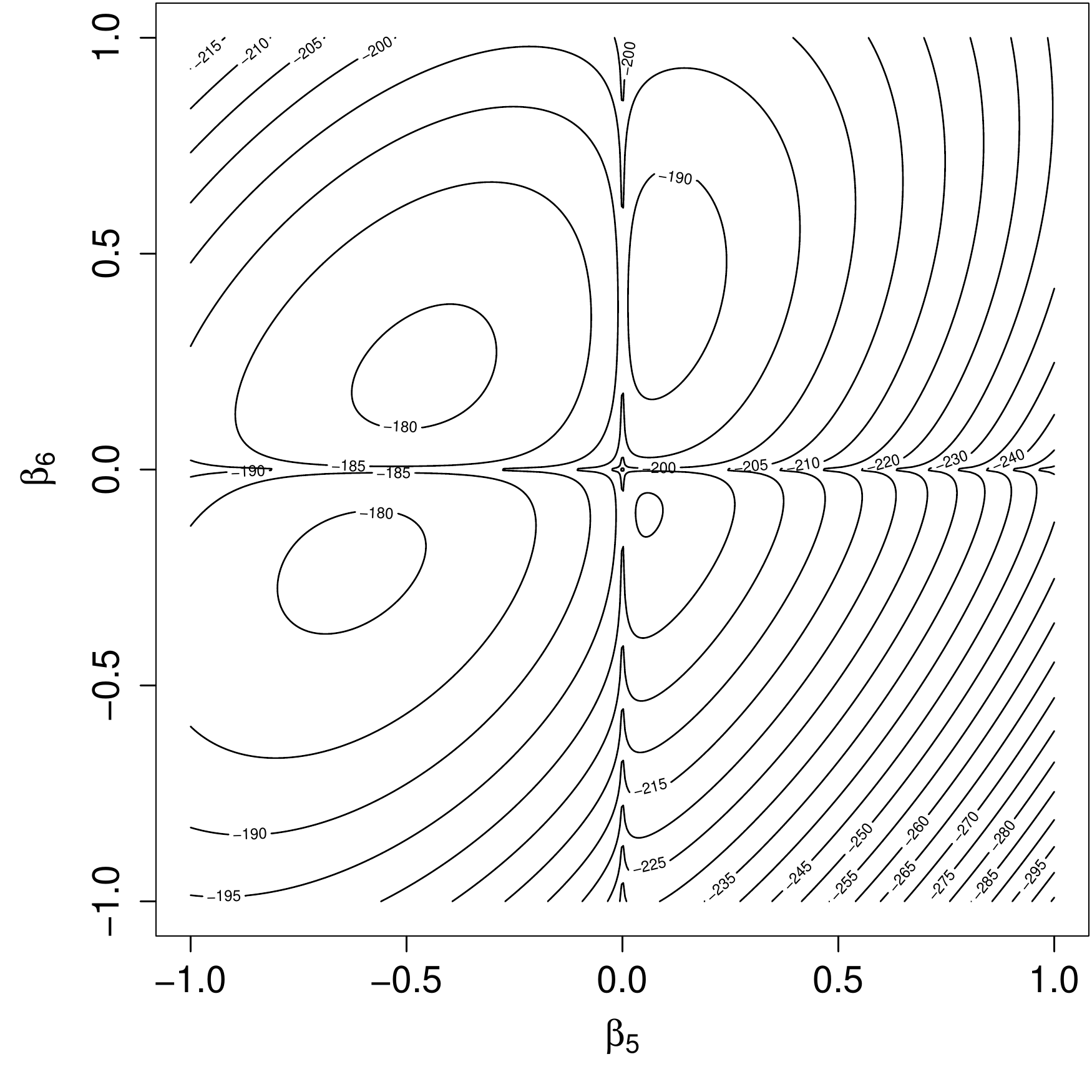}
\end{tabular}
\end{center}
\caption{Simulated linear regression, gMOM prior ($n=50$, $p=10$, $\bbeta=(0.4,0.6,1.2,0.8,0,\ldots,0)$, $\phi=1$).
Left: mean error $\log \tilde{p}^N(y \mid \gamma) - \log p^N(y \mid \gamma)$ for $x_i \sim \mathcal{N}(0,I)$ and $x_i \sim \mathcal{N}(0,V)$, random non-diagonal $V$.
Right: log-integrand $p(y \mid \beta, \phi) p^N(\beta \mid \gamma)$ contours versus two spurious parameters for a randomly-selected dataset, $x_i \sim \mathcal{N}(0,V)$
}
\label{fig:ala_la_error}
\end{figure}

Consider an example with $p=10$, $n \in \{50,200,500\}$ and truly $y \sim \mathcal{N}(Z \beta^*, I)$,
where $\beta^*= (0.4,0.6,1.2,0.8,0,\ldots,0)$. The rows in $Z$ are independent draws $z_i \sim \mathcal{N}(0,I)$ or, alternatively, $z_i \sim \mathcal{N}(0,V)$ for random non-diagonal $V$.
Specifically, $V$ is the correlation matrix associated to $W^TW$, where $W$ is a $p \times p$ matrix with $w_{ij} \sim \mathcal{N}(0,1)$ independently across $i,j$.
We evaluated the integrated likelihood $p^N(y \mid \gamma)$ under the gMOM prior for a sequence of models including $p_\gamma=1,2,\ldots,10$ covariates.
For $p_\gamma \leq 10$ one can evaluate $p^N(y \mid \gamma)$ exactly, and hence the error when estimating $\hat{p}^N(y \mid \gamma)$.
We report average errors across 100 simulations.

Figure \ref{fig:ala_la_error} summarizes the results for the $n=50$ case.
The left panel shows the mean of $\log( \hat{p}^N(y \mid \gamma) / p^N(y \mid \gamma))$, which quantifies the relative approximation error.
Both LA and ALA provided fairly accurate estimates for models including up to 4 covariates. Note that $\beta^*$ is such that for $p_\gamma \leq 4$ all included covariates are truly active. For models with $>4$ covariates the LA error was significantly higher than for ALA, particularly in the non-diagonal covariance setting. 
The right panel in Figure \ref{fig:ala_la_error} illustrates the difficulty of the integration exercise by plotting the contours of the log-integrand $p(y \mid \beta,\phi=1)p^N(\beta \mid \phi=1,\gamma)$ versus two truly spurious parameters $(\beta_5,\beta_6)$ in a randomly selected dataset. The marked multi-modality, in general, does not disappear even as $n \rightarrow \infty$. 
The results for $n=200$ and $n=500$ were largely analogous, see Figure \ref{fig:ala_la_error_boxplot}. 


\subsubsection{Group constraints in non-linear regression}\label{ssec:linreg}

\begin{figure}
\begin{center}
\begin{tabular}{cc}
Non-linear $(x_{i1},x_{i2})$ ($J=5$, $p=30$) & Non-linear $(x_{i1},x_{i2})$ ($J=25$, $p=150$) \\
\includegraphics[width=0.5\textwidth]{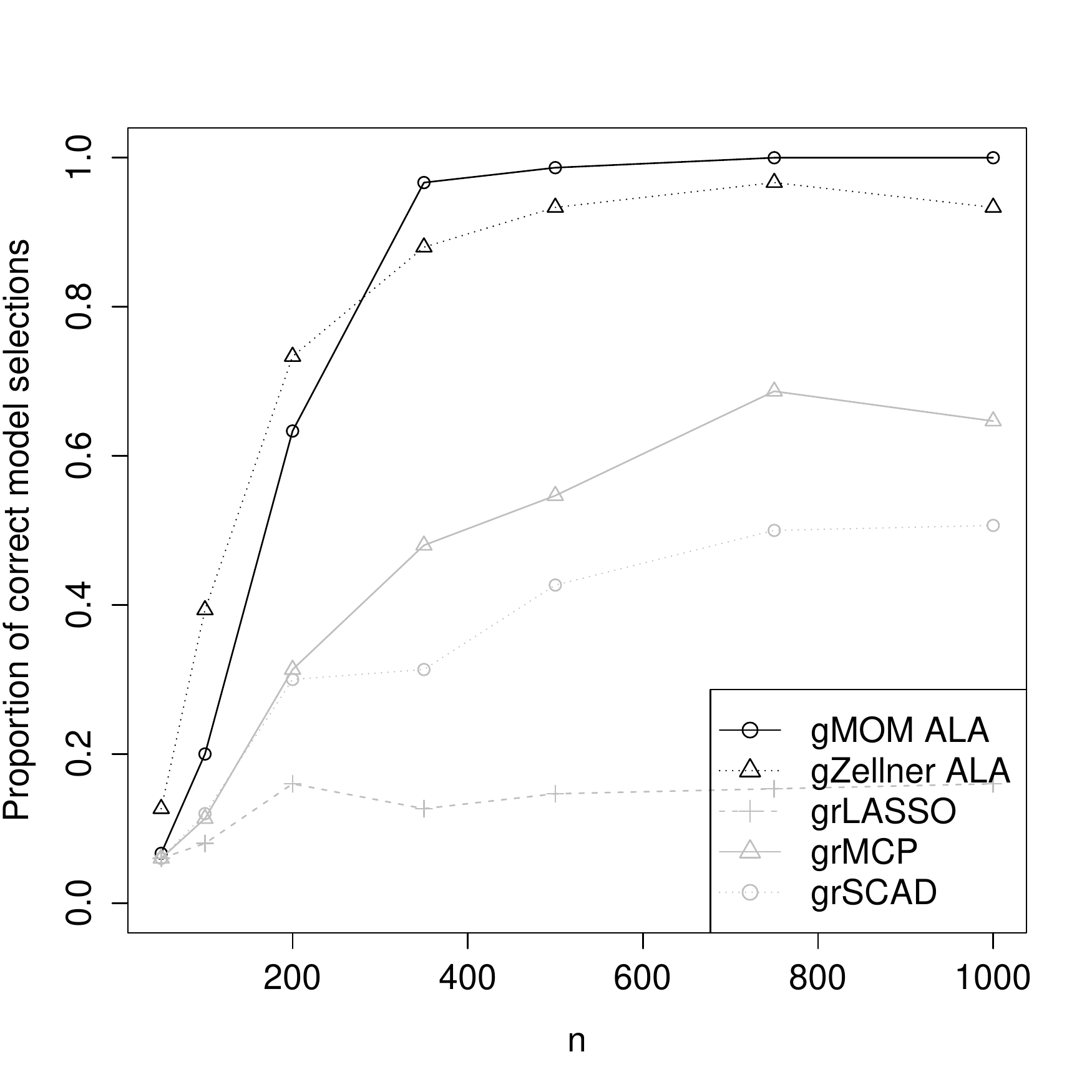} &
\includegraphics[width=0.5\textwidth]{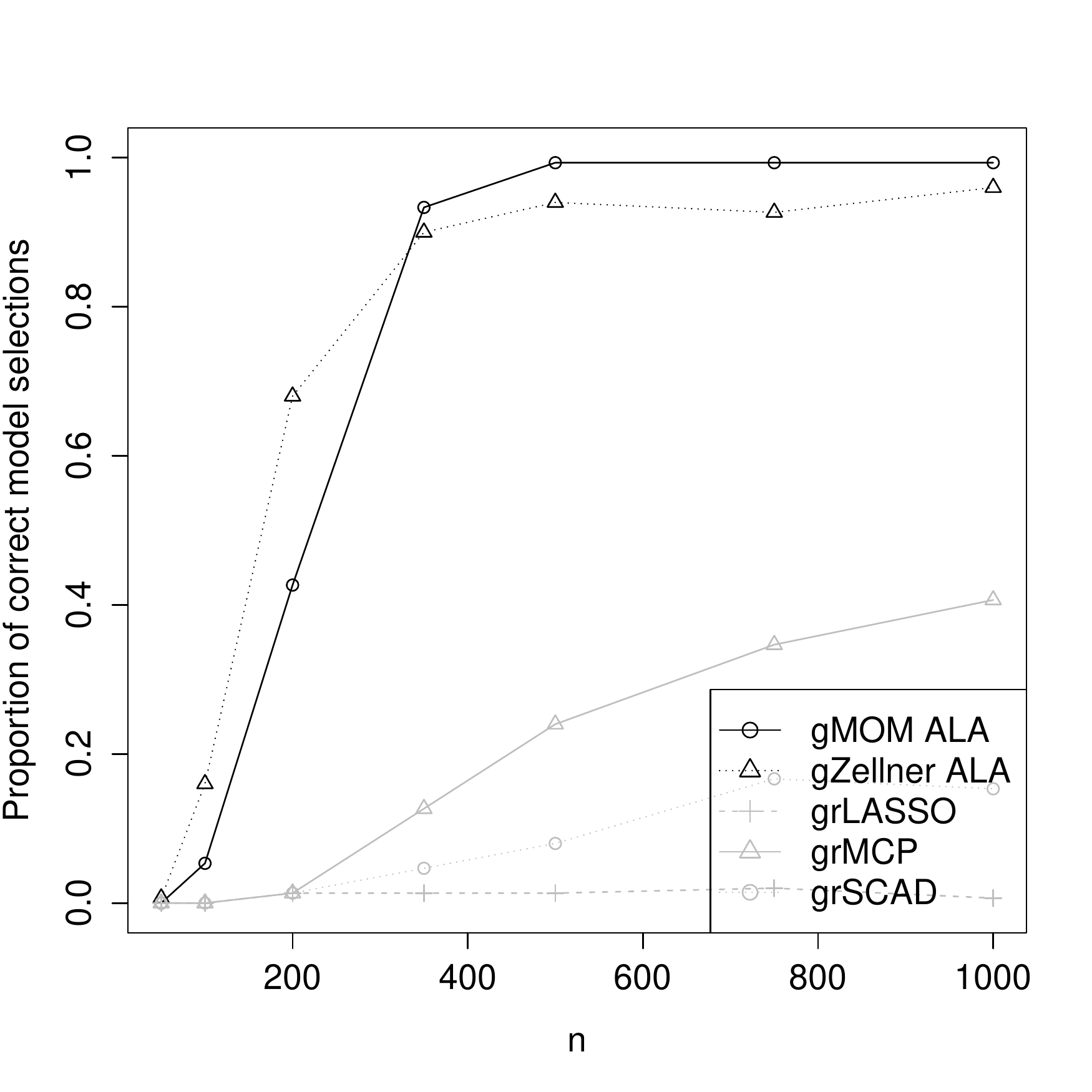}
\end{tabular}
\end{center}
  \caption{Proportion of correct model selections in Gaussian simulations with non-linear effects for two truly active covariates and a total of $J=5$ and $J=25$ covariates ($p=30$ and 150, respectively)}
\label{fig:simlm}
\end{figure}

We present a simulation example where one incorporates group constraints  to model non-linear covariate effects. See Section Section \ref{ssec:suppl_simgaussian} for examples where groups are defined by a categorical covariate. 
We considered the following data-generating truth
\begin{equation*}
  y_i = \sin(-x_{i1} + .1) + \sin(2.5 x_{i2}) + \epsilon_i
\end{equation*}
where $\epsilon_i \sim \mathcal{N}(0,1)$. We considered scenarios where one observes a total of 5 and 50 covariates (including the two truly active covariates), generated from a multivariate Normal with zero mean, unit variances and 0.5 pairwise correlations.

 Suppose that the data analyst poses an additive model that considers non-linear effects but, unaware that the true expectation of $y_i$ depends on $\sin$ functions,  misspecifies their form. Specifically,  the assumed model decomposes covariate effects into a linear plus a deviation-from-linearity component via a 5-dimensional cubic splines, following \cite{rossell:2019t}.
That is, each covariate is coded into the design matrix via one column for its linear effect and a five-variable group capturing deviations from linearity.
Hence the scenarios with 5 and 50 covariates result in $J=10$ and $J=100$ groups (respectively) and in $p=30$ and $p=300$ total parameters (respectively).

Figure \ref{fig:simlm} reports the proportion of correct model selections across 150 simulations. The regularization parameter for grLASSO, grMCP and grSCAD was set via 10-fold cross-validation.
In all settings the proportion of correct model selections were highest for either the exact gZellner calculation (particularly for smaller $n$) or the ALA-based gMOM prior (for larger $n$),  showing that the latter leads to high-quality approximate inference.

%

\subsection{Poverty line}\label{ssec:salary_data}


We studied what factors are associated to individuals working full-time being below the poverty line in the USA.
We used a large dataset from the Current Population Survey \citep{flood:2020} conducted in 2010 and 2019 for single-race individuals aged 18-65 years who were non-military employed for 35-40 hours/week. 

The response is a binary indicator $y_i \in \{0,1\}$ for individual $i$ being below the poverty line. The covariates $x_i$ include gender and hispanic origin indicators, race, marital status, level of education, citizenship status, nativity status, occupation sector, size of the firm employing the individual, the presence of impairment/difficulties, type of employment, moving to another state from within or outside the USA, and the weekly hours worked (35-40). Many of these variables are categorical, see Section \ref{ssec:povertyline_suppl} for a description.
The data has $n=89,755$ individuals.
In a first exercise we considered a logistic regression analysis with only main effects, where $p=60$. Subsequently we considered pairwise interactions between all covariates, then the corresponding design matrix $Z$ has $p=1,469$ columns.

We first discuss the computation times.
In the main effects analysis LA took $>10$ hours to run and ALA took 19.5 seconds. For comparison, GLASSO took 42.7 seconds.
When adding interactions LA took 17.3 days, ALA 5.2 minutes and grLASSO 10.7 minutes.

In the main effects analysis
ALA and LA selected the same model, with virtually identical marginal posterior inclusion probabilities (the correlation between $\hat{p}(\gamma \mid y)$ and $\tilde{p}(\gamma \mid y)$ was $>0.999$).
All main effects had posterior inclusion probability $>0.95$, except for nativity status, with posterior probability $<0.02$.
To help interpret the results, Table \ref{tab:poverty_maineffects_alamodel_mle} provides the estimated coefficients. Briefly, higher poverty odds were estimated for females, hispanics, blacks and native Americans, individuals with difficulties, lower education levels, non-citizens, working in small firms, having moved from outside the US, and working in sectors such as farming or maintenance.
The grLASSO results were similar, except that it selected all main effects, including nativity status (both when setting the penalization parameter via cross-validation or to minimize the BIC).
For comparison, the P-value obtained from a maximum likelihood fit under the full model was 0.1959 for nativity, and $<0.0001$ for all other main effects.

Regarding the analysis with interactions, LA and ALA selected the same 13 main effects (inclusion probability $>0.5$, Table \ref{tab:poverty_selected}).
LA selected 5 interaction terms and ALA selected 8.
Both selected the interaction of gender vs. marital status, education level vs. hispanic origin, and hispanic origin vs. marital status.
LA also selected education vs firm size, whereas ALA selected education vs marital status, and moving state vs hispanic origin, marital status and education level.
The ALA- and LA-estimated marginal posterior probabilities $\tilde{p}(\gamma \mid y)$ and $\hat{p}(\gamma \mid y)$ had 0.827 correlation. 

Table \ref{tab:poverty_alamodel_mle} displays parameter estimates.
We refrain from making any causal interpretation of these findings, but they suggest interesting future research to better understand poverty.
For instance, gender and hispanic origin were associated with poverty, but the differences between hispanic and non-hispanic males was larger (estimated odds-ratio=1.63) than between hispanic and non-hispanic females (odds-ratio= 1.36).
As another example, the odds of poverty were similar for non-hispanic married males and divorced males (odds-ratio=0.92), but married females had significantly lower odds than divorced females (odds-ratio=0.281) and married hispanics had higher odds than divorced hispanics (odds-ratio= 1.46)

For comparison, the grLASSO selected 4 main effects and 16 interaction terms. For only 1 out of the 16 interactions the two corresponding main effects were also selected, 
illustrating the need to explicitly enforce hierarchical restrictions.

\subsection{Cancer survival}\label{ssec:survival_example}

\begin{figure}
\begin{center}
\begin{tabular}{cc}
\includegraphics[width=0.48\textwidth]{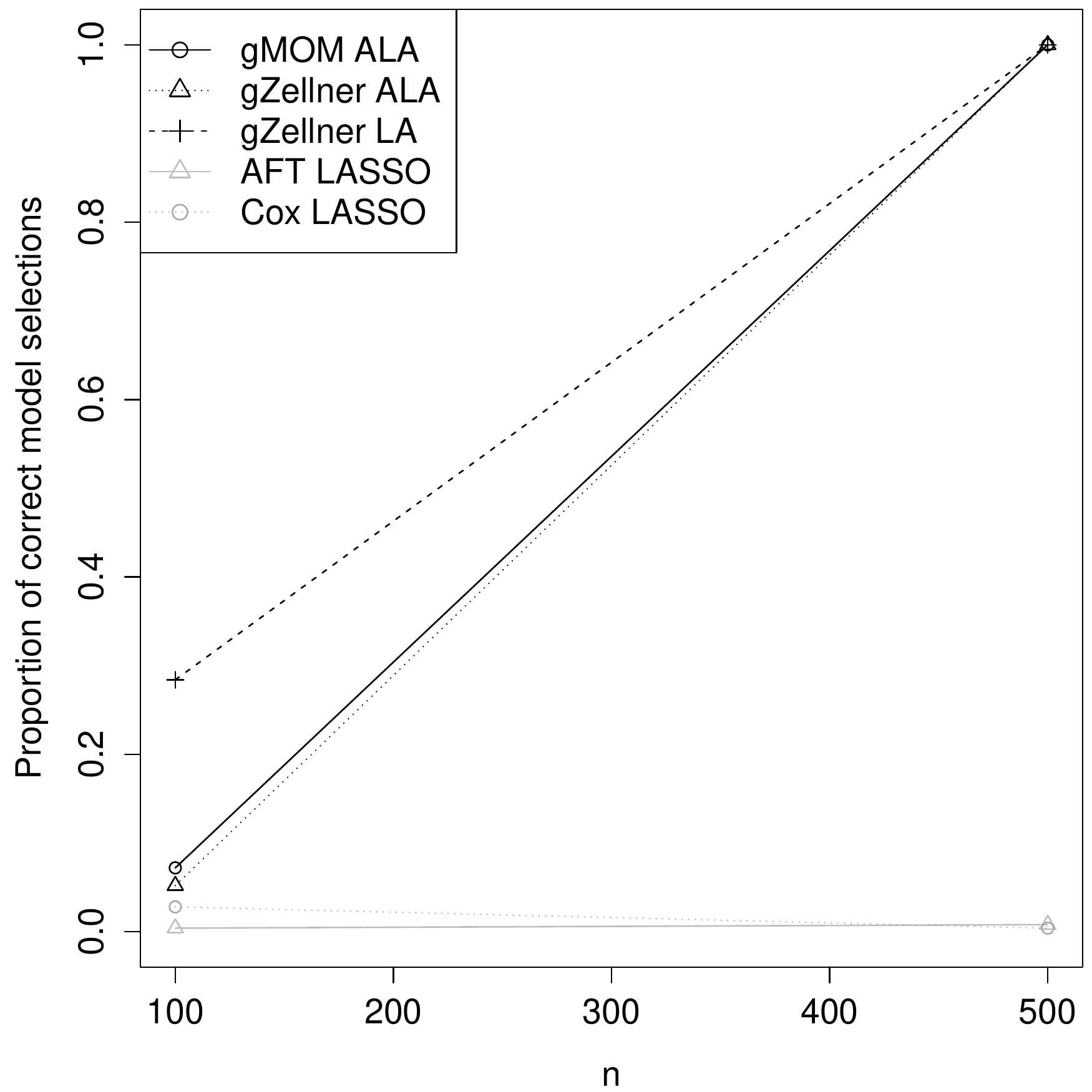} &
\includegraphics[width=0.48\textwidth]{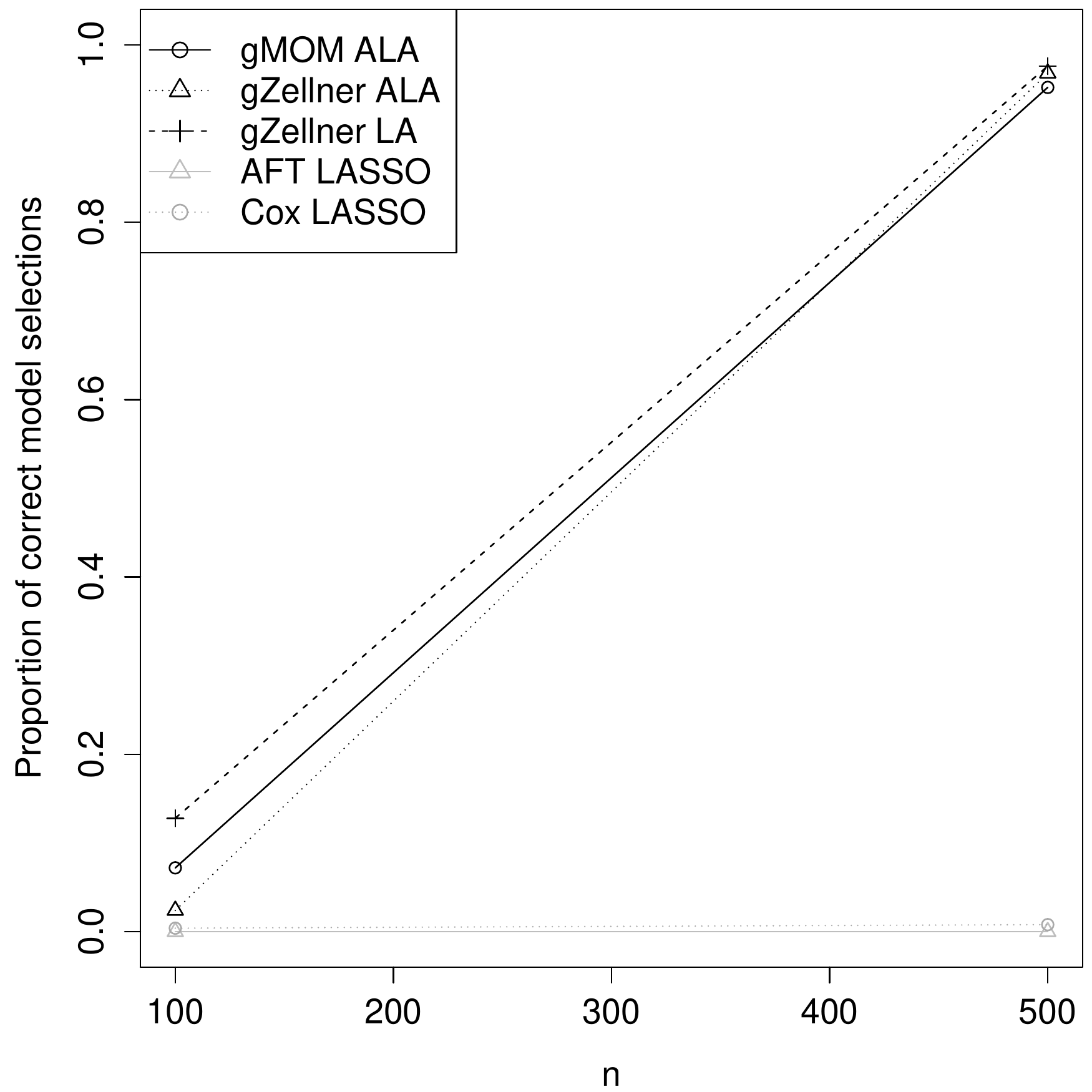}
\end{tabular}
\end{center}
\caption{Non-linear survival regression $(J=100,p=300)$. Proportion of correct model selections under a truly accelerated failure (left) and truly proportional hazards model (right)}
\label{fig:pcorrect_survsim}
\end{figure}

\begin{table}
\begin{center}
\begin{tabular}{|c|cc|} \hline
\multicolumn{3}{|c|}{Poverty data ($n=89,755$)} \\ \hline
              & Main effects ($p=60$) & Interactions ($p=1,469$)\\
gZellner ALA  & 19.5 sec.  & 5.2 min.  \\
gZellner LA   & 10.1 days  & 17.3 days \\
\hline
\multicolumn{3}{|c|}{Simulation under true AFT model ($J=100,p=300$)} \\ \hline
             & $n=100$ &   $n=500$\\
gMOM ALA     & 4.0 sec.  &    7.5 sec.  \\
gZellner ALA & 9.6 sec.  &   13.2 sec.  \\
gZellner LA   & 188.9 sec.&  251.1 sec.\\
\hline
\multicolumn{3}{|c|}{Simulation under true PH model ($J=100,p=300$)} \\ \hline
             & $n=100$ &   $n=500$\\
gMOM ALA     &  2.0 sec.  &    4.1 sec.\\
gZellner ALA &  2.3 sec.  &   10.7 sec.\\
gZellner LA  &  152.4 sec.&   81.1 sec.\\
\hline
\multicolumn{3}{|c|}{Colon cancer data ($n=260$, $p=175$)} \\ \hline
             & Uniform $p(\gamma)$  & BetaBin $p(\gamma)$ \\
gMOM ALA     & 3.3min  & 39.3 sec.\\
gZellner ALA & 11.1min & 48.8 sec. \\
gZellner LA  & 55.0h   & 15.1min \\
\hline
\end{tabular}
\end{center}
\caption{Mean run times for poverty data, survival analysis in truly AFT simulations, truly PH simulations, and colon cancer data.
Laptop with Ubuntu 18.04, Intel i7 1.8GHz processor, 15.4Gb RAM, 1 core}
\label{tab:cputime_survival}
\end{table}

We illustrate the ALA in models outside the exponential family via the non-linear additive survival model from \cite{rossell:2019t}, a spline-based log-normal accelerated failure time (AFT) model.
If one re-parameterizes the model by dividing the regression parameters by the error standard deviation, then the log-likelihood is concave 
\citep{silvapulle:1986}, hence amenable to be analyzed via ALA.
We present a simulation study, and analyze a colon cancer dataset in Section \ref{ssec:coloncancer_suppl}.

We compared the ALA results obtained under gMOM and gZellner priors to the LA under the gZellner prior,
the semi-parametric AFT model with LASSO penalties of \cite{rahaman:2019} (AFT-LASSO),
and to the Cox model with LASSO penalties of \cite{simon:2011} (Cox-LASSO).
For AFT-LASSO and Cox-LASSO we used functions \texttt{AEnet.aft} and \texttt{glmnet} in R packages  \texttt{AdapEnetClass} and  \texttt{glmnet}, and we set the penalization parameter via 10-fold cross-validation.


The simulation study extends that in \cite{rossell:2019t} (Section 6.2).
Briefly, there are two covariates that truly have non-linear effects and 48 spurious covariates, generated from a zero-mean multivariate normal with unit variances and 0.5 correlation between all covariate pairs.
The assumed model poses $E(y_i \mid x_i)= \sum_{j=1}^{50} x_{ij} \beta_j + \sum_{j=51}^{100} z_{ij}^T \beta_j$, where $x_{ij} \in \mathbb{R}$ and $z_{ij} \in \mathbb{R}^5$ captures deviations from linearity by projecting $x_{ij}$ onto a spline basis and orthogonalizing the result to $x_{ij}$. We used a five-dimensional spline basis given that \cite{rossell:2019t} found that larger dimensions gave very similar results.
In our notation, there are $J=100$ groups and $p= 50 + 250= 300$ parameters.

We considered two data-generating truths. In Scenario 1 the truth is an accelerated failure time model and in Scenario 2 a proportional hazards model. Both scenarios are challenging in that there is a significant amount of censored data, which effectively reduces the information in the likelihood, and the dimension is moderately high.
\begin{itemize}
\item Scenario 1. Log-survival times are $x_{i1} + 0.5 \log(|x_{i2}|) + \epsilon_i$, where $\epsilon_i \sim N(0,\sigma=0.5)$.
All log-censoring times are $0.5$, giving an average of 69\% censored individuals.

\item Scenario 2. 
Let $h_0$ be the log-Normal(0,0.5) baseline hazard. Log-survival times arise from a proportional hazards structure $h(t)= h_0(t) \exp \left\{ 3x_{i1}/4 - 5\log(|x_{i2}|)/4 \right\}$.
All log-censoring times are $0.55$, giving an average of 68\% censored individuals.
\end{itemize}

In Scenario 1 the assumed model is well-specified, except for approximating the non-linear truth by a finite spline basis, whereas in Scenario 2 the whole hazard function is misspecified.
Figure \ref{fig:pcorrect_survsim} shows the proportion of correct model selections, across 250 independent simulations for $n \in \{100,500\}$.
Generally, ALA showed a competitive performance. The LA selected the data-generating model slightly more frequently under $n=100$ and the well-specified Scenario 1 (upon inspection this was due to higher power to include $x_{i1}$), whereas it performed similar to ALA for larger $n$ and in Scenario 2.
For $n=500$ the proportions of correct model selections for ALA were near 1. Both ALA and LA outperformed significantly AFT-LASSO and Cox-LASSO.

The ALA provided significant computational gains over LA. In most scenarios the computation time was reduced by a factor ranging from 20-70, see Table \ref{tab:cputime_survival}.

\section{Discussion}
\label{sec:discussion}

Our main contribution was proposing an approximate inference tool that can be particularly helpful in non-Gaussian regression, and in Gaussian outcomes where one wishes to use non-local priors.
 The proposal focuses on scoring models quickly for structural learning problems, and can be combined with parallel computing strategies to accelerate model search. 
The posterior probability rates require the same type of technical conditions than exact inference, and attain essentially the same functional rates in $n$.
 Importantly, ALA and exact calculations asymptotically recover, in general, different models. However we characterized situations where the ALA recovers the correct model, even under misspecification, and showed numerous examples where ALA results agree with exact inference. 
We also illustrated a significant applied potential in reducing computation times, enabling the use of Bayesian model selection to settings where it was previously impractical.

We focused our theory and examples on a simple strategy where Taylor expansions are taken around an $\eta_{\gamma 0}$ with zero regression coefficients.
 We also illustrated that $\eta_{\gamma 0}$ given by 1-2 Newton-Raphson steps typically improves precision, though computations scale more poorly with $(n,p)$. 
In future work it may be interesting to study alternative choices of $\eta_{\gamma 0}$ that balance computational cost and accuracy and/or model selection properties. For instance, under concave log-likelihoods and minimal regularity conditions, it is possible to show that if $\eta_{\gamma 0}= \eta_\gamma^* + O_p(1/\sqrt{n})$, then the ALA provides a consistent estimator of the exact integrated likelihood.

The ALA can in principle be applied to any model, but one expects it to work best when the log-likelihood is concave or at least locally concave around $\eta_{\gamma 0}$, e.g. models satisfying local asymptotic normality.
Yet another avenue is to apply ALA to conditionally concave settings in the spirit of INLA for latent Gaussian models \citep{rue:2009}.

From a foundational Bayesian point-of-view, a limitation of ALA is it they only provides approximate inference.
It would be interesting to explore strategies to combine the ALA with exact computation, e.g. to build proposal distributions for MCMC algorithms, or sequential Monte Carlo strategies  as in \cite{schafer:2013}.
Such extensions require care as both our theory and examples show that a naive combination of ALA and importance sampling can lead to degenerate weights. 
In summary, the current work provides a basis which we hope may lead the ground for multiple interesting extensions.

\beginsupplement

\section*{Supplementary material}

Section \ref{sec:derivation_ala} derives formulae for the approximate Laplace approximation (ALA) to an integral, first for a general likelihood and subsequently for exponential family models with known and unknown dispersion parameter $\phi$.
Section \ref{sec:logl_popularmodels} provides the log-likelihood, gradient and derivatives for common exponential family models such as logistic and Poisson regression.
Section \ref{sec:default_priordispersion} outlines the derivation of default values for the prior dispersion parameters $(g_L,g_N)$ in the Normal- and MOM-based priors.

Sections \ref{sec:postmean_mompenalty}-\ref{sec:proof_bf_ala_highdim} contain the proofs and auxiliary results for 
Lemma \ref{lem:mean_pmompenalty}, Theorem \ref{thm:bf_ala}, Corollary \ref{cor:bf_ala_unknownphi} and Theorem \ref{thm:bf_ala_highdim}.
The auxiliary results in Section \ref{sec:auxiliary_bfala} are either reminders or simple extensions of existing results.
Those in Section \ref{sec:auxiliary_bfala_highdim} are novel results that bound the integral of tail probabilities related to sub-Gaussian random variables, a result that has some independent interest, e.g. to characterize the $L_1$ convergence rate of Bayes factors, which in turn bound the frequentist probability of correct model selection, type I and II errors (see \cite{rossell:2018}, Proposition 1, Corollary 1 and Lemma 1).

Finally, Section \ref{sec:suppl_results} contains supplementary empirical results to complement our examples.

\section{Derivation of ALA}
\label{sec:derivation_ala}

For simplicity we drop the $\gamma$ subindex and denote by $\eta \in \mathbb{R}^p$ the vector of $p$ unknown model parameters.
Let $\eta_0$ be an initial value for $\eta$, $g_0$ and $H_0$ be the gradient and hessian of the negative log-likelihood $-\log p(y \mid \eta)$ evaluated at $\eta_0$.
The goal is to approximate
$$
 p(y)= \int p(y \mid \eta) p(\eta) d\eta.
$$

\subsection{General case}
\label{ssec:derivation_ala_general}

A second order Taylor expansion of $\log p(y \mid \eta)$ at $\eta=\eta_0$ gives
\begin{align}
  &p(y) \approx
  p(y \mid \eta_0)   \int e^{-(\eta - \eta_0)^T g_0 - \frac{1}{2} (\eta - \eta_0)^T H_0 (\eta - \eta_0)} p(\eta)  d\eta
  \nonumber \\
  &=p(y \mid \eta_0) e^{\eta_0^T g_0 - \frac{1}{2} \eta_0 H_0 \eta_0}
  \int e^{-\frac{1}{2} [\eta^T H_0 \eta - 2 (H_0 \eta_0 - g_0)^T \eta]} p(\eta) d\eta
   \nonumber \\
  &=p(y \mid \eta_0) e^{\eta_0^T g_0 - \frac{1}{2} \eta_0 H_0 \eta_0 + \frac{1}{2} \tilde{\eta}^T H_0 \tilde{\eta}}
    \int e^{-\frac{1}{2} (\eta - \tilde{\eta})^T H_0 (\eta - \tilde{\eta})} p(\eta) d\eta
\label{eq:ala_step1}
\end{align}
where $\tilde{\eta}= -H_0^{-1} (-H_0 \eta_0 + g_0)= \eta_0 - H_0^{-1} g_0$.
If $p(\eta)$ is a Normal density then the integral on the right-hand side has closed-form.
More generally, one may use the Laplace approximation
\begin{align}
\int e^{-\frac{1}{2} (\eta - \tilde{\eta})^T H_0 (\eta - \tilde{\eta})} p(\eta) d\eta \approx p(\tilde{\eta}) (2 \pi)^{p/2} |H_0|^{-\frac{1}{2}}.
\label{eq:approx_integral_normallogl}
\end{align}
As a technical remark, the relative error of the latter approximation can be shown to vanish under minimal conditions.
Specifically, one may apply Proposition 8 in \cite{rossell:2019t}. Note that $(\eta - \tilde{\eta})^T H_0 (\eta - \tilde{\eta})$ is convex in $\eta$ and $p(\eta) >0$ for all $\eta$, 
and assume the minimal condition that $H_0 \stackrel{P}{\longrightarrow} H^*$ for some positive-definite matrix $H^*$, when data arise from some true $F^*$. Then
$$
\frac{\int e^{-\frac{1}{2} (\eta - \tilde{\eta})^T H_0 (\eta - \tilde{\eta})} p(\eta) d\eta}
{p(\tilde{\eta}) (2 \pi)^{p/2} |H_0|^{-\frac{1}{2}}}
\stackrel{P}{\longrightarrow} 1
$$
as $n \rightarrow \infty$, where $\stackrel{P}{\longrightarrow}$ denotes convergence in probability under $F^*$.

Plugging \eqref{eq:approx_integral_normallogl} into the right-hand side of \eqref{eq:ala_step1} gives
\begin{align}
\tilde{p}(y)= \frac{p(y \mid \eta_0) p(\tilde{\eta}) (2 \pi)^{p/2}}{|H_0|^{\frac{1}{2}}}
  e^{\eta_0^T g_0 - \frac{1}{2} \eta_0 H_0 \eta_0 + \frac{1}{2} \tilde{\eta}^T H_0 \tilde{\eta}}.
\nonumber
\end{align}
Since $\tilde{\eta}= \eta_0 - H_0^{-1} g_0$, it follows that
$
  \tilde{\eta}^T H_0 \tilde{\eta}=
  \eta_0^T H_0 \eta_0 - 2 g_0^T \eta_0 + g_0^T H_0^{-1} g_0,
$
therefore
\begin{align}
\tilde{p}(y)=  \frac{p(y \mid \eta_0) p(\tilde{\eta}) (2 \pi)^{p/2}}{|H_0|^{\frac{1}{2}}} e^{\frac{1}{2} g_0^T H_0^{-1} g_0},
\label{eq:intlhood_ala_asymp}
\end{align}
as we wished to prove.

It is possible to derive an alternative ALA that is similar expression to~\eqref{eq:intlhood_ala_asymp}, by performing a Taylor approximation on the log-joint
$\log p(y \mid \eta) + \log p(\eta)$, rather than the log-likelihood.
Then $(g_0,H_0)$ are the gradient and hessian of the log-joint $\log p(y \mid \eta) + \log p(\eta)$ at $\eta_0$ and in the right-hand side of \eqref{eq:ala_step1} we would use that
$$
\int e^{-\frac{1}{2} (\eta - \tilde{\eta})^T H_0 (\eta - \tilde{\eta})} d\eta= (2 \pi)^{p/2} |H_0|^{-\frac{1}{2}}
$$
to obtain
\begin{align}
\tilde{p}(y)=  \frac{p(y \mid \eta_0) p(\eta_0) (2 \pi)^{p/2}}{|H_0|^{\frac{1}{2}}} e^{\frac{1}{2} g_0^T H_0^{-1} g_0}.
\nonumber
\end{align}

\subsection{Exponential family case, known $\phi$}
\label{ssec:derivation_ala_knownphi}

We obtain simple expressions for ALA in exponential family densities of the form in \eqref{eq:glm_pdf}, where one takes $\beta_0=0$ and $\phi$ is known.
As before we drop the $\gamma$ subindex, hence $(\beta_\gamma,Z_\gamma,p_\gamma)$ are denoted by $(\beta,Z,p)$, and $p(y \mid \gamma)$ by $p(y)$.
Let $\beta_0$ be an initial value for $\beta$, $g_0$ and $H_0$ be the gradient and hessian of the negative log-likelihood with respect to $\beta$, evaluated at $\beta_0$.
The ALA in \eqref{eq:intlhood_ala_asymp} gives
\begin{align}
\tilde{p}(y)= \frac{p(y \mid \beta_0, \phi) p(\tilde{\beta} \mid \phi) (2 \pi)^{p/2}}{|H_0|^{\frac{1}{2}}}
  e^{\frac{1}{2} g_0^T H_0^{-1} g_0}.
\label{eq:intlhood_ala_knownphi}
\end{align}
If $\beta_0=0$ then
$g_0= -Z^T(y - b(0) \mathbbm{1})/\phi $ and $H_0= (b''(0) / \phi) Z^T Z$, giving
$g_0^T H_0^{-1} g_0= (b''(0)/\phi) \tilde{y}^T Z (Z^TZ)^{-1} Z^T \tilde{y}= (b''(0)/\phi) \tilde{\beta}^T Z^T Z \tilde{\beta}$,
thus
\begin{align}
\tilde{p}(y) &=  \frac{p(y \mid \beta_0, \phi) p(\tilde{\beta} \mid \phi)}{|H_0|^{\frac{1}{2}}}
       \left(\frac{2 \pi \phi}{b''(0)} \right)^{p/2}  e^{\frac{b''(0)}{2 \phi} \tilde{\beta}^T Z^T Z \tilde{\beta}}
       \nonumber \\
  &= e^{ - n b(0)/\phi + \sum_{i=1}^n c(y_i, \phi)}
  p(\tilde{\beta} \mid \phi) \left( \frac{2 \pi \phi}{b''(0)} \right)^{p/2}
|Z^T Z|^{-\frac{1}{2}}  e^{\frac{b''(0)}{2\phi} \tilde{\beta}^T Z^T Z \tilde{\beta}}
,
\nonumber
\end{align}
where $\tilde{\beta}= (Z^T Z)^{-1} Z^T \tilde{y}$ and $\tilde{y}= (y - b(0) \mathbbm{1})/ b''(0)$.

\subsection{Exponential family case, unknown $\phi$}
\label{ssec:derivation_ala_unknownphi}

Consider the case where $\phi$ is unknown, $\beta_0=0$ and $\phi_0= \arg\max p(y \mid \beta_0,\phi)$ is the maximum likelihood estimator of $\phi$ given $\beta=\beta_0$. The ALA in \eqref{eq:intlhood_ala_asymp} gives
\begin{align}
\tilde{p}(y)= p(y \mid \beta_0, \phi_0) p(\tilde{\beta}, \tilde{\phi}) (2 \pi)^{p/2} |H_0|^{-\frac{1}{2}}
  e^{\frac{1}{2} g_0^T H_0^{-1} g_0}.
\nonumber
\end{align}
where recall that
\begin{align}
  g_0&= -\frac{b''(0)}{\phi_0} \begin{pmatrix} Z^T \tilde{y} \\ 0 \end{pmatrix} \nonumber \\
  H_0&= \frac{b''(0)}{\phi_0} \begin{pmatrix} Z^TZ & -Z^T \tilde{y}/\phi_0 \\ 
-\tilde{y}^T Z/\phi_0 & s(\phi_0) \end{pmatrix}, \nonumber
\end{align}
where $s(\phi_0)= [2n b(0)/\phi_0^2 + \phi_0 \sum_{i=1}^n \nabla_{\phi \phi}^2 c(y_i,\phi_0)]/b''(0)$.

Since the last entry in $g_0$ is 0 it follows that
$ g_0^T H_0^{-1} g_0= (b''(0)/\phi_0)^2 \tilde{y}^T Z A Z^T \tilde{y}$,
where $A$ is the submatrix of $H_0^{-1}$ corresponding to $\tilde{\beta}$.
Hence, applying the block-wise matrix inversion formula to $H_0$ gives that
\begin{align}
 g_0^T H_0^{-1} g_0&=
\frac{b''(0)}{\phi_0} \tilde{y}^T Z
\left[(Z^TZ)^{-1} + \frac{ (Z^TZ)^{-1} Z^T \tilde{y} \tilde{y}^T Z (Z^TZ)^{-1}}{ \phi_0^2( s(\phi_0) - \tilde{y}^T Z  (Z^TZ)^{-1} Z^T \tilde{y})} \right]
Z^T \tilde{y}=
\nonumber \\
&= 
\frac{b''(0)}{\phi_0} \left[ \tilde{\beta}^T Z^T Z \tilde{\beta}
+ \frac{\tilde{y}^T Z \tilde{\beta} \tilde{\beta}^T Z^T \tilde{y}}{\phi_0^2 (s(\phi_0) - \tilde{\beta}^T Z^TZ \tilde{\beta})} \right]
\nonumber \\
&=
\frac{b''(0)}{\phi_0} \left[ \tilde{\beta}^T Z^T Z \tilde{\beta}
+ \frac{(\tilde{\beta}^T Z^TZ \tilde{\beta})^2}{\phi_0^2 (s(\phi_0) - \tilde{\beta}^T Z^TZ \tilde{\beta})} \right].
\nonumber \\
&=
\frac{b''(0)}{\phi_0} \tilde{\beta}^T Z^T Z \tilde{\beta} 
\left[ 1 + \frac{\tilde{\beta}^T Z^TZ \tilde{\beta}}{\phi_0^2 (s(\phi_0) - \tilde{\beta}^T Z^TZ \tilde{\beta})} \right]
\nonumber
\end{align}
Therefore,
\begin{align}
\tilde{p}(y)= p(y \mid \beta_0, \phi_0) p(\tilde{\beta}, \tilde{\phi}) (2 \pi)^{p/2} |H_0|^{-\frac{1}{2}}
  \exp\left\{\frac{b''(0) t(\tilde{\beta},\phi_0)}{2\phi_0} \tilde{\beta}^T Z^T Z \tilde{\beta}  \right\},
\label{eq:intlhood_ala_unknownphi}
\end{align}
where
$$
t(\tilde{\beta},\phi_0)= 1 + \frac{\tilde{\beta}^T Z^TZ \tilde{\beta}}{\phi_0^2 (s(\phi_0) - \tilde{\beta}^T Z^TZ \tilde{\beta})},
$$
as we wished to prove.

\subsection{Exponential family case, known $\phi$. Curvature adjustment}
\label{ssec:ala_curvature_adjustment}

This section derives a curvature-adjusted ALA that adjusts the log-likelihood hessian to account for differences in the model-predicted and observed variance of $y$, in the case where the dispersion parameter $\phi$ is a fixed constant (e.g. $\phi=1$ in logistic and Poisson models).

Let $g_{\gamma 0}= -Z_\gamma^T(y - b'(Z_\gamma\beta_{\gamma 0}))/\phi$, $H_{\gamma 0}= Z_\gamma^T \mbox{diag}(b''(Z_\gamma\beta_{\gamma 0})) Z_\gamma/\phi$ be the gradient and hessian at an arbitrary $\beta_{\gamma 0}$, that can potentially depend on the data $(y,Z_\gamma)$.
Then 
\begin{align}
g_{\gamma 0}^T H_{\gamma 0}^{-1} g_{\gamma 0}
     &=  (y - \mu_{\gamma 0})^T Z_\gamma (Z_\gamma^T \textnormal{Cov}(y \mid Z_\gamma, \beta=\beta_{\gamma 0},\phi) Z_\gamma)^{-1} Z_\gamma^T (y - \mu_{\gamma 0}),
                    \nonumber \\
  H_{\gamma 0}&= Z_\gamma^T \mbox{Cov}(y \mid Z_\gamma, \beta=\beta_{\gamma 0},\phi) Z_\gamma / \phi^2
  \nonumber
\end{align}
where $\mu_{\gamma 0}= E(y \mid Z_\gamma, \beta=\beta_{\gamma 0},\phi)= b'(Z_\gamma \beta_{\gamma 0})$ is the model-predicted expectation at $\beta_{\gamma 0}$ and
$\mbox{Cov}(y \mid Z_\gamma, \beta=\beta_{\gamma 0},\phi)= \phi \, \textnormal{diag}(b''(Z_\gamma \beta_{\gamma 0}))$ the model-predicted covariance.
Hence the ALA is
\begin{align}
  \tilde{p}(y \mid \gamma)= \frac{p(y \mid \beta_{\gamma 0},\phi) p(\tilde{\beta}_\gamma \mid \phi,\gamma) (2\pi \phi^2)^{p_\gamma/2}}{|Z_\gamma^T \mbox{Cov}(y \mid Z_\gamma, \beta=\beta_{\gamma 0},\phi) Z_\gamma|^{1/2}}
  e^{\frac{1}{2}  (y - \mu_{\gamma 0})^T Z_\gamma (Z_\gamma^T \textnormal{Cov}(y \mid Z_\gamma, \beta=\beta_{\gamma 0},\phi) Z_\gamma)^{-1} Z_\gamma^T (y - \mu_{\gamma 0})}.
\nonumber
\end{align}

This expression relates directly $Z_\gamma^T(y - \mu_{\gamma 0})$ to its model-predicted covariance $Z_\gamma^T \textnormal{Cov}(y \mid Z_\gamma, \beta=\beta_{\gamma 0},\phi) Z_\gamma$.
A practical issue is that, even when the data are truly generated from a distribution $F^*$ included in the assumed model \eqref{eq:glm_pdf} for some true $\beta^*$, there is a mismatch between
$\textnormal{Cov}(y \mid Z_\gamma, \beta=\beta_{\gamma 0},\phi)=E((y-\mu_{\gamma 0})^T(y-\mu_{\gamma 0}) \mid Z_\gamma, \beta=\beta_{\gamma 0},\phi)$
and $E_{F^*}[(y-\mu_{\gamma 0})^T(y-\mu_{\gamma 0})]= \phi \, \textnormal{diag}(b''(Z_\gamma\beta^*)) + E_{F^*}[(\mu_{\gamma 0} - \mu_\gamma^*) (\mu_{\gamma 0} - \mu_\gamma^*)^T]$.
That is, the data may be either over- or under-dispersed relative to the model prediction.

To address this issue we estimate
$E_{F^*}[(y-\mu_{\gamma 0})^T(y-\mu_{\gamma 0})]$ non-parametrically, under the assumption that $(y_i,z_i)$ are independent draws from $F^*$ across $i=1,\ldots, n$. Specifically, we use the usual over-dispersion estimator based on Pearson chi-square residuals,
\begin{align}
\hat{\rho}= \frac{\sum_{i=1}^n (y_i - h^{-1}(z_i^T\beta_{\gamma 0}))^2/[\phi b''(z_i^T \beta_{\gamma 0})]}{n - |\beta_{\gamma 0}|_0}
\nonumber
\end{align}
where $h()$ is the link function,
$h^{-1}(z_i^T\beta_{\gamma 0})= E(y_i \mid z_{i\gamma}, \beta=\beta_{\gamma 0},\phi)$ 
and  $|\beta_{\gamma 0}|_0$ the number of non-zero elements in $\beta_{\gamma 0}$.

The curvature-adjusted ALA is obtained by replacing $\textnormal{cov}(y \mid \beta=\beta_{\gamma 0},\phi)$ for
$\hat{\rho} \textnormal{cov}(y \mid \beta=\beta_{\gamma 0},\phi)= \hat{\rho} \phi \textnormal{diag}(b''(Z_\gamma\beta_{\gamma 0}))$, which gives
\begin{align}
  \tilde{p}(y \mid \gamma)&=
  \frac{p(y \mid \beta_{\gamma 0},\phi) p(\tilde{\beta}_\gamma \mid \phi,\gamma) (2\pi \phi/\hat{\rho})^{p_\gamma/2}}{|Z_\gamma^T \textnormal{diag}(b''(Z_\gamma\beta_{\gamma 0})) Z_\gamma|^{1/2}}
  e^{\frac{1}{2 \hat{\rho} \phi}  (y - \mu_{\gamma 0})^T Z_\gamma (Z_\gamma^T \textnormal{diag}(b''(Z_\gamma\beta_{\gamma 0})) Z_\gamma)^{-1} Z_\gamma^T (y - \mu_{\gamma 0})}.
\label{eq:marglhood_ala_curvadj}
\end{align}

We remark that there are other possible strategies to adjust the curvature. For example, one could obtain the model-predicted curvature at $\tilde{\beta}_{\gamma 0}$ and adjust for over-dispersion by replacing $\textnormal{Cov}(y \mid \beta=\beta_{\gamma 0},\phi)$ for
$\tilde{\rho}_\gamma \textnormal{Cov}(y \mid Z_\gamma, \beta=\tilde{\beta}_{\gamma 0},\phi)= \tilde{\rho} \, \phi \, \textnormal{diag}(b''(Z_\gamma\tilde{\beta}_{\gamma 0}))$,
where $\tilde{\rho}_\gamma= \sum_{i=1}^n (y_i - z_i^T\tilde{\beta}_{\gamma 0})^2/[\phi b''(z_i^T \tilde{\beta}_{\gamma 0})] /(n - p_\gamma)$.
Such an alternative has the potential to improve inference by obtaining a better approximation to the hessian at the limiting $\beta^*$ recovered by maximum likelihood estimation that is specific for each model $\gamma$. The computational cost of $\tilde{\rho}_\gamma$ is however slightly higher than for $\hat{\rho}$, when the latter is evaluated at a common $\beta_{\gamma 0}$ that does not vary across models. We therefore focus on $\hat{\rho}$ in the sequel.

Expression \eqref{eq:marglhood_ala_curvadj} can be simplified in the particular case where $Z_\gamma\beta_{\gamma 0}$ does not depend on $Z_\gamma$, that is when either $\beta_{\gamma 0}=0$ or the only non-zero entry in $\beta_{\gamma 0}$ corresponds to the intercept.
That is, suppose that $z_i^T\beta_{\gamma 0}=\nu_0$ for all $i=1,\ldots,n$ and some $\nu_0 \in \mathbb{R}$, then
\begin{align}
  \tilde{p}(y \mid \gamma)&=
                \frac{p(y \mid \beta_{\gamma 0},\phi) p(\tilde{\beta}_\gamma \mid \phi,\gamma)}{|Z_\gamma^T Z_\gamma|^{1/2}}
                 \left(\frac{2\pi \phi}{\hat{\rho} b''(\nu_{\gamma 0})}\right)^{p_\gamma/2}
                e^{\frac{1}{2 \hat{\rho} \phi b''(\nu_{\gamma 0})}  (y - \mu_{\gamma 0})^T Z_\gamma (Z_\gamma^T Z_\gamma)^{-1} Z_\gamma^T (y - \mu_{\gamma 0})}
                \nonumber \\
              &=
                \frac{e^{n[\nu_0 \bar{y} - b(\nu_0)]/\phi + \sum_{i=1}^n c(y_i,\phi)}
                p(\tilde{\beta}_\gamma \mid \phi, \gamma)}{|Z_\gamma^T Z_\gamma|^{1/2}}
                 \left(\frac{2\pi \phi}{\hat{\rho} b''(\nu_0)}\right)^{p_\gamma/2}
                e^{\frac{b''(\nu_0)}{2 \hat{\rho} \phi}  \tilde{y}^T Z_\gamma (Z_\gamma^T Z_\gamma)^{-1} Z_\gamma^T \tilde{y}}
\nonumber
\end{align}
where $\tilde{y}= (y - b'(\nu_0))/b''(\nu_0)$, and $\bar{y}= \sum_{i=1}^n y_i /n$ is the sample mean.
Note that $\nu_0$ may depend on the data, and in particular we take the MLE under the intercept-only model $\nu_0=h(\bar{y})$, but is constant across models. Similarly $\hat{\rho}$ is also constant across models by construction.
Therefore the curvature-adjusted ALA BF is
\begin{align}
\tilde{B}_{\gamma \gamma'}=
\frac{p(\tilde{\beta}_\gamma \mid \phi, \gamma) |Z_{\gamma'}^T Z_{\gamma'}|^{1/2}}{p(\tilde{\beta}_{\gamma'} \mid \phi, \gamma') |Z_\gamma^T Z_\gamma|^{1/2}}
\left(\frac{2\pi \phi}{\hat{\rho} b''(\nu_0)}\right)^{(p_\gamma-p_{\gamma'})/2}
e^{\frac{b''(\nu_0)}{2 \hat{\rho} \phi}  [\tilde{y}^T Z_\gamma (Z_\gamma^T Z_\gamma)^{-1} Z_\gamma^T \tilde{y} - \tilde{y}^T Z_{\gamma'} (Z_{\gamma'}^T Z_{\gamma'})^{-1} Z_{\gamma'}^T \tilde{y}]}.
\nonumber
\end{align}
Finally, note also that for $\nu_0=h(\bar{y})$ we obtain $b''(nu_0)=b''(h(\bar{y}))$ and a particularly simple over-dispersion estimator
\begin{align}
\hat{\rho}= \frac{\sum_{i=1}^n (y_i - \bar{y})^2/[\phi b''(h(\bar{y}))]}{n - 1}.
\nonumber
\end{align}

For logistic regression $h(\bar{y})= \log(\bar{y}/(1-\bar{y}))$ and $b''(u)=e^{-u}/(1+e^{-u})^2$ (Section \ref{ssec:logl_logisticregr}), then using straightforward algebra gives $b''(h(\bar{y}))= \bar{y}(1-\bar{y})$.
For Poisson regression $h(\bar{y})= \log(\bar{y})$ and $b''(u)=e^u$ (Section \ref{ssec:logl_poissonregr}), so that $b''(h(\bar{y}))= \bar{y}$.

\section{Log-likelihood and derivatives for popular exponential family models}
\label{sec:logl_popularmodels}

\subsection{Logistic regression}
\label{ssec:logl_logisticregr}

The logistic regression log-likelihood is
\begin{align}
  \log p(y \mid \beta)=  y^T Z\beta - \sum_{i=1}^n \log (1 + e^{z_i^T\beta}),
\nonumber
\end{align}
corresponding to $b(u)= \log (1 + e^u)$ in~\eqref{eq:glm_pdf}, $b'(u)= 1/(1+e^{-u})$ and $b''(u)=e^{-u}/(1+e^{-u})^2$.
Its gradient and hessian are
\begin{align}
  g_\beta(\beta)&=  Z^Ty - \sum_{i=1}^n z_i \frac{1}{1 + e^{-z_i^T\beta}}
  \nonumber \\
  H_{\beta\beta}(\beta)&= - \sum_{i=1}^n z_i z_i^T \frac{e^{-z_i^T\beta}}{(1 + e^{-z_i^T\beta})^2}= - Z^T D_\beta Z \nonumber
\end{align}
where $D_\beta$ is a diagonal matrix with $(i,i)$ entry
$
1/[(1 + e^{-z_i^T\beta})(1+ e^{z_i^T\beta})].
$

\subsection{Poisson regression}
\label{ssec:logl_poissonregr}

The Poisson regression log-likelihood is
\begin{align}
  \log p(y \mid \beta)= 
y^T Z\beta - \sum_{i=1}^n e^{z_i^T\beta} - \log(y_i!),
\nonumber
\end{align}
corresponding to $b(u)= e^u$, $b'(u)=b''(u)=e^u$ and $c(y_i)= \log(y_i!)$ in~\eqref{eq:glm_pdf}.
Its gradient and hessian are
\begin{align}
  g_\beta(\beta)&= Z^T y - \sum_{i=1}^n z_i e^{z_i^T\beta}
  \nonumber \\
  H_{\beta,\beta}(\beta)&= - \sum_{i=1}^n z_i z_i^T e^{z_i^T\beta}= - Z^T
  D_\beta Z \nonumber
\nonumber
\end{align}
where $D_\beta$ is a diagonal matrix with $(i,i)$ entry $e^{z_i^T\beta}$.

\subsection{Gaussian accelerated failure time model}
\label{ssec:logl_aft}

Let $o_i \in \mathbb{R}^+$ be survival times and $c_i \in \mathbb{R}^+$ be right-censoring times
for individuals $i=1,\ldots,n$, $d_i=\mbox{I}(o_i<c_i)$ censoring indicators, and $y_i=\min\{\log(o_i),\log(c_i)\}$ the observed log-times. The Gaussian accelerated failure time model assumes

\begin{eqnarray}
\log(o_i) = z_i^T\beta +  \epsilon_i,
\nonumber
\end{eqnarray}
where $\epsilon_i \sim N(0,\phi)$ independently across $i=1,\ldots,n$.
It is convenient to reparameterize the model via $\alpha = \beta/\sigma$ and $\tau = 1/\sqrt{\phi}$, as then that the log-likelihood is concave, provided that the number of uncensored individuals $n_o \geq \mbox{dim}(\beta)$ and the sub-matrix of $Z$ corresponding to uncensored individuals has full column rank \citep{burridge:1981,silvapulle:1986}.
The log-likelihood function is
\begin{eqnarray}
p(y \mid \alpha, \tau) &=& -\frac{n_o}{2} \log\left(\frac{2\pi}{\tau^2}\right) - \dfrac{1}{2} \sum_{d_i=1}(\tau y_i-z_i^T \alpha)^2 
 + \sum_{d_i=0} \log \left\{\Phi\left(z_i^T \alpha - \tau y_i\right)\right\},
\nonumber
\end{eqnarray}
its gradient is
\begin{equation}
\begin{pmatrix}
g_\alpha(\alpha,\tau) \\
g_\tau(\alpha,\tau) 
\end{pmatrix} =
\begin{pmatrix}
\sum_{d_i=1}  z_i  \left(\tau y_i-z_i^T \alpha  \right) + \sum_{d_i=0}  z_i r\left(z_i^T\alpha  - \tau y_i \right)  \\
\dfrac{n_o}{\tau} -  \sum_{d_i=1} y_i\left(\tau y_i-z_i^T\alpha \right)  - \sum_{d_i=0} y_i r\left(z_i^T\alpha  - \tau y_i \right)
\end{pmatrix},
\nonumber
\end{equation}
and its hessian is
\begin{eqnarray*}
  H_{\alpha,\alpha}(\alpha,\tau)&=&
 - \sum_{d_i=1} z_i z_i^T
- \sum_{d_i=0} z_i z_i^T  D\left(\tau y_i - z_i^T\alpha \right),\\
 H_{\alpha,\kappa}(\alpha,\tau) &=& \sum_{d_i=1} z_i  y_i + \sum_{d_i=0} z_i y_i D\left(\tau y_i - z_i^T\alpha  \right),\\
 H_{\tau,\tau}(\alpha,\tau)  &=&  -\dfrac{n_o}{\tau^2} - \sum_{d_i=1} y_i^2  - \sum_{d_i=0} y_i^2 D\left(\tau y_i - z_i^T\alpha \right),
\end{eqnarray*}
where $D\left(z \right) = r\left(-z \right)^2 -z r\left(-z \right) \in (0,1)$,
$r(t) = f(t)/\Phi(t)$ is the Normal inverse Mills ratio, $f()$ is the standard Normal probability density function and $\Phi()$ its cumulative density function.

\section{Derivation of default $(g_L,g_N)$}
\label{sec:default_priordispersion}

The goal is to find prior dispersion parameters such that the prior expectation
$$E( \beta_j^T Z_j^T Z_j \beta_j)/[n\phi] )=1$$
for all $j$.
First note that if $\beta_j \sim \mathcal{N}(0, g \phi n/p_j] (Z_j^T Z_j)^{-1})$,
then $\beta_j^T Z_j^T Z_j \beta_j p_j/[g \phi n] \sim \chi_{p_j}^2$, which has mean $p_j$ and variance $2 p_j$.

Therefore, under the Normal prior in~\eqref{eq:prior_parameters} we have
$$E\left(\frac{\beta_j^T Z_j^T Z_j^T \beta_j}{\phi n} \right)=
\frac{g_L}{p_j}
\int \frac{\beta_j^T Z_j^T Z_j^T \beta_j p_j}{g_L \phi n} \mathcal{N}(0, \frac{g_L \phi n}{p_j} (Z_j^T Z_j)^{-1}) d \beta_j=
\frac{g_L}{p_j} p_j,$$
obtaining the default $g_L= 1$.
For the gMOM prior
\begin{align}
&E\left( \frac{\beta_j^T Z_j^T Z_j \beta_j}{\phi n} \right)=
  \int \left[ \frac{\beta_j^T Z_j^T Z_j \beta_j}{\phi n} \right]^2 \frac{(p_j+2)}{p_j g_N}
  \mathcal{N}\left(\beta_j; 0, \frac{g_N \phi n}{p_j+2} (Z_j^T Z_j)^{-1} \right) d\beta_j
                \nonumber \\
& = \frac{g_N}{p_j(p_j+2)} \int \left[ \frac{\beta_j^T Z_j^T Z_j \beta_j (p_j+2)}{\phi n g_N} \right]^2
  \mathcal{N}\left(\beta_j; 0, \frac{g_N \phi n}{p_j(p_j+2)} (Z_j^T Z_j)^{-1} \right) d\beta_j
  =\frac{g_N (2 p_j + p_j^2)}{p_j(p_j+2)}= g_N,
    \nonumber
  \end{align}
giving the default $g_N= 1$.

\section{Proof of Lemma \ref{lem:mean_pmompenalty}}
\label{sec:postmean_mompenalty}

  Consider the first integral with respect to $\xi$, which is equal to the expectation $\phi^{-1} E( \xi^T A \xi )$.
  Let $\tilde{\xi}=A^{1/2} \xi$, since $\tilde{\xi} \sim \mathcal{N}(A^{1/2} m, \phi A^{1/2} S A^{1/2})$ we have that
  $$
  E( \xi^T A \xi )= E(\tilde{\xi}^T \tilde{\xi})= 
  \sum_{j=1}^{l} V(\tilde{\xi}_{jl}) + [E(\tilde{\xi}_{jl})]^2
  = \phi \mbox{tr}(\tilde{S}) + \tilde{m}^T\tilde{m}=
  \phi \mbox{tr}(A S) + m^T A m,
  $$
  where $\tilde{S}= A^{1/2} S A^{1/2}$, $\tilde{m}= A^{1/2} m$, proving the first part of the result.
  Now, integrating with respect to $\phi$ gives
  \begin{align}
  \int (\mbox{tr}(A S) + m^T A m/\phi)  \mbox{IG}(\phi; a,b) d\phi=
  \mbox{tr}(A S) + m^T A m \frac{a}{b},
  \nonumber
  \end{align}
  as we wished to prove.
  
\section{Proof of Theorem \ref{thm:sparse_linproj}}
\label{sec:sparse_linproj}

We shall use $(\msf{y}, \msf{z})$ to generically denote a random vector in $\mathbb{R} \otimes \mathbb{R}^p$ distributed as $F^*$, so that $\{(y_i, z_i)\}_{i=1}^n$ are $n$ iid copies of $(\msf{y}, \msf{z})$. To ease notation, we shall use $\msf{z}_1$ and $\msf{z}_2$ in place of $\msf{z}_{\delta}$ and $\msf{z}_{\upsilon}$ in this proof. Let us partition $\Sigma := \mbox{Cov}_{F^*}(\msf{z})$ into the corresponding blocks to define $\Sigma_{11} := \mbox{Cov}_{F^*}(\msf{z}_1)$, $\Sigma_{22} := \mbox{Cov}_{F^*}(\msf{z}_2)$, and $\Sigma_{12} := \mbox{Cov}_{F^*}(\msf{z}_1,\msf{z}_2) = \Sigma_{21}'$. Since $\delta$ is the true regression model, $E_{F^*}(\msf{y} \mid \msf{z}) = f(\msf{z}_1)$. Also, by the assumption in equation \eqref{eq:lin_exp} (henceforth assumption \eqref{eq:lin_exp}), $E_{F^*}(\msf{z}_2 \mid \msf{z}_1) = A \msf{z}_1$. 

Observe that $E_{F^*} (Z'Z) = n E_{F^*} (\msf{z} \msf{z}') = n \Sigma$ and $E_{F^*}(Z'\tilde{y}) = n E_{F^*}[\msf{z} \, (\msf{y} - b'(0))/b''(0)] = n E_{F^*}(\msf{z} \, \msf{y})$, where in both identities we use $E_{F^*}(\msf{z}) = 0$ in the last step. Let $m = E_{F^*}(\msf{z} \, \msf{y})$ so that $\tilde{\beta}^* = \Sigma^{-1} m$. Using the tower property of conditional expectations, $m = E_{F^*}[ \msf{z} \, E_{F^*}( \msf{y} \mid \msf{z} )] = E_{F^*}[ \msf{z} \, f(\msf{z}_1)]$. Partition $m$ into the same blocks as $\msf{z}$ to write $m = (m_1; m_2)$, so that $m_1 = E_{F^*}[ \msf{z}_1 \, f(\msf{z}_1)]$ and $m_2 = E_{F^*}[ \msf{z}_2 \, f(\msf{z}_1)]$. Using the tower property and assumption \eqref{eq:lin_exp}, write $m_2 = E_{F^*}[ f(\msf{z}_1) \, E_{F^*}( \msf{z}_2 \mid \msf{z}_1 ) ] = E_{F^*} [ f(\msf{z}_1) \, A \msf{z}_1] = A E_{F^*}[ \msf{z}_1 \, f(\msf{z}_1)] = A m_1$. 

Since $\mbox{Cov}_{F^*}(\msf{z}_2 - E(\msf{z}_2 \mid \msf{z}_1), \msf{z}_1) = 0$, assumption \eqref{eq:lin_exp} implies that $\Sigma_{21} = A \Sigma_{11}$, and thus $A = \Sigma_{21} \Sigma_{11}^{-1}$; note that since $\Sigma$ is positive definite, both $\Sigma_{11}$ and $\Sigma_{22}$ are positive definite as well. Let $Q = \Sigma^{-1}$ be the precision matrix of $\msf{z}$. Then, from the block matrix inversion formula,  $Q_{22} = \Sigma_{22} - \Sigma_{21} \Sigma_{11}^{-1} \Sigma_{12}$ and $Q_{21} = - Q_{22} \Sigma_{21} \Sigma_{11}^{-1} = -Q_{22} A$. Partition $\tilde{\beta}^* = (\tilde{\beta}_1^*; \tilde{\beta}_2^*)$ the same way as $\msf{z}$, and note that $\tilde{\beta}_2^* = Q_{21} m_1 + Q_{22} m_2 = - Q_{22} A m_1 + Q_{22} A m_1 = 0$. This proves that $\tilde{\beta}^*_j = 0$ whenever $j \in \delta$.  

To prove the remaining bit (and also provide an alternative argument for the previous part), we proceed as follows. Note that for any $p \times p$ invertible matrix $B$, we can write $\tilde{\beta}^* = (B \Sigma)^{-1} (Bm)$. Choose $B$ to be the matrix with blocks $B_{11} = I_{p_\delta}$, $B_{12} = 0_{p_\delta \times p-p_\delta}$, $B_{21} = -A$, and $B_{22} = I_{p-p_\delta}$. Then, it follows that $(Bm)_1 = m_1$, and $(Bm)_2 = m_2 - A m_1 = 0$. Moreover, $(B \Sigma)_{11} = \Sigma_{11}$, $(B \Sigma)_{12} = \Sigma_{12}$, $(B \Sigma)_{21} = \Sigma_{21} - A \Sigma_{11} = 0$, and $(B \Sigma)_{22} = \Sigma_{22.1} :\,= \Sigma_{22} - \Sigma_{21} \Sigma_{11}^{-1} \Sigma_{12}$. By the block inverse formula, we have $[(B \Sigma)^{-1}]_{11} = \Sigma_{11}^{-1}$ and $[(B \Sigma)^{-1}]_{21} = 0$. Putting everything together, we get $\tilde{\beta}^*_1 = \Sigma_{11}^{-1} m_1$ and $\tilde{\beta}_2^* = 0$. A similar argument as in the beginning of this proof shows that $\tilde{\beta}^*_\delta = \Sigma_{11}^{-1} m_1$. This concludes the proof that $\tilde{\beta}^* = (\tilde{\beta}^*_1; \tilde{\beta}^*_2) = (\tilde{\beta}_\delta^*; 0)$.

\newpage
\section{Auxiliary results for Theorem \ref{thm:bf_ala} and Corollary \ref{cor:bf_ala_unknownphi}}
\label{sec:auxiliary_bfala}

For completeness we state and prove two basic results that are helpful in the proofs of our main results.
Lemma \ref{lem:waldstat_nested} is a well-know results. It expresses the difference between the explained sum of squares by two nested models $\gamma' \subset \gamma$ in terms of the regression parameters obtained under the larger model, after suitably orthogonalizing its design matrix.

Lemma \ref{lem:phi0_convergence} establishes the convergence in probability of the conditional MLE $\phi_0$ to a finite optimal $\phi_0^*$, as $n \rightarrow \infty$, and is a simple adaptation of Lemma 5.10 in \cite{vandervaart:1998}).

\begin{lemma}
  Let $y \in \mathbb{R}^n$ and $X_\gamma= (X_{\gamma'}, X_{\gamma \setminus \gamma'})$ be an $n$ times $p_\gamma$ matrix.
  Let $\hat{\beta}_\gamma= (X_\gamma^T X_\gamma)^{-1} X_\gamma^T y$
  and $\hat{\beta}_{\gamma'}= (X_{\gamma'}^T X_{\gamma'})^{-1} X_{\gamma'}^T y$
  the least-squares estimates associated to $\gamma$ and $\gamma'$, respectively.
Then
$$  \hat{\beta}_{\gamma}^T X_{\gamma}^TX_{\gamma} \hat{\beta}_{\gamma} - \hat{\beta}_{\gamma'}^T X_{\gamma'}^TX_{\gamma'} \hat{\beta}_{\gamma'}
=\tilde{\beta}_{\gamma \setminus \gamma'}^T \tilde{X}_{\gamma \setminus \gamma'}^T\tilde{X}_{\gamma \setminus \gamma'} \tilde{\beta}_{\gamma \setminus \gamma'}=
y^T \tilde{X}_{\gamma \setminus \gamma'} (\tilde{X}_{\gamma \setminus \gamma'}^T \tilde{X}_{\gamma \setminus \gamma'})^{-1} \tilde{X}_{\gamma \setminus \gamma'}^T y
$$
where $\tilde{\beta}_{\gamma \setminus \gamma'}= (\tilde{X}_{\gamma \setminus \gamma'}^T\tilde{X}_{\gamma \setminus \gamma'})^{-1} \tilde{X}_{\gamma \setminus \gamma'}^T y$
and $\tilde{X}_{\gamma \setminus \gamma'}= X_{\gamma \setminus \gamma'} - X_{\gamma'}(X_{\gamma'}^TX_{\gamma'})^{-1}X_{\gamma'}^TX_{\gamma \setminus \gamma'}$.
\label{lem:waldstat_nested}
\end{lemma}

\begin{proof}
First note that $\tilde{X}_{\gamma \setminus \gamma'}$ is orthogonal to the projection of $X_{\gamma \setminus \gamma'}$ on $X_{\gamma'}$, hence
$$X_{\gamma'}^T\tilde{X}_{\gamma \setminus \gamma'}=
X_{\gamma'}^T X_{\gamma \setminus \gamma'} - X_{\gamma'}^T X_{\gamma'}(X_{\gamma'}^TX_{\gamma'})^{-1}X_{\gamma'}^TX_{\gamma \setminus \gamma'}
=0.
$$
Further, simple algebra shows that the least-squares prediction $X_\gamma \hat{\beta}_\gamma$
from regressing $y$ on $X_\gamma$
is equal to the least-squares prediction from regressing $y$ on $(X_{\gamma'},\tilde{X}_{\gamma \setminus \gamma'})$,
that is
$X_\gamma \hat{\beta}_\gamma= X_{\gamma'} \hat{\beta}_{\gamma'} + \tilde{X}_{\gamma \setminus \gamma} \tilde{\beta}_{\gamma \setminus \gamma}$,
where note that the coefficient for $X_{\gamma'}$ is $\hat{\beta}_{\gamma'}$, since $X_\gamma^T \tilde{X}_{\gamma \setminus \gamma'}=0$.
Therefore
\begin{align}
  &\hat{\beta}_{\gamma}^T X_{\gamma}^TX_{\gamma} \hat{\beta}_{\gamma} - \hat{\beta}_{\gamma'}^T X_{\gamma'}^TX_{\gamma'} \hat{\beta}_{\gamma'}=
(\hat{\beta}_{\gamma'}^T X_{\gamma'}^TX_{\gamma'} \hat{\beta}_{\gamma'}
  + \tilde{\beta}_{\gamma \setminus \gamma'}^T \tilde{X}_{\gamma \setminus \gamma'}^T\tilde{X}_{\gamma \setminus \gamma'} \tilde{\beta}_{\gamma \setminus \gamma'}) - \hat{\beta}_{\gamma'}^T X_{\gamma'}^TX_{\gamma'} \hat{\beta}_{\gamma'}
  \nonumber \\
&=\tilde{\beta}_{\gamma \setminus \gamma'}^T \tilde{X}_{\gamma \setminus \gamma'}^T\tilde{X}_{\gamma \setminus \gamma'} \tilde{\beta}_{\gamma \setminus \gamma'},
\nonumber
\end{align}
as we wished to prove.
\end{proof}

\begin{lemma}
Let
\begin{align}
 \phi_0&= \arg\max_{\phi >0} n^{-1} \log p(y \mid \beta=0, \phi)= - \frac{b(0)}{\phi} + \frac{1}{n} \sum_{i=1}^n c(y_i,\phi),
\nonumber \\
 \phi_0^*&= \arg\max_{\phi >0} - \frac{b(0)}{\phi} + E_{F^*}\left[ c(y_i,\phi)\right],
\nonumber
\end{align}
and assume Conditions (C1) and (C4). Then $\phi_0 \stackrel{P}{\longrightarrow} \phi_0^*$ as $n \rightarrow \infty$, when data truly arise from $F^*$.
\label{lem:phi0_convergence}
\end{lemma}

\begin{proof}
Denote the derivatives of the observed and expected log-likelihoods with respect to $\phi$, evaluated at $\beta=0$, by
\begin{align}
 \psi_n(\phi)&= \nabla_\phi \frac{1}{n} \log p(y \mid \beta=0,\phi)= \frac{b(0)}{\phi^2} + \frac{1}{n} \sum_{i=1}^n \nabla_\phi c(y_i,\phi)
\\ \nonumber
 \psi(\phi) &= \frac{b(0)}{\phi^2} + E_{F^*} \nabla_\phi c(y_i,\phi).
\nonumber
\end{align}
Recall that $\phi_0$ is the conditional MLE defined to satisfy $\psi_n(\phi_0)=0$, $\phi_0^*$ satisfies $\psi(\phi_0^*)=0$, and both $(\phi_0,\phi_0^*)$ are unique by assumption.
Let $(\phi_0^* - \epsilon, \phi_0^* + \epsilon)$ be an $\epsilon$ neighborhood around $\phi_0^*$.
For sufficiently small $\epsilon$, the event that $\Psi_n(\phi_0^* - \epsilon) < 0$ and $\Psi_n(\phi_0^* + \epsilon)>0$ implies that
$\phi_0 \in (\phi_0^* - \epsilon, \phi_0^* + \epsilon)$, hence
\begin{align}
P_{F^*} \left( \Psi_n(\phi_0^* - \epsilon) < 0, \Psi_n(\phi_0^* + \epsilon) \right) \leq
P_{F^*} \left(\phi_0 \in (\phi_0^* - \epsilon, \phi_0^* + \epsilon) \right).
\nonumber
\end{align}

We shall show that the left-hand size converges to 1 as $n \rightarrow \infty$ for any $\epsilon >0$, implying that the right-hand side also converges to 1 and hence that $\phi_0 \stackrel{P}{\longrightarrow} \phi_0^*$.
First note that from (C1) we have that $(y_i,z_i) \sim F^*$ independently, hence by the weak law of large numbers $\psi_n(\phi) \stackrel{P}{\longrightarrow} \psi(\phi)$ for any fixed $\phi$.
This implies that $\psi_n(\phi_0^* - \epsilon) \stackrel{P}{\longrightarrow} \psi(\phi_0^* - \epsilon) < 0$
and $\psi_n(\phi_0^* + \epsilon) \stackrel{P}{\longrightarrow} \psi(\phi_0^* + \epsilon) > 0$, thus
\begin{align}
\lim_{n \rightarrow \infty} P_{F^*} \left( \Psi_n(\phi_0^* - \epsilon) < 0, \Psi_n(\phi_0^* + \epsilon) \right)=0.
\nonumber
\end{align}
\end{proof}

\section{Proof of Theorem \ref{thm:bf_ala}}
\label{sec:proof_bf_ala}

Let $\tilde{y}= (y - b(0) \mathbbm{1})/ b''(0)$, $\tilde{\beta}_\gamma= (Z_\gamma^T Z_\gamma)^{-1} Z_\gamma^T \tilde{y}$ and
assume the dispersion parameter $\phi$ is known.
Recall that the ALA Bayes factor for a local prior density $\tilde{p}^L(\beta \mid \phi)$, which we assume to be strictly positive for any $\beta$, is
\begin{align}
  \tilde{B}^L_{\gamma \tgamma^*}=
  \left( \frac{2 \pi \phi}{b''(0)} \right)^{(p_\gamma - p_{\tgamma^*})/2}
  \frac{|Z_{\tgamma^*}^T Z_{\tgamma^*}|^{\frac{1}{2}} p(\tilde{\beta}_\gamma \mid \phi, \gamma)}
  {|Z_{\gamma}^T Z_{\gamma}|^{\frac{1}{2}} p(\tilde{\beta}_{\tgamma^*} \mid \phi, \tgamma^*)}
  e^{\frac{b''(0)}{2\phi} (\tilde{\beta}_\gamma^T Z_\gamma^T Z_\gamma \tilde{\beta}_\gamma - \tilde{\beta}_{\tgamma^*}^T Z_{\tgamma^*}^T Z_{\tgamma^*} \tilde{\beta}_{\tgamma^*})}
\label{eq:normal_bf_knownphi}
\end{align}
and that for the gMOM prior it is
\begin{align}
\tilde{B}^N_{\gamma \tgamma^*}= \tilde{B}^L_{\gamma, \tgamma^*}
\frac{\prod_{\gamma_j=1} \mbox{tr}(V_j S_{\gamma j}) / p_j + \tilde{\beta}_{\gamma j}^T V_j \tilde{\beta}_{\gamma j}/ (\phi p_j)}
{\prod_{\tgamma^*_j=1} \mbox{tr}(V_j S_{\tgamma^* j}) / p_j + \tilde{\beta}_{\gamma j}^T V_j \tilde{\beta}_{\gamma j}/ (\phi p_j)}
\label{eq:gmom_bf_knownphi},
\end{align}
where $S_{\gamma,j}$ is the sub-matrix of $(Z_\gamma^T D_\gamma Z_\gamma)^{-1}$ corresponding to $\tilde{\beta}_{\gamma j}$
and $V_j= Z_j^T Z_j (p_j+2)/(n p_j g_N)$.

The proof strategy for $\tilde{B}^L_{\gamma, \tgamma^*}$
is to show that $Z^T_\gamma Z_\gamma/n$ and $Z^T_{\tgamma^*} Z_{\tgamma^*}/n$ converge to positive-definite matrices,
show that $\tilde{\beta}_\gamma$ is asymptotically normally-distributed,
to then prove that $p^L(\tilde{\beta}_\gamma \mid \phi, \gamma)$  and $p^L(\tilde{\beta}_{\tgamma^*} \mid \phi, \gamma)$ 
converge in probability to positive constants and characterize the sampling distribution of
$\tilde{\beta}_\gamma^T Z_\gamma^T Z_\gamma \tilde{\beta}_\gamma - \tilde{\beta}_{\tgamma^*}^T Z_{\tgamma^*}^T Z_{\tgamma^*} \tilde{\beta}_{\tgamma^*}$.
The proof for $\tilde{B}^N_{\gamma, \tgamma^*}$ then follows easily by characterizing the second term on the right-hand side of \eqref{eq:gmom_bf_knownphi}.

Recall that we denoted the parameter value minimizing mean squared error by
\begin{align}
 \tilde{\beta}_\gamma^*&=   \arg\min_{\beta_\gamma} E_{F^*} ||\tilde{y} - Z_\gamma \beta_\gamma||_2^2=
  [E_{F^*}(Z_\gamma^T Z_\gamma)]^{-1} E_{F^*}[ Z_\gamma^T \tilde{y} ].
  \nonumber
\end{align}

To prove the asymptotic Normality of $\tilde{\beta}_\gamma$ we use the results in~\cite{hjort:2011} (Section 3D)
for misspecified least squares, since $\tilde{\beta}_\gamma$ is the least-squares estimate for regressing $\tilde{y}$ on $Z_\gamma$.
The result in~\cite{hjort:2011} (Section 3D) requires that
$$
\tilde{L}_\gamma=E_{F^*(z_i)} (z_{i \gamma} z_{i \gamma}^T [E_{F^*}(\tilde{y}_i \mid z_{i \gamma}) - z_{i \gamma}^T \tbeta_\gamma^*]^2)
$$
is a finite matrix, and that
$Z_\gamma^TZ_\gamma/n$ and $Z_\gamma^T \Sigma_{\tilde{y} \mid z, \gamma} Z_\gamma /n$
converge to finite positive-definite matrices as $n \rightarrow \infty$ under $F^*$.
Without loss of generality assume that $Z_\gamma$ has zero column means.
Since $\Sigma_{y | z, \gamma}= \mbox{Cov}_{F^*}(y \mid Z_\gamma)$ and $L_\gamma$ are finite and positive-definite by assumption (C2),
then so are $\Sigma_{\tilde{y} | z, \gamma}$ and $\tilde{L}_\gamma$.
Also by assumption $\Sigma_{z \gamma}$ is a finite matrix and $z_1,\ldots,z_n$ are independent realizations from $F^*$, hence
$Z^TZ/n \stackrel{P}{\longrightarrow} \Sigma_{z \gamma}$
and $Z^T \Sigma_{\tilde{y} \mid z} Z /n=
\sum_{i=1}^n z_i z_i^T V_{F^*}(\tilde{y}_i \mid z_i)/n
\stackrel{P}{\longrightarrow} \Sigma_{\tilde{y} \mid z, \gamma}^{1/2} \Sigma_{z \gamma} \Sigma_{\tilde{y} \mid z, \gamma}^{1/2}$
by the weak law of large numbers.
Then~\cite{hjort:2011} (Section 3D) gives that
\begin{align}
  \sqrt{n} (\tilde{\beta}_\gamma - \tilde{\beta}_\gamma^*) \stackrel{D}{\longrightarrow}
  \mathcal{N}(0, V_\gamma),
\nonumber
\end{align}
where $V_\gamma= \Sigma_{z, \gamma}^{-1} (\Sigma_{\tilde{y} \mid z, \gamma}^{1/2} \Sigma_{z, \gamma} \Sigma_{\tilde{y} \mid z, \gamma}^{1/2} + \tilde{L}_\gamma) \Sigma_{z, \gamma}^{-1}$.
This implies $\tilde{\beta}_\gamma \stackrel{P}{\longrightarrow} \tilde{\beta}_\gamma^*$, hence by the continuous mapping theorem
\begin{align}
  \frac{\tilde{B}^L_{\gamma \tgamma^*}}
  {n^{(p_{\tgamma^*} - p_{\gamma})/2}}
  \left( \frac{b''(0)}{2 \pi \phi} \right)^{(p_\gamma - p_{\tgamma^*})/2}
  \frac{|\Sigma_{z, \gamma}|^{\frac{1}{2}} p(\tilde{\beta}_{\tgamma^*}^* \mid \phi, \tgamma^*)}
  {|\Sigma_{z,\tgamma^*}|^{\frac{1}{2}} p(\tilde{\beta}_\gamma^* \mid \phi, \gamma)}
  e^{-\frac{b''(0)}{2\phi} (\tilde{\beta}_\gamma^T Z_\gamma^T Z_\gamma \tilde{\beta}_\gamma - \tilde{\beta}_{\tgamma^*}^T Z_{\tgamma^*}^T Z_{\tgamma^*} \tilde{\beta}_{\tgamma^*})}
\stackrel{P}{\longrightarrow} 1,
\label{eq:bf_ala_asymp}
\end{align}
where
$$
  \left( \frac{b''(0)}{2 \pi \phi} \right)^{(p_\gamma - p_{\tgamma^*})/2}
  \frac{|\Sigma_{z, \gamma}|^{\frac{1}{2}} p(\tilde{\beta}_{\tgamma^*}^* \mid \phi, \tgamma^*)}
  {|\Sigma_{z,\tgamma^*}|^{\frac{1}{2}} p(\tilde{\beta}_\gamma^* \mid \phi, \gamma)}
$$
is a finite non-zero constant that does not depend on $n$.

To conclude the proof it suffices to characterize the exponential term in~\eqref{eq:bf_ala_asymp},
separately for Part (i) and (ii) in Theorem \ref{thm:bf_ala}.
In both cases we shall use that, from Lemma~\ref{lem:waldstat_nested}, for any two nested models $\gamma' \subset \gamma$ in the sense that $Z_{\gamma'}$ is a submatrix of $Z_\gamma$, 
\begin{align}
  \tilde{\beta}_\gamma^T Z_\gamma^T Z_\gamma \tilde{\beta}_\gamma - \tilde{\beta}_{\gamma'}^T Z_{\gamma'}^T Z_{\gamma'} \tilde{\beta}_{\gamma'}=
  b_{\gamma \setminus \gamma'}^T X_{\gamma \setminus \gamma'}^T X_{\gamma \setminus \gamma'} b_{\gamma \setminus \gamma'},
\label{eq:teststat_knownphi}
\end{align}
where
$X_{\gamma \setminus \gamma'}= Z_{\gamma \setminus \gamma'} - Z_{\gamma'}(Z_{\gamma'}^T Z_{\gamma'})^{-1} Z_{\gamma'}^TZ_{\gamma \setminus \gamma'}$
is orthogonal to the projection of $Z_\gamma$ onto $Z_{\gamma'}$
and $b_{\gamma \setminus \gamma'}= (X_{\gamma \setminus \gamma'}^TX_{\gamma \setminus \gamma'})^{-1} X_{\gamma \setminus \gamma'}^T \tilde{y}$
the least-squares estimate regressing $\tilde{y}$ on $X_{\gamma \setminus \gamma'}$.
Our earlier argument to prove the asymptotic normality of $\tilde{\beta}_\gamma$ also implies that
\begin{align}
  \sqrt{n}(b_{\gamma \setminus \gamma'} - \tilde{\beta}_{\gamma \setminus \gamma'}^*) \stackrel{D}{\longrightarrow} \mathcal{N}(0,S_{\gamma \setminus \gamma'}),
\label{eq:limitnormal_extraparameters}
\end{align}
for finite positive-definite $S_{\gamma \setminus \gamma'}$, where
\begin{align}
 \tilde{\beta}_{\gamma \setminus \gamma'}^*= [E_{F^*}(X_{\gamma \setminus \gamma'})^T X_{\gamma \setminus \gamma'} ]^{-1}
E_{F^*}[ X_{\gamma \setminus \gamma'}^T \tilde{y}].
\nonumber
\end{align}

\subsection*{Part (i), ${\bm \tilde{B}^L_{\gamma, \tgamma^*}}$}
Suppose that $\tgamma^* \subset \gamma$. Then, by definition of $\tgamma^*$, we have that $\tilde{\beta}_{\gamma \setminus \tgamma^*}^*=0$.
Set $\gamma'=\tgamma^*$, denote by $W= b_{\gamma \setminus \tgamma^*}^T S_{\gamma \setminus \gamma'}^{-1} b_{\gamma \setminus \tgamma^*}$ and note that
\begin{align}
 \lambda_1 W \leq
b_{\gamma \setminus \tgamma^*}^T X_{\gamma \setminus \tgamma^*}^T X_{\gamma \setminus \tgamma^*} S_{\gamma \setminus \gamma'} S_{\gamma \setminus \gamma'}^{-1} b_{\gamma \setminus \tgamma^*}
\leq \lambda_2 W
\label{eq:waldstat_bound_eigenvals}
\end{align}
where $(\lambda_1,\lambda_2)$ are the smallest and largest eigenvalues of
$X_{\gamma \setminus \tgamma^*}^T X_{\gamma \setminus \tgamma^*} S_{\gamma \setminus \gamma'}$.
Since $\sqrt{n} b_{\gamma \setminus \tgamma^*} \stackrel{D}{\longrightarrow} \mathcal{N}(0,S_{\gamma \setminus \gamma'})$, by Slutsky's theorem
$W \stackrel{D}{\longrightarrow} \chi^2_{p_\gamma - p_{\tgamma^*}}$ as $n \rightarrow \infty$.
Thus the exponential term in~\eqref{eq:bf_ala_asymp} is $O_p(1)$ and
$$
  \tilde{B}^L_{\gamma \tgamma^*}= n^{(p_{\tgamma^*} - p_{\gamma})/2} O_p(1),
$$
as we wished to prove.

\subsection*{Part (ii), ${\bm \tilde{B}^L_{\gamma, \tgamma^*}}$}
Suppose now that $\tgamma^* \not\subseteq \gamma$.
Then by Lemma \ref{lem:waldstat_nested}
\begin{align}
  \tilde{\beta}_\gamma^T Z_\gamma^T Z_\gamma \tilde{\beta}_\gamma - \tilde{\beta}_{\tgamma^*}^T Z_{\tgamma^*}^T Z_{\tgamma^*} \tilde{\beta}_{\tgamma^*}
  \pm \tilde{\beta}_{\gamma'}^T Z_{\gamma'}^T Z_{\gamma'} \tilde{\beta}_{\gamma'}
  =
  b_{\gamma' \setminus \tgamma^*}^T X_{\gamma' \setminus \tgamma^*}^T X_{\gamma' \setminus \tgamma^*} b_{\gamma' \setminus \tgamma^*}
-  b_{\gamma \setminus \gamma'}^T X_{\gamma \setminus \gamma'}^T X_{\gamma \setminus \gamma'} b_{\gamma \setminus \gamma'},
\nonumber
\end{align}
where we set $\gamma'= \tgamma^* \cup \gamma$ to be the model such that $Z_{\gamma'}$ includes all columns from $Z_{\tgamma^*}$ and $Z_{\gamma}$.
Note that Lemma~\ref{lem:waldstat_nested} applies since both
$\tgamma^* \subseteq \gamma'$ and $\gamma \subseteq \gamma'$.
Now, $\tgamma^* \subseteq \gamma'$ implies that $\tilde{\beta}^*_{\gamma' \setminus \tgamma^*}=0$,
applying the same argument as in~\eqref{eq:waldstat_bound_eigenvals} gives that
$b_{\gamma' \setminus \tgamma^*}^T X_{\gamma' \setminus \tgamma^*}^T X_{\gamma' \setminus \tgamma^*} b_{\gamma' \setminus \tgamma^*}= O_p(1)$.
Finally consider the term $b_{\gamma \setminus \gamma'}^T X_{\gamma \setminus \gamma'}^T X_{\gamma \setminus \gamma'} b_{\gamma \setminus \gamma'}$.
Analogously to~\eqref{eq:waldstat_bound_eigenvals}
let $W= b_{\gamma \setminus \gamma'}^T S_{\gamma \setminus \gamma'}^{-1} b_{\gamma \setminus \gamma'}$
and note that
\begin{align}
  \lambda_1 W \leq
b_{\gamma \setminus \gamma'}^T X_{\gamma \setminus \gamma'}^T X_{\gamma \setminus \gamma'} S_{\gamma \setminus \gamma'} S_{\gamma \setminus \gamma'}^{-1} b_{\gamma \setminus \gamma'}
\leq \lambda_2 W
\nonumber
\end{align}
where now $(\lambda_1,\lambda_2)$ denote the smallest and largest eigenvalues of
$X_{\gamma \setminus \gamma'}^T X_{\gamma \setminus \gamma'} S_{\gamma \setminus \gamma'}$.
Now from~\eqref{eq:limitnormal_extraparameters} we have that
$$
\sqrt{n} (b_{\gamma \setminus \gamma'} - \tilde{\beta}^*_{\gamma \setminus \gamma'}) \stackrel{D}{\longrightarrow} \mathcal{N}(0, S_{\gamma \setminus \gamma'}),
$$
where $\tilde{\beta}^*_{\gamma \setminus \gamma'} \neq 0$. Therefore by Slutky's theorem
$W \stackrel{D}{\longrightarrow}  \chi^2_{p_{\gamma'} - p_\gamma}(l(\gamma',\gamma))$
with non-zero non-centrality parameter
$l(\gamma',\gamma)= n (\tilde{\beta}^*_{\gamma' \setminus \gamma})^T S_{\gamma' \setminus \gamma} \tilde{\beta}^*_{\gamma' \setminus \gamma} \neq 0$.
Note that $l(\gamma',\gamma)$ can be interpreted as the increase in expected sum of squares in $\tilde{y}$ explained by $\gamma'= \tgamma^* \cup \gamma$, relative to that explained by $\gamma$ alone.
Therefore, from~\eqref{eq:bf_ala_asymp} we obtain
\begin{align}
  \frac{1}{n} \log \tilde{B}^L_{\gamma \tgamma^*}&=
  \frac{b''(0)}{n \phi} [ O_p(1) - b_{\gamma \setminus \gamma'}^T X_{\gamma \setminus \gamma'}^T X_{\gamma \setminus \gamma'} b_{\gamma \setminus \gamma'}]
  <   \frac{b''(0)}{n \phi} [ O_p(1) - \lambda_1 W]
  \nonumber \\
  &\stackrel{P}{\longrightarrow}
  \phi^{-1} b''(0) \lambda_1 (\tilde{\beta}^*_{\gamma' \setminus \gamma})^T S_{\gamma' \setminus \gamma} \tilde{\beta}^*_{\gamma' \setminus \gamma} > 0,
\nonumber
\end{align}
as we wished to prove.

\subsection*{Part (i), ${\bm \tilde{B}^N_{\gamma, \tgamma^*}}$}

First note that 
$\mbox{tr}(V_j S_{\gamma j}) / p_j= O_p(1/n)$, since $S_j$ is a sub-matrix of $(Z_\gamma^T D_\gamma Z_\gamma)^{-1}$ and we showed earlier that $Z_\gamma^T D_\gamma Z_\gamma/n \stackrel{P}{\longrightarrow} \Sigma_{\gamma z}$, where $\Sigma_{\gamma z}$ is a strictly positive-definite matrix, and similarly $V_j= Z_j^T Z_j (p_j+2)/(n p_j g_N)$ converges to a finite positive-definite matrix.
The same reasoning implies that $\mbox{tr}(V_j S_{\tgamma^* j}) / p_j= O_p(1/n)$.

Second, we showed that 
$\tilde{\beta}_{\tgamma^* j} \stackrel{P}{\longrightarrow} \tilde{\beta}_{j}^* \neq 0$,
$\tilde{\beta}_{\gamma j} \stackrel{P}{\longrightarrow} \tilde{\beta}_{j}^* \neq 0$ for truly active $j \in \tgamma^*$,
and $\tilde{\beta}_{\gamma j}= O_p(1/n)$ for $j \in \gamma \setminus \tgamma^*$.
This immediately implies that
\begin{align}
&\frac{\prod_{\gamma_j=1} \mbox{tr}(V_j S_{\gamma j}) / p_j + \tilde{\beta}_{\gamma j}^T V_j \tilde{\beta}_{\gamma j}/ (\phi p_j)}
{\prod_{\tgamma^*_j=1} \mbox{tr}(V_j S_{\tgamma^* j}) / p_j + \tilde{\beta}_{\gamma j}^T V_j \tilde{\beta}_{\gamma j}/ (\phi p_j)}=
\nonumber \\
&\frac{\prod_{j \in \gamma \setminus \tgamma^*} O_p(1/n) + O_p(1/n) \prod_{\tgamma^*_j=1} O_p(1/n) + [(\tilde{\beta}_{j}^*)^T V_j \tilde{\beta}_{j}^* + o_p(1)]/ (\phi p_j)}
{\prod_{\tgamma^*_j=1} O_p(1/n) + [(\tilde{\beta}_{j}^*)^T V_j \tilde{\beta}_{j}^* + o_p(1)]/ (\phi p_j)}=
O_p(n^{-(p_\gamma - p_{\tgamma^*})}),
\nonumber
\end{align}
since $(\tilde{\beta}_{j}^*)^T V_j \tilde{\beta}_{j}^*$ is a non-zero constant, as we wished to prove.

\subsection*{Part (ii), ${\bm \tilde{B}^N_{\gamma, \tgamma^*}}$}

Arguing as above,
\begin{align}
&\frac{\prod_{\gamma_j=1} \mbox{tr}(V_j S_{\gamma j}) / p_j + \tilde{\beta}_{\gamma j}^T V_j \tilde{\beta}_{\gamma j}/ (\phi p_j)}
{\prod_{\tgamma^*_j=1} \mbox{tr}(V_j S_{\tgamma^* j}) / p_j + \tilde{\beta}_{\gamma j}^T V_j \tilde{\beta}_{\gamma j}/ (\phi p_j)}=
\frac{\prod_{\gamma_j=1} O_p(1/n) + \tilde{\beta}_{\gamma j}^T V_j \tilde{\beta}_{\gamma j}/ (\phi p_j)}
{\prod_{\tgamma^*_j=1} [(\tilde{\beta}_{j}^*)^T V_j \tilde{\beta}_{j}^* + o_p(1)]/ (\phi p_j)}=
O_p(n^{-|\tilde{\beta}^*_\gamma|_0} )
\nonumber
\end{align}
where $|\tilde{\beta}^*_\gamma|_0= \sum_{\gamma_j=1} \mbox{I}(\tilde{\beta}_{\gamma j}^*)$ is the number of parameters with optimal value equal to 0.
Hence
\begin{align}
\frac{1}{n} \log \tilde{B}^N_{\gamma, \tgamma^*}= \frac{1}{n} \log \tilde{B}^L_{\gamma, \tgamma^*} + o_p(1),
\nonumber
\end{align}
proving the desired result.

\section{Proof of Corollary \ref{cor:bf_ala_unknownphi}}
\label{sec:proof_bf_ala_unknownphi}

The proof strategy is as follows.
First, note that $\tilde{\phi}_\gamma= \phi_0 - H_0^{-1} g_0$ converges in probability to a finite non-zero constant,
since $\phi_0 \stackrel{P}{\longrightarrow} \phi_0^*$ by Lemma \ref{lem:phi0_convergence}, $g_0/n \stackrel{P}{\longrightarrow} g_0^*$ and $H_0/n \stackrel{P}{\longrightarrow} H_0^*$ for some fixed non-zero $g_0^*$, $H_0^*$.
Arguing as in the proof of Theorem \ref{thm:bf_ala} shows that
\begin{align}
\tilde{B}^N_{\gamma \tgamma^*}= \tilde{B}^L_{\gamma \tgamma^*}
\frac{\prod_{\gamma_j=1} \mbox{tr}(V_j S_{\gamma j})/p_j + \frac{\tilde{\beta}_{\gamma j}^T V_j \tilde{\beta}_{\gamma j}}{\tilde{\phi}_\gamma p_j}}
{\prod_{\tgamma^*_j=1} \mbox{tr}(V_j S_{\tgamma^* j})/p_j + \frac{\tilde{\beta}_{\tgamma^* j}^T V_j \tilde{\beta}_{\tgamma^* j}}{\tilde{\phi}_\tgamma^* p_j}}
\nonumber
\end{align}
is equal to $\tilde{B}^L_{\gamma \tgamma^*} O_p(n^{-(p_\gamma - p_{\tgamma^*})})$ when $\tgamma^* \in \gamma$,
and equal to $\tilde{B}^L_{\gamma \tgamma^*} O_p(n^{-|\tilde{\beta}_\tgamma^*|_0})$ when $\tgamma^* \in \gamma$.
These expressions for $\tilde{B}^N_{\gamma \tgamma^*}$ are the same as those in Theorem \ref{thm:bf_ala} where $\phi$ was assumed known,
hence it suffices to show that $\tilde{B}^L_{\gamma \tgamma^*}$ also attains the same rates as in Theorem \ref{thm:bf_ala}.
To prove this, we use algebraic manipulation and auxiliary Lemma \ref{lem:phi0_convergence}
to show that the leading term in the Bayes factor exponent for unknown $\phi$ converges to a constant multiple of that in the known $\phi$ case, which immediately gives the desired result.

Specifically, recall that
\begin{align}
\tilde{B}^L_{\gamma,\gamma'}= 
e^{\frac{b''(0)}{2\phi_0}[ 
t_\gamma \tilde{\beta}_\gamma^T Z_\gamma^T Z_\gamma \tilde{\beta}_\gamma
- t_{\gamma'} \tilde{\beta}_{\gamma'}^T Z_{\gamma'}^T Z_{\gamma'} \tilde{\beta}_{\gamma'}
 ]}
(2 \pi)^{\frac{p_\gamma-p_{\gamma'}}{2}}
\frac{|H_{\gamma' 0}|^{\frac{1}{2}} p^L(\tilde{\beta}_\gamma, \tilde{\phi}_\gamma \mid \gamma)}{|H_{\gamma 0}|^{\frac{1}{2}} p^L(\tilde{\beta}_{\gamma'}, \tilde{\phi}_{\gamma'} \mid \gamma)} 
\nonumber
\end{align}
where
\begin{align}
H_{\gamma 0}&= \frac{b''(0)}{\phi_0} \begin{pmatrix} Z_\gamma^TZ_\gamma & -Z_\gamma^T \tilde{y}/\phi_0 \\ 
-\tilde{y}^T Z_\gamma/\phi_0 & s(\phi_0) \end{pmatrix}
\nonumber \\
t_\gamma&= 1 + \frac{\tilde{\beta}_\gamma^T Z_\gamma^TZ_\gamma \tilde{\beta}_\gamma}{\phi_0^2 (s(\phi_0) - \tilde{\beta}_\gamma^T Z_\gamma^TZ_\gamma \tilde{\beta}_\gamma)}
\nonumber \\
s(\phi_0)&= [2n b(0)/\phi_0^2 + \phi_0 \sum_{i=1}^n \nabla_{\phi \phi}^2 c(y_i,\phi_0)]/b''(0).
\nonumber
\end{align}

Arguing as in the proof of Theorem \ref{thm:bf_ala} gives that
\begin{align}
 \frac{|H_{\gamma' 0}|^{\frac{1}{2}}}{|H_{\gamma 0}|^{\frac{1}{2}}}=
n^{\frac{(p_{\gamma'} - p_\gamma)}{2}} \frac{|n^{-1} H_{\gamma' 0}|^{\frac{1}{2}}}{|n^{-1} H_{\gamma 0}|^{\frac{1}{2}}}
=
n^{\frac{(p_{\gamma'} - p_\gamma)}{2}} \left(\frac{\phi_0}{b''(0)}\right)^{\frac{p_\gamma - p_{\gamma'}}{2}}
\left[ \frac{|\Sigma_{z,\gamma'}|^{\frac{1}{2}}}{|\Sigma_{z,\gamma}|^{\frac{1}{2}}} + o_p(1) \right],
\nonumber
\end{align}
where $\Sigma_{z,\gamma'}$ and $\Sigma_{z,\gamma}$ are fixed and strictly positive-definite matrices.
This gives an expression analogous to \eqref{eq:bf_ala_asymp}, specifically
\begin{align}
  \frac{\tilde{B}^L_{\gamma \gamma'}}
  {n^{(p_{\gamma'} - p_{\gamma})/2}}
  \left( \frac{b''(0)}{2 \pi \phi_0} \right)^{\frac{p_\gamma - p_{\gamma'}}{2}}
  \frac{|\Sigma_{z, \gamma}|^{\frac{1}{2}} p(\tilde{\beta}_{\gamma'}^*, \tilde{\phi}_{\gamma'}^* \mid \gamma')}
  {|\Sigma_{z,\gamma'}|^{\frac{1}{2}} p(\tilde{\beta}_\tgamma^*, \tilde{\phi}_\tgamma^* \mid \gamma)}
  e^{-\frac{b''(0)}{2\phi_0} (t_\gamma \tilde{\beta}_\gamma^T Z_\gamma^T Z_\gamma \tilde{\beta}_\gamma - t_{\gamma'} \tilde{\beta}_{\gamma'}^T Z_{\gamma'}^T Z_{\gamma'} \tilde{\beta}_{\gamma'})}
\stackrel{P}{\longrightarrow} 1,
\label{eq:bf_ala_asymp_unknownphi}
\end{align}
where
$$
  \left( \frac{b''(0)}{2 \pi \phi_0} \right)^{(p_\gamma - p_{\gamma'})/2}
  \frac{|\Sigma_{z, \gamma}|^{\frac{1}{2}} p(\tilde{\beta}_{\gamma'}^*, \tilde{\phi}_{\gamma'} \mid \gamma')}
  {|\Sigma_{z,\gamma'}|^{\frac{1}{2}} p(\tilde{\beta}_\tgamma^*, \tilde{\phi}_\gamma \mid \gamma)}
$$
converges in probability to a finite non-zero constant that does not depend on $n$,
since $\phi_0 \stackrel{P}{\longrightarrow} \phi_0^*$ by Lemma \ref{lem:phi0_convergence} and $p(\beta_\gamma, \phi_\gamma \mid \gamma)>0$ is continuous by Condition (C3).

Now, set $\gamma'= \tgamma^*$. The term in the exponent of \eqref{eq:bf_ala_asymp_unknownphi} can be written
\begin{align}
t_\gamma \tilde{\beta}_\gamma^T Z_\gamma^T Z_\gamma \tilde{\beta}_\gamma - t_{\gamma'} \tilde{\beta}_{\gamma'}^T Z_{\gamma'}^T Z_{\gamma'} \tilde{\beta}_{\gamma'}=
 t_\gamma (\tilde{\beta}_\gamma^T Z_\gamma^T Z_\gamma \tilde{\beta}_\gamma - \tilde{\beta}_{\tgamma^*}^T Z_{\tgamma^*}^T Z_{\tgamma^*} \tilde{\beta}_{\tgamma^*})
+ (t_\gamma - t_{\tgamma^*}) \tilde{\beta}_{\tgamma^*}^T Z_{\tgamma^*}^T Z_{\tgamma^*} \tilde{\beta}_{\tgamma^*}
\nonumber
\end{align}
To complete the proof we show that the second term in the right-hand side is $O_p(1)$, then proceed as in the proof of Theorem \ref{thm:bf_ala} to conclude the proof.
Specifically, the asymptotic normality of $\tilde{\beta}_{\tgamma^*}$ implies that
$\tilde{\beta}_{\tgamma^*}^T Z_{\tgamma^*}^T Z_{\tgamma^*} \tilde{\beta}_{\tgamma^*}= O_p(1)$ (see proof of Theorem \ref{thm:bf_ala}).
Further, in
\begin{align}
 t_\gamma= 1 + \frac{\tilde{\beta}_\gamma^T Z_\gamma^TZ_\gamma \tilde{\beta}_\gamma/n}{\phi_0^2 (s(\phi_0)/n - \tilde{\beta}_\gamma^T Z_\gamma^TZ_\gamma \tilde{\beta}_\gamma/n)}
\nonumber
\end{align}
the term 
$\tilde{\beta}_\gamma^T Z_\gamma^TZ_\gamma \tilde{\beta}_\gamma/n \stackrel{P}{\longrightarrow} (\tilde{\beta}_\tgamma^*)^T \Sigma_{\gamma, z} \tilde{\beta}_\tgamma^*>0$ converges in probability to a finite constant by the strong law of large numbers if $\tilde{\beta}_\tgamma^* \neq 0$, and $\tilde{\beta}_\gamma^T Z_\gamma^TZ_\gamma \tilde{\beta}_\gamma/n=O_p(1/n)$ if $\tilde{\beta}_\tgamma^* =0$.
Further,
\begin{align}
\frac{s(\phi_0)}{n}= \frac{1}{b''(0)} \left[ \frac{2 b(0)}{\phi_0^2} + \frac{\phi_0}{n} \sum_{i=1}^n \nabla_{\phi \phi}^2 c(y_i,\phi_0) \right]
\stackrel{P}{\longrightarrow} 
\frac{1}{b''(0)} \left[ \frac{2 b(0)}{\phi_0^2} + \phi_0 E_{F^*} [\nabla_{\phi \phi}^2 c(y_i,\phi_0) ] \right] >0.
\nonumber
\end{align}
Hence by the continuous mapping theorem $t_\gamma$ converges in probability to a finite constant, so does $t_\gamma - t_{\tgamma^*}$,
implying that 
$
(t_\gamma - t_{\tgamma^*}) \tilde{\beta}_{\tgamma^*}^T Z_{\tgamma^*}^T Z_{\tgamma^*} \tilde{\beta}_{\tgamma^*}= O_p(1)
$
and
 \begin{align}
   e^{-\frac{b''(0)}{2\phi_0} (t_\gamma \tilde{\beta}_\gamma^T Z_\gamma^T Z_\gamma \tilde{\beta}_\gamma - t_{\gamma'} \tilde{\beta}_{\gamma'}^T Z_{\gamma'}^T Z_{\gamma'} \tilde{\beta}_{\gamma'})}=
e^{-\frac{b''(0) t_\gamma}{2\phi_0} (\tilde{\beta}_\gamma^T Z_\gamma^T Z_\gamma \tilde{\beta}_\gamma - \tilde{\beta}_{\tgamma^*}^T Z_{\tgamma^*}^T Z_{\tgamma^*} \tilde{\beta}_{\tgamma^*})} \times O_p(1)
\nonumber
\end{align}
The proof is completed by noting that $b''(0) t_\gamma/ \phi_0$ converges in probability to a constant and proceeding as in the proof of Theorem \ref{thm:bf_ala}, after \eqref{eq:bf_ala_asymp}, to characterize the exponential term 
$\tilde{\beta}_\gamma^T Z_\gamma^T Z_\gamma \tilde{\beta}_\gamma - \tilde{\beta}_{\tgamma^*}^T Z_{\tgamma^*}^T Z_{\tgamma^*} \tilde{\beta}_{\tgamma^*}$
for the cases $\tgamma^* \subset \gamma$ and $\tgamma^* \not\subseteq \gamma$.

\newpage
\section{Auxiliary results on sub-Gaussian random vectors}
\label{sec:auxiliary_bfala_highdim}

This section contains results on sub-Gaussian random vectors. Although their main interest for the current work is as auxiliary results used in the proof of Theorem \ref{thm:bf_ala_highdim}, they have some independent interest in bounding integrals related to Bayes factors that can bound posterior model probabilities in more general settings beyond the ALA framework considered in Theorem \ref{thm:bf_ala_highdim}.

We first recall the definition of a sub-Gaussian random vector and the basic property that linear combinations of sub-Gaussian vectors are sub-Gaussian (Lemma \ref{lem:subgaussian_lincomb}). 
Lemma \ref{lem:subgaussian_quadform} proves a useful property that certain quadratic forms of $n$-dimensional sub-Gaussian vectors can be re-expressed as a quadratic form of $d$-dimensional sub-Gaussian vectors.

A second set of results in Lemma \ref{lem:tail_quadform_subgaussian} and Lemma \ref{lem:lefttail_quadform_subgaussian} bound right tail and left tail probabilities (respectively) involving quadratic forms of sub-Gaussian random vectors.
Lemma \ref{lem:tail_quadform_subgaussian} is an adaptation of Theorem 1 in \cite{hsu:2012}. 
The interpretation is that, the probability that $s^Ts$ exceeds a threshold $\sigma^2 d q$ that is larger than its expectation $E(s^Ts) \leq \sigma^2 d$, decreases exponentially in $dq/2(1+k_0)$, where $k_0$ is of order $1/\sqrt{q}$ as $q$ grows, and hence $k_0$ becomes arbitrarily close to 0 for large $q$. The obtained bound is analogous to those for quadratic forms of Gaussian random variables. 
Regarding Lemma \ref{lem:lefttail_quadform_subgaussian}, the interpretation is that the probability that $s^Ts$ is less than a threshold $a$ that is smaller than the non-centrality parameter $\mu^T\mu$ decreases exponentially in $-\mu^T\mu/\sigma^2$.

The last set of results are Proposition \ref{prop:tailintegral_quadform_subgaussian} and Proposition \ref{prop:tailintegral_difquadform_subgaussian}. Both are novel results.
Proposition \ref{prop:tailintegral_quadform_subgaussian} provides a finite-sample bound for an integral involving tail probabilities as in Lemma \ref{lem:tail_quadform_subgaussian}. 
Finally, Proposition \ref{prop:tailintegral_difquadform_subgaussian} provides a bound for a similar integral involving differences between quadratic forms of central and non-central sub-Gaussian vectors. 
In both propositions one should think of $c$ as being a constant and $h$ being a large number, e.g. in ALA Bayes factors $h$ typically grows logarithmically in $n$.

If Bayes factors can be expressed (or bounded by) functions involving sub-Gaussian quadratic forms, as is the case for ALA to Bayes factors, then these propositions allow bounding the expectation of the posterior probability assigned to any given model (see Section \ref{sec:proof_bf_ala_highdim}). Such bounds on posterior probabilities in turn bound their $L_1$ convergence rate to 0, which in turn implies strong frequentist properties for the posterior probabilities, see \cite{rossell:2018}.

We first state the results, and subsequently provide the proofs.

\begin{defn}
A $d$-dimensional random vector $s=(s_1,\ldots,s_d)$ follows a sub-Gaussian distribution with parameters $\mu \in \mathbb{R}^d$ and $\sigma^2 > 0$, which we denote by $s \sim \mbox{SG}(\mu,\sigma^2)$ if and only if
$$
E\left[ \exp\{ \alpha^T(s - \mu) \} \right] \leq \exp\{ \alpha^T\alpha \sigma^2/2 \}
$$
for all $\alpha \in \mathbb{R}^d$.
\label{def:subgaussian}
\end{defn}

\begin{lemma}
Let $s \sim \mbox{SG}(\mu, \sigma^2)$ be a sub-Gaussian $d$-dimensional random vector, and $A$ be a $q \times d$ matrix.
Then $As \sim \mbox{SG}(A\mu, \lambda \sigma^2)$, where $\lambda$ is the largest eigenvalue of $A^TA$, or equivalently the largest eigenvalue of $A A^T$.
\label{lem:subgaussian_lincomb}
\end{lemma}

\begin{lemma}
Let $y \sim \mbox{SG}(\mu,\sigma^2)$ be an $n$-dimensional sub-Gaussian random vector.
Let $W= (y-a)^T X (X^T X)^{-1} X^T (y-a)$, where $X$ is an $n \times d$ matrix such that $X^TX$ is invertible.
Then $W= s^T s$, where $s \sim \mbox{SG}((X^TX)^{-1/2} X^T (\mu - a), \sigma^2)$ is $d$-dimensional.
\label{lem:subgaussian_quadform}
\end{lemma}

\begin{lemma} {\bf Central sub-Gaussian quadratic forms. Right-tail probabilities}
Let $s=(s_1,\ldots,s_d) \sim \mbox{SG}(0,\sigma^2)$. Then
\begin{enumerate}[leftmargin=*,label=(\roman*)]
\item For any $t>0$,
\begin{align}
 P \left( \frac{s^Ts}{\sigma^2} > d t [1 + \sqrt{2/t} + 1/t] \right) \leq \exp\left\{- \frac{dt}{2} \right\}.
\nonumber
\end{align}

\item For any $q>0$ and any $k_0$ such that $k_0 \geq \frac{(1+k_0)}{q} + \sqrt{\frac{2(1+k_0)}{q}}$,
\begin{align}
P \left( \frac{ s^T s}{\sigma^2} > dq \right) \leq \exp \left\{  - \frac{d q}{2(1+k_0)} \right\}
\nonumber
\end{align}

\item For any $q \geq 2 (1+\sqrt{2})^2$,
\begin{align}
P \left( \frac{ s^T s}{\sigma^2} > dq \right) \leq \exp \left\{  - \frac{d q}{2(1+k_0)} \right\}
\nonumber
\end{align}
where $k_0\geq \sqrt{2}(1+\sqrt{2})/\sqrt{q}$.
\end{enumerate}
\label{lem:tail_quadform_subgaussian}
\end{lemma}

\begin{lemma} {\bf Non-central sub-Gaussian quadratic forms. Left-tail probabilities}
Let $s=(s_1,\ldots,s_d) \sim \mbox{SG}(\mu,\sigma^2)$. Then
\begin{align}
 P(s^Ts < a) \leq \exp \left\{ - \frac{\mu^T\mu}{8\sigma^2} \left( 1 - \frac{a}{\mu^T\mu} \right)^2 \right\}.
\nonumber
\end{align}

\label{lem:lefttail_quadform_subgaussian}
\end{lemma}

\begin{prop}
Let $s=(s_1,\ldots,s_d) \sim \mbox{SG}(0,\sigma^2)$ and 
\begin{align}
U(d,h,c)= \int_0^1 P \left( \frac{s^T s}{c} > d \log \left( \frac{h}{(1/u-1)^{2/d}} \right) \right) du.
\nonumber
\end{align}

\begin{enumerate}[leftmargin=*,label=(\roman*)]
\item Denote $h_0=e^{2(1+\sqrt{2})^2}$. Suppose that $(c,h)$ satisfy $\log(h) \geq h_0$ and
$c \geq \sigma^2 (1 + \sqrt{\sigma^2/[c \log h]})$. Then
\begin{align}
 U(d,h,c) \leq  \frac{2 \max \left\{ [\log h]^{d/2}, \log(h^{d/2}) \right\}}{h^{d/2}}
\nonumber
\end{align}
In particular, the result holds for $c \geq \sigma^2 (1 + \sqrt{1/\log h}) \Leftrightarrow \log h \geq (c/\sigma^2-1)^{-2}$.

\item Denote $h_0= \max\{ [\sigma^2(1+\epsilon)/c] e^{2(1+\sqrt{2})^2}, e^{2(1+\sqrt{2})^2 \sigma^2/c} \}$. Suppose that $(c,h)$ satisfy $\log(h) \geq h_0$ and $c < \sigma^2 (1 + \sqrt{\sigma^2/[c \log h]})$. Then
\begin{align}
 U(d,h,c) \leq  \frac{2.5 \max \left\{ \left[\log (h^{\frac{c}{\sigma^2(1+\epsilon)}})\right]^{d/2}, \log\left(h^{\frac{dc}{2\sigma^2(1+\epsilon)}}\right) \right\}}{h^{\frac{d c}{2\sigma^2 (1 + \epsilon \sqrt{2}(1+\sqrt{2}))}}}
\nonumber
\end{align}
where $\epsilon= \sqrt{\sigma^2/[c \log h]}$.
\end{enumerate}

\label{prop:tailintegral_quadform_subgaussian}
\end{prop}

\begin{prop}
Let $s_1 \sim \mbox{SG}(0,\sigma^2)$ be a $d_1$-dimensional and $s_2 \sim \mbox{SG}(\mu,\sigma^2)$ be a $d_2$-dimensional sub-Gaussian random vector where $\sigma^2>0$ is finite and $\mu \in \mathbb{R}^{d_2}$. Define
\begin{align}
 U(d_1,d_2,c,d,h)= \int_0^1 P \left( \frac{s_1^Ts_1 - s_2^Ts_2}{c} > d \log \left(\frac{h}{(1/u-1)^{2/d}} \right) \right) du
\nonumber
\end{align}
where $c>0$, $d \in \mathbb{R}$ and $h \geq 1$. Assume that
$\mu^T\mu/ \log(\mu^T\mu) \geq -cd \log(h) + c d_1 \log(h_0)$, where $h_0$ is the positive constant in Proposition \ref{prop:tailintegral_quadform_subgaussian}.

\begin{enumerate}[leftmargin=*,label=(\roman*)]
\item Suppose that $c > \sigma^2$. Then, there exists a finite $t_0$ such that
\begin{align}
 U(d_1,d_2,c,d,h) \leq \frac{3 \max \left\{ [\frac{2}{d_1} \log(l(\mu))]^{d_1/2}, \log(l(\mu)) \right\}}{l(\mu)}
\nonumber
\end{align}
for all $\mu^T\mu \geq t_0$, where $l(\mu)= h^{d/2} e^{\frac{\mu^T\mu}{2c\log(\mu^T\mu)}}$.
Hence, for any fixed $a<1$ it holds that $\lim_{\mu^T\mu \rightarrow \infty} U(d_1,d_2,c,d,h)/l(\mu)^a=0$.

\item Suppose that $c \leq \sigma^2$. Then, there exists a finite $t_0$ such that
\begin{align}
 U(d_1,d_2,c,d,h) \leq \frac{3.5 \max \left\{[\frac{2}{d_1} \log(l'(\mu))^{1/(1+\epsilon)} ]^{d_1/2}, \log(l'(\mu)^{1/(1+\epsilon)})  \right\}}{l'(\mu)^{1/(1+\epsilon)}},
\nonumber
\end{align}
for all $\mu^T\mu \geq t_0$ and $\epsilon\geq \sqrt{2}(1+\sqrt{2})\sigma^2 d_1 / [c \mu^T\mu/\log(\mu^T\mu) + d \log(h)]$,
where $l'(\mu)= h^{d/2} e^{\frac{\mu^T\mu}{2\sigma^2\log(\mu^T\mu)}}$.

Hence, for any fixed $a<1$ it holds that $\lim_{\mu^T\mu \rightarrow \infty} U(d_1,d_2,c,d,h)/l'(\mu)^a=0$.
\end{enumerate}

\label{prop:tailintegral_difquadform_subgaussian}
\end{prop}

\subsection*{Proof of Proposition \ref{prop:tailintegral_difquadform_subgaussian}}

Denote by $w=d \log(h/[1/u-1]^{2/d})$ and let $w'>0$ be an arbitrary number, the union bound gives
\begin{align}
&P\left( \frac{s_1^Ts_1 - s_2^Ts_2}{c} > w \right)
=P\left( \frac{s_1^Ts_1 - s_2^Ts_2}{c} > \frac{w}{2} + w' - (w' - \frac{w}{2}) \right)
\nonumber \\
&\leq P \left(\frac{s_1^Ts_1}{c} >  \frac{w}{2} + w' \right)
 + P \left( \frac{s_2^T s_2}{c} < w' - \frac{w}{2}  \right).
\nonumber
\end{align}
We shall take $w'= 0.5 \log (h^d/[1/u-1]^2 e^{2\mu^T\mu/[c \log(\mu^T\mu)]})$, so that
$
 w' - \frac{w}{2}= \mu^T\mu/[c \log(\mu^T\mu)
$
and
$$
 w' + \frac{w}{2}= d\log \left( \frac{h}{(1/u-1)^{2/d}} \right) + \frac{\mu^T\mu}{c \log(\mu^T\mu)}.
$$

Applying Lemma \ref{lem:lefttail_quadform_subgaussian} immediately gives that
\begin{align}
\int_0^1 P \left( \frac{s_2^T s_2}{c} < w' - \frac{w}{2}  \right) du \leq
\exp \left\{  - \frac{\mu^T\mu}{8\sigma^2} \left( 1 - \frac{1}{\log(\mu^T\mu)} \right)^2  \right\}.
\label{eq:proofquadform_term2}
\end{align}

Next consider
\begin{align}
 \int_0^1 P \left(\frac{s_1^Ts_1}{c} >  \frac{w}{2} + w' \right) du =
\int_0^1 P \left(\frac{s_1^Ts_1}{c} > d_1 \log \left( \frac{h^{\frac{d}{d_1}} e^{\mu^T\mu/[c d_1 \log(\mu^T\mu)]}}{(1/u-1)^{2/d_1}} \right)  \right) du.
\label{eq:proofquadform_term1}
\end{align}
This integral is in the form required by Proposition \ref{prop:tailintegral_quadform_subgaussian}.
That proposition requires the condition that
$$
\log(h^{\frac{d}{d_1}} e^{\mu^T\mu/[c d_1 \log(\mu^T\mu)]}) \geq h_0
\Leftrightarrow
\frac{\mu^T\mu}{\log(\mu^T\mu)} \geq -cd \log(h) + c d_1 \log(h_0)
$$
where $h_0$ is the finite constant given in Proposition \ref{prop:tailintegral_quadform_subgaussian} defined separately for the cases $c>\sigma^2$ and $c \leq \sigma^2$. Since this condition holds by assumption, we may apply Proposition \ref{prop:tailintegral_quadform_subgaussian}.

Consider first the case $c> \sigma^2$. Proposition \ref{prop:tailintegral_quadform_subgaussian} Part (i) gives that \eqref{eq:proofquadform_term1} is
\begin{align}
 \leq \frac{2 \max \left\{  [\frac{2}{d_1} \log(l(\mu))]^{d_1/2} , \log(l(\mu)) \right\}}{l(\mu)},
\label{eq:proofquadform_boundterm1_c1}
\end{align}
where $l(\mu)= h^{d/2} e^{\mu^T\mu/[2 c \log(\mu^T\mu)]} $.
Combining \eqref{eq:proofquadform_term2} and \eqref{eq:proofquadform_boundterm1_c1} gives
\begin{align}
 U(d_1,d_2,c,d,h) \leq 
e^{  - \frac{\mu^T\mu}{8\sigma^2} \left( 1 - \frac{1}{\log(\mu^T\mu)} \right)^2} +
\frac{2 \max \left\{  [\frac{2}{d_1} \log(l(\mu))]^{d_1/2} , \log(l(\mu)) \right\}}{l(\mu)}.
\nonumber
\end{align}
Noting that the second term is larger than the first term as $\mu^T\mu$ grows gives that there exists $t_0$ such that, for all $\mu^T\mu>t_0$,
\begin{align}
  U(d_1,d_2,c,d,h) \leq 
\frac{3 \max \left\{  [\frac{2}{d_1} \log(l(\mu))]^{d_1/2} , \log(l(\mu)) \right\}}{l(\mu)},
\nonumber
\end{align}
as we wished to prove.

Consider now the case $c \leq \sigma^2$. Proposition \ref{prop:tailintegral_quadform_subgaussian} Part (ii) and simple algebra gives that \eqref{eq:proofquadform_term1} is
\begin{align}
 \leq \frac{2.5 \max \left\{[\frac{2}{d_1} \log(l'(\mu))^{1/(1+\epsilon)} ]^{d_1/2}, \log(l'(\mu)^{1/(1+\epsilon)})  \right\}}{l'(\mu)^{1/(1+\epsilon \sqrt{2}(1+\sqrt{2}))}}
\label{eq:proofquadform_boundterm1_c2}
\end{align}
where $l'(\mu)= h^{dc/2\sigma^2} e^{\mu^T\mu/[2 \sigma^2 \log(\mu^T\mu)]}$,
and $\epsilon= \sigma^2 d_1 / [c \mu^T\mu/\log(\mu^T\mu) + d \log(h)] $.
Combining \eqref{eq:proofquadform_term2} and \eqref{eq:proofquadform_boundterm1_c2} gives
\begin{align}
  U(d_1,d_2,c,d,h) \leq 
e^{  - \frac{\mu^T\mu}{8\sigma^2} \left( 1 - \frac{1}{\log(\mu^T\mu)} \right)^2} +
\frac{2.5 \max \left\{[\frac{2}{d_1} \log(l'(\mu))^{1/(1+\epsilon)} ]^{d_1/2}, \log(l'(\mu)^{1/(1+\epsilon)})  \right\}}{l'(\mu)^{1/(1+\epsilon \sqrt{2}(1+\sqrt{2}))}}.
\nonumber
\end{align}
Noting that the first term is smaller than the second term as $\mu^T\mu$ grows gives that there exists $t_0$ such that, for all $\mu^T\mu>t_0$,
\begin{align}
   U(d_1,d_2,c,d,h) \leq 
\frac{3.5 \max \left\{[\frac{2}{d_1} \log(l'(\mu))^{1/(1+\epsilon')} ]^{d_1/2}, \log(l'(\mu)^{1/(1+\epsilon')})  \right\}}{l'(\mu)^{1/(1+\epsilon')}},
\nonumber
\end{align}
where $\epsilon'\geq \sqrt{2}(1+\sqrt{2})\sigma^2 d_1 / [c \mu^T\mu/\log(\mu^T\mu) + d \log(h)]$.

\subsection*{Proof of Lemma \ref{lem:subgaussian_lincomb}}
Consider an arbitrary $\alpha \in \mathbb{R}^d$, and let $\tilde{\alpha}= A^T \alpha$. Then
\begin{align}
 E \left( e^{\alpha^T (As - A\mu)}\right) =
 E \left( e^{\tilde{\alpha}^T (s - \mu)} \right) \leq
 e^{\frac{\tilde{\alpha}^T \tilde{\alpha} \sigma^2}{2}} 
=  e^{\frac{\alpha^T A A^T \alpha \sigma^2}{2}} 
\leq  e^{\frac{\alpha^T \alpha \lambda \sigma^2}{2}},
\nonumber
\end{align}
where $\lambda$ is the largest eigenvalue of $A^TA$, as we wished to prove.

\subsection*{Proof of Lemma \ref{lem:subgaussian_quadform}}
Let $s= (X^TX)^{-1/2} X^T (y - a) \in \mathbb{R}^d$, then $s \sim \mbox{SG}((X^TX)^{-1/2} X^T (\mu-a), \lambda \sigma^2)$
by Lemma \ref{lem:subgaussian_lincomb}, where $\lambda$ is the largest eigenvalue of
\begin{align}
 (X^TX)^{-1/2} X^T X (X^TX)^{-1/2},
\nonumber
\end{align}
which is the identity matrix, hence $\lambda=1$, as we wished to prove.

\subsection*{Proof of Lemma \ref{lem:tail_quadform_subgaussian}}
\underline{Part (i).} Theorem 1 in \cite{hsu:2012} for the $\mu=0$ case gives that
\begin{align}
e^{-w} \geq  P \left( \frac{s^Ts}{\sigma^2} > (d + 2 \sqrt{dw} + 2w)  \right)
= P \left( \frac{s^Ts}{\sigma^2} > d \left[ 1 + 2 \sqrt{\frac{w}{d}} + \frac{2w}{d} \right]  \right)
\nonumber
\end{align}
for any $w>0$. Equivalently, letting $t=2w/d$,
\begin{align}
 P \left( \frac{s^Ts}{\sigma^2} > d t \left[ 1 + \sqrt{2/t} + 1/t \right]  \right) \leq e^{-dt/2}.
\nonumber
\end{align}

\underline{Part (ii).} 
Let $t_0>0$ be fixed and define $k_0= 1/t_0 + \sqrt{2c/t_0}$. For all $t \geq t_0$, it holds that $k_0 \geq 1/t + \sqrt{2/t}$.
Hence, by Part (i),
\begin{align}
P \left( \frac{s^Ts}{\sigma^2} > d t k_0 \right) \leq e^{-dt/2}.
\nonumber
\end{align}
Equivalently, letting $q= t(1+k_0)$,
\begin{align}
P \left( \frac{s^Ts}{\sigma^2} > d q \right) \leq e^{-\frac{dq}{2(1+k_0)}},
\nonumber
\end{align}
for any $k_0$ such that $k_0 \geq 1/t + \sqrt{2/t}= (1+k_0)/q + \sqrt{2(1+k_0)/q}$.

\underline{Part (iii).}
The result follows by showing that $q \geq 2 (1+\sqrt{2})^2$ and $k_0= \sqrt{2}(1+\sqrt{2})/\sqrt{q}$ satisfy the condition
$$(1+k_0)/q + \sqrt{2(1+k_0)/q} \leq k_0$$
required by Part (ii).

First, for the specified values it holds that $k_0 \leq 1$ and $k_0 \leq q-1$, since $q \geq 2$. Hence
\begin{align}
 (1+k_0)/q + \sqrt{2} \sqrt{(1+k_0)/q} \leq (1 + \sqrt{2}) \sqrt{(1+k_0)/q}
\nonumber
\end{align}
and it suffices to prove that
\begin{align}
(1 + \sqrt{2}) \sqrt{(1+k_0)/q} \leq k_0
\Leftrightarrow
\sqrt{\frac{1+k_0}{k_0^2}} \leq \frac{\sqrt{q}}{1+\sqrt{2}}.
\nonumber
\end{align}
Since $k_0 \leq 1$, we have $\sqrt{(1+k_0)/k_0^2} \leq \sqrt{2/k_0^2}$, hence it suffices to show
\begin{align}
\frac{\sqrt{2}}{k_0} \leq \frac{\sqrt{q}}{1+\sqrt{2}},
\nonumber
\end{align}
which holds for $k_0= \sqrt{2}(1+\sqrt{2})/\sqrt{q}$, as we wished to prove.

\subsection*{Proof of Lemma \ref{lem:tail_quadform_subgaussian}}
Let $t>0$ be an arbitrary real number. Clearly
\begin{align}
 P(s^Ts < a)= P \left( e^{-t s^Ts} > e^{-t a} \right) \leq e^{t a} E(e^{-t s^T s}),
\nonumber
\end{align}
the right-hand side following from Markov's inequality. Since $e^{-t (s-\mu)^T (s-\mu)}<1$,
\begin{align}
e^{t a} E(e^{-t s^T s})= 
e^{t a} E(e^{-t [(s-\mu)^T (s-\mu) + \mu^T \mu^T + 2 \mu^T (s-\mu)] }) \leq
e^{t (a - \mu^T \mu)} E(e^{-2 t \mu^T (s-\mu)})
\leq e^{t(a - \mu^T \mu) + 2 t^2 \mu^T \mu \sigma^2},
\nonumber
\end{align}
where the right-hand side follows from the definition of sub-Gaussianity.
The bound holds for any $t>0$, we obtain the tightest bound by minimizing the exponent with respect to $t$.
Setting its first derivative to 0 gives that the minimum is attained at
$t_0= (\mu^T\mu-a)/(4 \mu^T\mu \sigma^2)$, giving that
\begin{align}
 P(s^Ts < a) \leq e^{t_0(a - \mu^T \mu) + 2 t_0^2 \mu^T \mu \sigma^2}
= e^{-\frac{(a - \mu^T \mu)^2}{8 \mu^T \mu \sigma^2}}= e^{-\frac{\mu^T \mu}{8 \sigma^2} (1 - a/\mu^T\mu)^2}.
\nonumber
\end{align}

\subsection*{Proof of Proposition \ref{prop:tailintegral_quadform_subgaussian}}
\underline{Part (i).} Let $\underline{u}= [\log(h)/h]^{d/2} \in (0,1)$. Note that
\begin{align}
 U(d,h,c) \leq \underline{u} + \int_{\underline{u}}^1 P \left( \frac{s^T s}{\sigma^2} > \frac{c d}{\sigma^2} \log \left( \frac{h}{(1/u-1)^{2/d}} \right) \right) du,
\nonumber
\end{align}
Let $q=(c/\sigma^2) \log(h / [1/u-1]^{2/d})$. We shall show that $q \geq 2 (1+\sqrt{2})^2$, satisfying the condition in Lemma \ref{lem:tail_quadform_subgaussian}, Part (iii). To ease notation let $a=2 (1+\sqrt{2})^2$. Then $q \geq a$ requires
\begin{align}
 \frac{c}{\sigma^2} \log\left(\frac{h}{[1/u-1]^{2/d}}\right) \geq a \Leftrightarrow
u \geq \frac{1}{1 + (h/e^{a \sigma^2/c})^{d/2}},
\nonumber
\end{align}
which holds since
\begin{align}
 \underline{u}= [\log(h)/h]^{d/2} \geq (e^{a \sigma^2/c}/h)^{d/2}
> \frac{1}{1 + (h/e^{a \sigma^2/c})^{d/2}} 
\nonumber
\end{align}
given that $\log h \geq e^a > e^{a \sigma^2/c}$ by assumption, since $\sigma^2<c$.
We remark that a possible alternative strategy is to use Part(ii) from Lemma \ref{lem:tail_quadform_subgaussian} to set a milder condition on $h$ than the required $\log h> e^a$, and still apply the remainder of the proof's strategy, however the algebra becomes somewhat more complicated and for simplicity we use Part (iii).

Using Lemma \ref{lem:tail_quadform_subgaussian}, Part (iii), gives
\begin{align}
 U(d,h,c) \leq \underline{u} + \int_{\underline{u}}^1 
\left[ \frac{(1/u-1)^{2/d}}{h} \right]^{\frac{d c}{2 \sigma^2 [1+k_u]}} du
\nonumber
\end{align}
for $k_u=\sqrt{a/q}$.
To bound the integrand let $\underline{q}= (c/\sigma^2) \log(h / [1/\underline{u}-1]^{2/d})$ and $\underline{k}=\sqrt{a/\underline{q}}$.
Since $q$ is increasing in $u$, it follows that $k_u \leq \underline{k}$.
Further, note that the integrand is $\leq 1$. To see this, 
for any $u \geq \underline{u}$ it holds that
$(1/u-1)^{2/d} \leq (1/\underline{u}-1)^{2/d}< h/\log(h)$, where $\log(h)\geq 1$ by assumption.
Hence the integrand is upper-bounded by taking a smaller power, specifically
\begin{align}
 U(d,h,c) \leq \underline{u} + \int_{\underline{u}}^1 
\left[ \frac{(1/u-1)^{2/d}}{h} \right]^{\frac{d c}{2 \sigma^2 [1+\underline{k}]}} du
\nonumber
\end{align}
Simple algebra shows that the assumption $c \geq \sigma^2 (1 + \sqrt{\sigma^2/[c \log h]} )$ implies that
$c \geq \sigma^2 (1+\underline{k})$.
Hence the integral on the right-hand side is
\begin{align}
\leq \frac{1}{h^{d/2}} \int_{\underline{u}}^1 (1/u-1) du=
\frac{1}{h^{d/2}}[\log(1) - \log(\underline{u}) - (1 - \underline{u})] < \frac{1}{h^{d/2}} \log(1/\underline{u}).
\nonumber
\end{align}
Plugging in $\underline{u}= [\log(h)/h]^{d/2}$ gives
\begin{align}
U(d,h,c) < \left[ \frac{\log h}{h} \right]^{d/2} +
\frac{1}{h^{d/2}} \log\left( \frac{h^{d/2}}{[\log(h)]^{d/2}} \right)
< 
\left[ \frac{\log h}{h} \right]^{d/2} +
\frac{\log\left( h^{d/2} \right)}{h^{d/2}} 
\nonumber
\end{align}
since $\log(h^{d/2}) \geq 1$.


\underline{Part (ii).} 
Define $\tilde{c}= \sigma^2 ( 1 + \epsilon)$, where $\epsilon= \sqrt{\sigma^2/[c \log h]}$, so that $c/\tilde{c}<1$.
We shall split $U(d,h,c)$ by integrating over $u \leq 0.5$ and $u>0.5$, then apply Part (i) to the first integral and Lemma \ref{lem:tail_quadform_subgaussian} to the second integral. Specifically,
\begin{align}
&\int_0^{0.5} P \left( \frac{s^T s}{c} > d \log \left( \frac{h}{(1/u-1)^{2/d}} \right) \right) du=
\int_0^{0.5} P \left( \frac{s^T s}{\tilde{c}} > d \log \left( h^{\frac{c}{\tilde{c}}} \left[\frac{u}{1-u} \right]^{\frac{2c}{d\tilde{c}}} \right) \right) du
\nonumber \\
&< \int_{0}^{0.5} P \left( \frac{s^T s}{\tilde{c}} > d \log \left( h^{\frac{c}{\tilde{c}}} \left[\frac{u}{1-u} \right]^{\frac{2}{d}} \right) \right) du.
\nonumber
\end{align}
the right-hand side following from $[u/(1-u)]^{c/\tilde{c}} > u/(1-u)$, since $u/(1-u) \leq 1$ and $c/\tilde{c}<1$.
The resulting integral can be bounded by Part (i), since $\tilde{c}$ satisfies $\tilde{c} \geq \sigma^2 (1 + \sqrt{\sigma^2/[\tilde{c} \log h]})$ and $h^{c/\tilde{c}}$ satisfies $\log (h^{c/\tilde{c}}) \geq e^{2(1+\sqrt{2})^2}$,
since 
$$
\log h \geq \frac{\tilde{c} e^{2(1+\sqrt{2})^2}}{c} =
\frac{\sigma^2 ( 1 + \epsilon) e^{2(1+\sqrt{2})^2}}{c} 
$$
by assumption.
Therefore, applying Part (i) of Proposition \ref{prop:tailintegral_quadform_subgaussian},
\begin{align}
 \int_{0}^{0.5} P \left( \frac{s^T s}{\tilde{c}} > d \log \left( h^{\frac{c}{\tilde{c}}} \left[\frac{u}{1-u} \right]^{\frac{2}{d}} \right) \right) du < \frac{2 \max \left\{ [\log( h^{c/\tilde{c}} )]^{d/2}, \log( h^{\frac{dc}{2\tilde{c}}}  ) \right\}}{h^{\frac{dc}{2\tilde{c}}}}.
\label{eq:bound_partii_int1}
\end{align}

Next consider
\begin{align}
&\int_{0.5}^1 P \left( \frac{s^T s}{\tilde{c}} > d \log \left( h^{\frac{c}{\tilde{c}}} \left[\frac{u}{1-u} \right]^{\frac{2c}{d\tilde{c}}} \right) \right) du \leq
\int_{0.5}^1 P \left( \frac{s^T s}{\tilde{c}} > d \log \left( h^{\frac{c}{\tilde{c}}} \right) \right) du
\nonumber \\
&=0.5 P \left( \frac{s^T s}{\sigma^2} > d \log \left( h^{\frac{c}{\sigma^2}} \right) \right),
\nonumber
\end{align}
since $u/(1-u) \geq 1$ for $u \geq 0.5$. To bound the right-hand side we use Lemma \ref{lem:tail_quadform_subgaussian}, Part (iii).
The lemma requires that
$\log(h^{\frac{c}{\sigma^2}}) \geq 2(1+\sqrt{2})^2$, which holds since 
$$
\log h \geq \frac{\sigma^2}{c} 2 (1+\sqrt{2}) \Leftrightarrow h \geq e^{\frac{\sigma^2}{c} 2 (1+\sqrt{2})}
$$
by assumption.
Thus
\begin{align}
 0.5 P \left( \frac{s^T s}{\sigma^2} > d \log \left( h^{\frac{c}{\sigma^2}} \right) \right) \leq
\frac{1}{2 h^{\frac{dc}{2\sigma^2(1+k_0)}} },
\label{eq:bound_partii_int2}
\end{align}
where $k_0= \sqrt{2}(1+\sqrt{2}) \sqrt{\sigma^2/[c \log h]}= \sqrt{2}(1+\sqrt{2}) \epsilon$.

To conclude the proof, plug in $\tilde{c}= \sigma^2 ( 1 + \epsilon)$ into \eqref{eq:bound_partii_int1} and combine with \eqref{eq:bound_partii_int2} to obtain
\begin{align}
 U(d,h,c) < 
\frac{2 \max \left\{ [\log( h^{c/[\sigma^2 (1 + \epsilon)]} )]^{d/2}, \log( h^{\frac{dc}{2 \sigma^2 (1+\epsilon) }}  ) \right\}}{h^{\frac{dc}{2 \sigma^2 (1 + \epsilon) }}}
+ \frac{1}{2 h^{\frac{dc}{2\sigma^2 \epsilon (1+ \sqrt{2}(1+\sqrt{2}))}} }
\nonumber \\
<
\frac{2 \max \left\{ [\log( h^{c/[\sigma^2 (1 + \epsilon)]} )]^{d/2}, \log( h^{\frac{dc}{2 \sigma^2 (1+\epsilon) }}  ) \right\}}{h^{\frac{dc}{2 \sigma^2 (1 + \epsilon \sqrt{2}(1+\sqrt{2})) }}}
+ \frac{1}{2 h^{\frac{dc}{2\sigma^2(1+ \epsilon \sqrt{2}(1+\sqrt{2}))}} },
\nonumber
\end{align}
giving the desired result.

\newpage

\section{Proof of Theorem \ref{thm:bf_ala_highdim}}
\label{sec:proof_bf_ala_highdim}

Fix any $\gamma \ne \tgamma^* \in \Gamma$. The goal is to show that the random variable $\tilde{p}(\gamma \mid y)$ converges to 0 in the $L_1$ sense,
that is $\lim_{n \rightarrow \infty} E_{F^*}[\tilde{p}(\gamma \mid y)]= 0$.
We also seek to bound its convergence rate, that is find a sequence $a_n$ such that
$\lim_{n \rightarrow \infty} E_{F^*}[\tilde{p}(\gamma \mid y)] / a_n= 0$, where the bound $a_n$ is as tight as possible.
Since $L_1$ convergence implies convergence in probability, this also proves the weaker result that $\tilde{p}(\tgamma^* \mid y)= O_p(a_n)$.

Following the construction in \cite{rossell:2018}, note that
\begin{align}
E_{F^*} \left( \tilde{p}(\gamma \mid y) \right)
= E_{F^*} \left( \left[ 1 + \sum_{\gamma' \neq \gamma} \tilde{B}_{\gamma' \gamma} \frac{p(\gamma')}{p(\gamma)} \right]^{-1} \right)
< E_{F^*} \left( \left[ 1 + \tilde{B}_{\tgamma^* \gamma} \frac{p(\tgamma^*)}{p(\gamma)} \right]^{-1} \right),
\nonumber
\end{align}
and that the right-hand side is the expectation of a positive random variable. Hence its expectation can be found by integrating the survival (or right-tail probability) function
\begin{align}
 \int_0^1 P_{F^*}\left( \left[ 1 + \tilde{B}_{\tgamma^* \gamma} \frac{p(\tgamma^*)}{p(\gamma)} \right]^{-1} > u \right) du=
\int_0^1 P_{F^*} \left( \log \tilde{B}_{\gamma \tgamma^*} > \log \left( \frac{p(\tgamma^*)}{p(\gamma) (1/u-1)} \right) \right) du.
\label{eq:integral_survfun}
\end{align}

Taking the expression for $\tilde{B}^L_{\gamma,\tgamma^*}$ in \eqref{eq:normal_bf_knownphi} under an arbitrary Normal prior $p(\beta_\gamma \mid \phi, \gamma)= N(\beta_\gamma; 0, \phi g_L V_\gamma^{-1}/b''(0))$, where $V_\gamma$ is a positive-definite matrix, gives
\begin{align}
\log \tilde{B}^L_{\gamma \tgamma^*}=
\log \left( \left[ \frac{1}{n g_L} \right]^{\frac{p_\gamma - p_{\tgamma^*}}{2}}
  \frac{|Z_{\tgamma^*}^T Z_{\tgamma^*}/n|^{\frac{1}{2}} |V_\gamma|^{\frac{1}{2}}}
  {|Z_{\gamma}^T Z_{\gamma}/n|^{\frac{1}{2}} |V_{\tgamma^*}|^{\frac{1}{2}}} \right)
+ \frac{b''(0)}{2\phi} \tilde{W}_{\gamma \tgamma^*}
\nonumber
\end{align}
where
$\tilde{W}_{\gamma \tgamma^*}= \tilde{\beta}_\gamma^T (Z_\gamma^T Z_\gamma - V_\gamma/g_L)\tilde{\beta}_\gamma - 
\tilde{\beta}_{\tgamma^*}^T (Z_{\tgamma^*}^T Z_{\tgamma^*} - V_{\tgamma^*}/g_L) \tilde{\beta}_{\tgamma^*}$.

We focus on the particular case of Zellner's prior where $V_\gamma=Z_\gamma^T Z_\gamma/n$.
Rearranging terms in \eqref{eq:integral_survfun} gives
\begin{align}
 \int_0^1 P_{F^*} \left( \left[ 1-\frac{1}{ng_L} \right] \frac{W_{\gamma \tgamma^*}}{\phi/b''(0)} > (p_\gamma - p_{\tgamma^*}) \log \left( \frac{h}{(1/u-1)^{\frac{2}{p_\gamma - p_{\tgamma*}}} } \right)   \right) du,
\label{eq:integral_survfun_nested}
\end{align}
where $W_{\gamma \tgamma^*}= \tilde{\beta}_\gamma^T Z_\gamma^T Z_\gamma \tilde{\beta}_\gamma - 
\tilde{\beta}_{\tgamma^*}^T Z_{\tgamma^*}^T Z_{\tgamma^*} \tilde{\beta}_{\tgamma^*}$ and
\begin{align}
h= n g_L \left[ \frac{p(\tgamma^*)}{p(\gamma)} \right]^{\frac{2}{p_\gamma - p_{\tgamma*}}}.
\nonumber
\end{align}

The proof's strategy is based on noting that $W_{\gamma \tgamma^*}$ can be expressed in terms of quadratic forms of sub-Gaussian random variables, which allows bounding the probability in \eqref{eq:integral_survfun_nested} and its integral with respect to $u$.
Given that $\lim_{n \rightarrow \infty} n g_L= \infty$ by assumption, to ease derivations we drop the factor $1- 1/(ng_L)$ from subsequent arguments (the results follow by noting that $1- 1/(n g_L)$ is arbitrarily close to 1 as $n$ grows).
We consider separately the case $\tgamma^* \subset \gamma$, where $W_{\gamma \tgamma^*}$ is a quadratic form of a zero-mean sub-Gaussian random vector, and the case $\tgamma^* \not\subseteq \gamma$ where $W_{\gamma \tgamma^*}$ is the difference between two sub-Gaussian quadratic forms.

Before proceeding, we remark that it is possible to extend our results to Normal priors with general $V_\gamma$. Briefly, one then obtains an analogous expression to \eqref{eq:integral_survfun_nested} that involves $\tilde{W}_{\gamma \tgamma^*}$ and determinants $|V_\gamma^{-1} Z_\gamma^T Z_\gamma/n|$. It is still possible to express $\tilde{W}_{\gamma \tgamma^*}$ in terms of quadratic forms of sub-Gaussian random variables, which then provides a bound for \eqref{eq:integral_survfun_nested} that depends on eigenvalues of $V_\gamma^{-1} Z_\gamma^T Z_\gamma/n$. By placing suitable conditions on said eigenvalues, one obtains rates analogous to those for Zellner's prior.
We also remark that, for more general non-Normal priors, one may replace $p(\beta_\gamma \mid \phi,\gamma)$ by a Taylor expansion at 0, which is equivalent to approximating $p(\beta_\gamma \mid \phi,\gamma)$ by a Normal prior, and then carrying out the proof as outlined above.

\subsection{Part (i)}
Suppose that $\tgamma^* \subset \gamma$. 
Let $X_{\gamma \setminus \tgamma^*}= Z_{\gamma \setminus \tgamma^*} - Z_{\tgamma^*}(Z_{\tgamma^*}^T Z_{\tgamma^*})^{-1} Z_{\tgamma^*}^TZ_{\gamma \setminus \tgamma^*}$, from Lemma \ref{lem:waldstat_nested} we can write
\begin{align}
& W_{\gamma \tgamma^*}= (\tilde{y} \pm Z_{\tgamma^*}\tbeta_{\tgamma^*}^*)^T X_{\gamma \setminus \tgamma^*} (X_{\gamma \setminus \tgamma^*}^T X_{\gamma \setminus \tgamma^*})^{-1} X_{\gamma \setminus \tgamma^*}^T (\tilde{y} \pm Z_{\tgamma^*}\tbeta_{\tgamma^*}^*)
\nonumber \\
&=(\tilde{y} - Z_{\tgamma^*}\tbeta_{\tgamma^*}^*)^T X_{\gamma \setminus \tgamma^*} (X_{\gamma \setminus \tgamma^*}^T X_{\gamma \setminus \tgamma^*})^{-1} X_{\gamma \setminus \tgamma^*}^T (\tilde{y} - Z_{\tgamma^*}\tbeta_{\tgamma^*}^*),
\nonumber
\end{align}
since $Z_{\tgamma^*}^T X_{\gamma \setminus \tgamma^*}= Z_{\tgamma^*}^T Z_{\gamma \setminus \tgamma^*} - Z_{\tgamma^*}^T Z_{\gamma \setminus \tgamma^*}=0$.
Since $\tilde{y} \sim \mbox{SG}(Z_{\tgamma^*} \tbeta_{\tgamma^*}^*, \sigma^2)$, Lemma \ref{lem:subgaussian_quadform} gives that $W_{\gamma \tgamma^*}=s^Ts$ where $s \sim \mbox{SG}(0, \sigma^2)$ is a $p_\gamma - p_{\tgamma^*}$ dimensional sub-Gaussian random vector.
Therefore the integral in \eqref{eq:integral_survfun_nested} is of the form required by Proposition \ref{prop:tailintegral_quadform_subgaussian}, letting $c=\phi/b''(0)$, $d=p_\gamma - p_{\gamma'}$ and $h$ as defined in \eqref{eq:integral_survfun_nested}.
We apply Proposition \ref{prop:tailintegral_quadform_subgaussian} Part (i) or Part (ii) depending on whether $\phi/b''(0)> \sigma^2$  or $\phi/b''(0) < \sigma^2$.

Consider first the case $c \geq \sigma^2 (1 + \sqrt{1/\log h})$.
Applying Proposition \ref{prop:tailintegral_quadform_subgaussian} Part (i) requires the condition that $\log(h) \geq e^{2(1+\sqrt{2})^2}$.
Note that, as $n$ grows,
\begin{align}
\log(h)= \log(ng_L) + \frac{2}{p_\gamma - p_{\tgamma*}} \log \left( \frac{p(\tgamma^*)}{p(\gamma)} \right) \gg 1
\nonumber
\end{align}
by assumption. Therefore the condition holds for large enough $n$ and Proposition \ref{prop:tailintegral_quadform_subgaussian} Part (i) gives
\begin{align}
E_{F^*} \left( \left[ 1 + \tilde{B}_{\tgamma^* \gamma} \frac{p(\tgamma^*)}{p(\gamma)} \right]^{-1} \right) \leq  \frac{2 \max \left\{ [\log h]^{d/2}, \log(h^{d/2}) \right\}}{h^{d/2}}
\nonumber
\end{align}
for all $n \geq n_0$, and some fixed $n_0$, as we wished to prove.

Consider next the case $c < \sigma^2 (1 + \sqrt{\sigma^2/[c \log h]})$. 
Applying Proposition \ref{prop:tailintegral_quadform_subgaussian} Part (ii) requires the condition that
$\log(h) \geq \max\{ [\sigma^2(1+\epsilon)/c] e^{2(1+\sqrt{2})^2}, e^{2(1+\sqrt{2})^2 \sigma^2/c} \}$, which again holds for large enough $n$, since $(c,\sigma^2)$ are constant and $\epsilon$ is arbitrarily close to 0 as $h$ grows.
Hence, Proposition \ref{prop:tailintegral_quadform_subgaussian} Part (ii) gives
\begin{align}
E_{F^*} \left( \left[ 1 + \tilde{B}_{\tgamma^* \gamma} \frac{p(\tgamma^*)}{p(\gamma)} \right]^{-1} \right) \leq  \frac{2.5 \max \left\{ \left[\log (h^{\frac{c}{\sigma^2(1+\epsilon)}})\right]^{d/2}, \log\left(h^{\frac{dc}{2\sigma^2(1+\epsilon)}}\right) \right\}}{h^{\frac{d c}{2\sigma^2 (1 + \epsilon \sqrt{2}(1+\sqrt{2}))}}}
\nonumber
\end{align}
where $\epsilon= \sqrt{\sigma^2/[c \log h]}$, for all $n \geq n_0$ and some fixed $n_0$, as we wished to prove.

\subsection{Part (ii)}
Suppose that $\tgamma^* \not\subseteq \gamma$. Denote by $\gamma'= \tgamma^* \cup \gamma$ the model such that $Z_{\gamma'}$ includes all columns from $Z_{\tgamma^*}$ and $Z_{\gamma}$. Note that
$$
W_{\gamma \tgamma^*}= 
\tilde{\beta}_\gamma^T Z_\gamma^T Z_\gamma \tilde{\beta}_\gamma - 
\tilde{\beta}_{\gamma'}^T Z_{\gamma'}^T Z_{\gamma'} \tilde{\beta}_{\gamma'}
+ \tilde{\beta}_{\gamma'}^T Z_{\gamma'}^T Z_{\gamma'} \tilde{\beta}_{\gamma'}
- \tilde{\beta}_{\tgamma^*}^T Z_{\tgamma^*}^T Z_{\tgamma^*} \tilde{\beta}_{\tgamma^*}
= W_{\gamma' \tgamma^*} - W_{\gamma' \gamma}.
$$

Since $\tgamma^* \subset \gamma'$, by Lemma \ref{lem:subgaussian_quadform} it holds that
$W_{\gamma' \tgamma^*}=s_1^Ts_1$ where $s_1 \sim \mbox{SG}(0,\sigma^2)$ is a $(p_{\gamma'} - p_{\tgamma^*})$-dimensional sub-Gaussian vector. Further, since $\gamma \subset \gamma'$, by Lemma \ref{lem:waldstat_nested},
\begin{align}
 W_{\gamma' \gamma}&= (\tilde{y} \pm Z_\gamma \tbeta_\gamma^*)^T \tilde{Z}_{\gamma' \setminus \gamma} (\tilde{Z}_{\gamma' \setminus \gamma}^T \tilde{Z}_{\gamma' \setminus \gamma})^{-1} \tilde{Z}_{\gamma' \setminus \gamma}^T (\tilde{y} \pm Z_\gamma \tbeta_\gamma^*)
\nonumber \\
&= (\tilde{y} - Z_\gamma \tbeta_\gamma^*)^T \tilde{Z}_{\gamma' \setminus \gamma} (\tilde{Z}_{\gamma' \setminus \gamma}^T \tilde{Z}_{\gamma' \setminus \gamma})^{-1} \tilde{Z}_{\gamma' \setminus \gamma}^T (\tilde{y} - Z_\gamma \tbeta_\gamma^*),
\nonumber
\end{align}
where $\tilde{Z}_{\gamma' \setminus \gamma}= (I - Z_\gamma(Z_\gamma^TZ_\gamma)^{-1}Z_\gamma^T) Z_{\gamma' \setminus \gamma}$,
and the right-hand side follows from $Z_\gamma^T \tilde{Z}_{\gamma' \setminus \gamma}=0$ and direct algebra.

Since $\tilde{y} - Z_\gamma \tbeta_\gamma^* \sim \mbox{SG}(Z_{\tgamma^*} \tbeta_{\tgamma^*}^* - Z_\gamma\tbeta_\gamma^* , \sigma^2)$ by assumption,
by Lemma \ref{lem:subgaussian_quadform} we may write $W_{\gamma' \gamma}=s_2^Ts_2$, where $s \sim \mbox{SG}(\mu, \sigma^2)$ is a $(p_{\gamma'} - p_{\gamma})$-dimensional sub-Gaussian random vector, and $\mu= (\tilde{Z}_{\gamma' \setminus \gamma}^T \tilde{Z}_{\gamma' \setminus \gamma})^{-1/2} \tilde{Z}_{\gamma' \setminus \gamma}^T (Z_{\tgamma^*} \tbeta_{\tgamma^*}^* - Z_\gamma \tbeta_\gamma^*) \neq 0$.
Therefore \eqref{eq:integral_survfun_nested} is equal to
\begin{align}
 \int_0^1 P_{F^*} \left( \frac{s_1^Ts_1 - s_2^Ts_2}{\phi/b''(0)} > (p_\gamma - p_{\tgamma^*}) \log \left( \frac{h}{(1/u-1)^{2/(p_\gamma - p_{\tgamma^*})}} \right)   \right) du.
\label{eq:integral_survfun_nonnested}
\end{align}
The latter expression has the form required by Proposition \ref{prop:tailintegral_difquadform_subgaussian},
setting $c= \phi/b''(0)$, $d_1= p_{\gamma'}-p_{\tgamma^*}$, $d_2=p_{\gamma'}- p_{\tgamma^*}$ and $d=p_\gamma - p_{\tgamma^*}$.
The non-centrality parameter, denoted $\lambda_\gamma$, is
\begin{align}
\lambda_\gamma= \mu^T\mu=  
(Z_{\tgamma^*} \tbeta_{\tgamma^*}^* - Z_\gamma \tbeta_\gamma^*)^T \tilde{Z}_{\gamma' \setminus \gamma}
(\tilde{Z}_{\gamma' \setminus \gamma}^T \tilde{Z}_{\gamma' \setminus \gamma})^{-1} \tilde{Z}_{\gamma' \setminus \gamma}^T (Z_{\tgamma^*} \tbeta_{\tgamma^*}^* - Z_\gamma \tbeta_\gamma^*)
\nonumber \\
=
(Z_{\tgamma^*} \tbeta_{\tgamma^*}^*)^T (I - H_\gamma) Z_{\tgamma^* \setminus \gamma}
[Z_{\tgamma^* \setminus \gamma}^T (I - H_\gamma) Z_{\tgamma^* \setminus \gamma}]^{-1} Z_{\tgamma^* \setminus \gamma}^T 
(I - H_\gamma) Z_{\tgamma^*} \tbeta_{\tgamma^*}^*,
\nonumber
\end{align}
where $H_\gamma= Z_\gamma (Z_\gamma^T Z_\gamma)^{-1} Z_\gamma^T$ is the projection matrix onto the column span of $Z_\gamma$,
and to obtain the right-hand side we used that $I-H_\gamma$ is idempotent,
that the optimal linear projection $Z_\gamma \tbeta_\gamma^*= H_\gamma Z_{\gamma^*} \tbeta_{\gamma^*}^*$,
and noting that $Z_{\gamma' \setminus \gamma}= Z_{\gamma^* \setminus \gamma}$.

The expression for $\lambda_\gamma$ can be simplified. First, note that $Z_{\tgamma^*} \tbeta_{\tgamma^*}^*= Z_{\tgamma^* \setminus \gamma} \tbeta_{\tgamma^* \setminus \gamma}^* + Z_{\gamma \cap \tgamma^*} b_{\gamma \cap \tgamma^*}$, where $b_{\gamma \cap \tgamma^*}$ is the subset of $\tbeta_{\tgamma^*}^*$ corresponding to the columns of $Z_{\gamma \cap \tgamma^*}$. Hence
\begin{align}
(Z_{\tgamma^*} \tbeta_{\tgamma^*}^*)^T (I-H_\gamma)= 
(\tbeta_{\tgamma^* \setminus \gamma}^*)^T Z_{\tgamma^* \setminus \gamma}^T (I-H_\gamma) +  b_{\gamma \cap \tgamma^*}^T Z_{\gamma \cap \tgamma^*}^T (I - H_\gamma)=
(\tbeta_{\tgamma^* \setminus \gamma}^*)^T Z_{\tgamma^* \setminus \gamma}^T (I-H_\gamma)
\nonumber
\end{align}
since $Z_{\gamma \cap \tgamma^*}^T (I - H_\gamma)=0$, as $Z_{\gamma \cap \tgamma^*}$ is in the linear span of $Z_\gamma$. Therefore
\begin{align}
 \lambda_\gamma&=
(\tbeta_{\tgamma^* \setminus \gamma}^*)^T Z_{\tgamma^* \setminus \gamma}^T (I-H_\gamma) Z_{\tgamma^* \setminus \gamma}
[Z_{\tgamma^* \setminus \gamma}^T (I - H_\gamma) Z_{\tgamma^* \setminus \gamma}]^{-1} Z_{\tgamma^* \setminus \gamma}^T 
(I - H_\gamma) Z_{\tgamma^*} \tbeta_{\tgamma^*}^*
\nonumber \\
&=(Z_{\tgamma^* \setminus \gamma} \tbeta_{\tgamma^* \setminus \gamma}^*)^T (I - H_\gamma) Z_{\tgamma^*} \tbeta_{\tgamma^*}^*
=(Z_{\tgamma^*} \tbeta_{\tgamma^*}^*)^T (I - H_\gamma) Z_{\tgamma^*} \tbeta_{\tgamma^*}^*.
\nonumber
\end{align}

Applying Proposition \ref{prop:tailintegral_difquadform_subgaussian} requires the following condition to hold:
\begin{align}
&\frac{\mu^T\mu}{c\log(\mu^T\mu)} \geq -d \log(h) + d_1 \log(h_0)
\nonumber \\
&= (p_{\tgamma^*} - p_\gamma) \log \left( n g_L \right)
+  2 \log \left( \frac{p(\gamma)}{p(\tgamma^*)} \right)
 +  (p_{\gamma'} - p_{\tgamma^*}) \log(h_0)
\label{eq:condition_ncp_proof}
\end{align}
where $\log(h_0)>0$ and $h_0$ is the positive finite constant in Proposition \ref{prop:tailintegral_quadform_subgaussian}.

The fact that \eqref{eq:condition_ncp_proof} holds follows directly from Assumption (D3).
If $p_\gamma \geq p_{\tgamma^*}$, then $p_{\gamma'}-p_{\tgamma^*} \leq p_{\gamma}$ and 
\begin{align}
 \eqref{eq:condition_ncp_proof} \leq (p_{\tgamma^*} - p_\gamma) \log \left( n g_L \right)
+  2 \log \left( \frac{p(\gamma)}{p(\tgamma^*)} \right)
 +  p_\gamma \log(h_0)
\ll \frac{\mu^T\mu}{\log(\mu^T\mu)},
\label{eq:condition_ncp_proof_large}
\end{align}
where the right-hand side holds by assumption.
If $p_\gamma < p_{\tgamma^*}$, then $p_{\gamma'} - p_{\tgamma^*} \leq p_{\tgamma^*}$ and
\begin{align}
  \eqref{eq:condition_ncp_proof} \leq 
(p_{\tgamma^*} - p_\gamma) \log \left( n g_L \right)
+  2 \log \left( \frac{p(\gamma)}{p(\tgamma^*)} \right) +  p_{\tgamma^*} \log(h_0)
\ll \frac{\mu^T\mu}{\log(\mu^T\mu)},
\label{eq:condition_ncp_proof_small}
\end{align}
which again holds by assumption.

To complete the proof we apply Proposition \ref{prop:tailintegral_quadform_subgaussian}.
Consider first the case $c=\phi/b''(0)> \sigma^2$. By Proposition \ref{prop:tailintegral_difquadform_subgaussian} Part (i), there exists a finite $t_0$ such that \eqref{eq:integral_survfun_nonnested} is
\begin{align}
 \leq \frac{3 \max \left\{ [\frac{2}{p_{\gamma'}-p_{\tgamma^*}} \log(l(\mu))]^{(p_{\gamma'}-p_{\tgamma^*})/2}, \log(l(\mu)) \right\}}{l(\mu)}
\nonumber
\end{align}
for all $\mu^T\mu \geq t_0$, where $l(\mu)= h^{(p_\gamma - p_{\tgamma^*})/2} e^{\frac{\mu^T\mu b''(0)}{2\phi\log(\mu^T\mu)}}$.
Given that $\lim_{n \rightarrow \infty} \mu^T\mu= \infty$ by assumption and
$$
\frac{\mu^T\mu b''(0)}{2\phi\log(\mu^T\mu)} \gg   (p_\gamma - p_{\tgamma^*})/2 \log(h)
$$
from \eqref{eq:condition_ncp_proof}, for any $a<1$ there exists a finite $n_0$ such that, for all $n \geq n_0$,
\begin{align}
\eqref{eq:integral_survfun_nonnested}
 \leq e^{-a \left[\frac{\mu^T\mu b''(0)}{2\phi\log(\mu^T\mu)} + \frac{(p_\gamma - p_{\tgamma^*}) \log(h)}{2} \right]}.
\label{eq:bfrate_nonspur_case1}
\end{align}

Consider now the case $c=\phi/b''(0) \leq \sigma^2$. By Proposition \ref{prop:tailintegral_difquadform_subgaussian} Part (ii), there exists a finite $t_0$ such that 
\begin{align}
\eqref{eq:integral_survfun_nonnested} \leq \frac{3.5 \max \left\{[\frac{2}{p_{\gamma'}-p_{\tgamma^*}} \log(l'(\mu))^{1/(1+\epsilon)} ]^{(p_{\gamma'}-p_{\tgamma^*})/2}, \log(l'(\mu)^{1/(1+\epsilon)})  \right\}}{l'(\mu)^{1/(1+\epsilon)}},
\nonumber
\end{align}
where $\lim_{\mu^T\mu \rightarrow \infty} \epsilon= 0$, $l'(\mu)= h^{(p_\gamma - p_{\tgamma^*})/2} e^{\frac{\mu^T\mu}{2\sigma^2\log(\mu^T\mu)}}$.
Hence, for any $a<1$ there exists a finite $n_0$ such that, for all $n \geq n_0$,
\begin{align}
\eqref{eq:integral_survfun_nonnested}
 \leq e^{-a \left[\frac{\mu^T\mu}{2\sigma^2\log(\mu^T\mu)} + \frac{(p_\gamma - p_{\tgamma^*}) \log(h)}{2} \right]}.
\label{eq:bfrate_nonspur_case2}
\end{align}
The desired result is obtained by noting that
\begin{align}
\exp\left\{ \frac{(p_\gamma - p_{\tgamma^*}) \log(h)}{2} \right\}=  
(n g_L)^{\frac{(p_\gamma - p_{\tgamma^*})}{2}} \frac{p(\tgamma^*)}{p(\gamma)}
\nonumber
\end{align}
and combining \eqref{eq:bfrate_nonspur_case1} and \eqref{eq:bfrate_nonspur_case2} to obtain
\begin{align}
\eqref{eq:integral_survfun_nonnested}
 \leq e^{-a \left[\frac{\mu^T\mu}{2 \max\{\sigma^2, \phi/b''(0) \} \log(\mu^T\mu)}  \right]}
\left(\frac{p(\gamma)/p(\tgamma^*)}{(n g_L)^{\frac{(p_\gamma - p_{\tgamma^*})}{2}}} \right)^a.
\label{eq:bfrate_nonspur_bothcases}
\end{align}

\newpage
\section{Proof of Theorem \ref{thm:pp_ala_highdim}}
\label{sec:proof_pp_ala_highdim}

\subsection{Part (i)}

Let $S_l= \{\gamma: |\gamma|=l, \tgamma^* \subset \gamma \}$ be the set of size $l$ models, and $\Gamma_l= \{ \gamma: |\gamma|=l, p(\gamma)>0 \}$ be those receiving positive prior probability (e.g. satisfying any specified hierarchical/group constraints). Then
\begin{align}
 E_{F_0} ( p(S \mid y))
=E_{F_0} \left( \sum_{l= |\tgamma^*|+1}^{\bar{J}} \sum_{\gamma \in S_l \cap \Gamma_l} \tilde{p}(\gamma \mid y) \right) 
\leq E_{F_0} \left( \sum_{l= |\tgamma^*|+1}^{\bar{J}} \sum_{\gamma \in S_l}  \tilde{p}(\gamma \mid y) \right).
\nonumber
\end{align}

From Theorem \ref{thm:bf_ala_highdim} Part (i), for each $\gamma$ there is a finite $n_{0 \gamma}$ such that
\begin{align}
E_{F_0} \left( \tilde{p}(\gamma \mid y) \right) \leq \frac{2 \max \left\{ [(2/[p_\gamma - p_{\tgamma^*}]) \log r_\gamma^a]^{(p_\gamma - p_{\tgamma^*})/2}, \log(r_\gamma^a) \right\}}{r_\gamma^a}
\label{eq:bound_pp_spurious}
\end{align}
for all $n \geq n_{0 \gamma}$, a constant $a$ such that $a=1$ if $\phi/b''(0) > \sigma^2$ and $a < \phi/(b''(0) \sigma^2)$ if $\phi/b''(0) \leq \sigma^2$, and
\begin{align}
  r_\gamma= \left( n g \right)^{\frac{p_\gamma - p_{\tgamma^*}}{2}} \frac{p(\tgamma^*)}{p(\gamma)}.
\nonumber
\end{align}
To alleviate algebra, by observing that
$\log(r_\gamma^a) \succeq  [(2/[p_\gamma - p_{\tgamma^*}]) \log r_\gamma^a]^{(p_\gamma - p_{\tgamma^*})/2}$ for $p_\gamma - p_{\tgamma^*} \geq 2$
and that $\log(r_\gamma^a) \ll  [(2/[p_\gamma - p_{\tgamma^*}]) \log r_\gamma^a]^{(p_\gamma - p_{\tgamma^*})/2}$ for $p_\gamma - p_{\tgamma^*}=1$,
it is simple to show that the bound $E_{F_0}(\tilde{p}(\gamma \mid y)) \leq \log (r_\gamma^a) / r_\gamma^a$ gives the same asymptotic rate as \eqref{eq:bound_pp_spurious}.
Further, inspecting the proof of Theorem \ref{thm:bf_ala_highdim} shows that $n_{0 \gamma}$ grows with the number of model parameters $p_\gamma$, and that if $p_\gamma \ll n$ then $n_0=\max_{\gamma \in S, |\gamma|=\bar{J}} n_{0 \gamma}$ can be taken to be fixed. 

We consider separately the cases $\phi/b''(0)> \sigma^2$ and $\phi/b''(0) \leq \sigma^2$.

\subsubsection*{\underline{Case $\phi/b''(0)> \sigma^2$}} Then $a=1$, and for any $n \geq n_0$,
\begin{align}
 E_{F_0} ( p(S \mid y))
\leq \sum_{l= |\tgamma^*|+1}^{\bar{J}} \sum_{\gamma \in S_l}  
\frac{\log \left( (n g)^{\frac{(p_\gamma - p_{\tgamma^*})}{2}} p(\tgamma^*)/p(\gamma) \right) }{(n g)^{\frac{(p_\gamma - p_{\tgamma^*})}{2}}} \frac{p(\gamma)}{p(\tgamma^*)}.
\label{eq:bound_pp_spurious2}
\end{align}

To prove the desired result we plug in the expression of $p(\gamma)/p(\tgamma^*)$, then use algebraic manipulation and the Binomial coefficient ordinary generating function to carry out the sum in \eqref{eq:bound_pp_spurious2}.
First note that
\begin{align}
\frac{p(\gamma)}{p(\tgamma^*)}
= \frac{p(|\gamma|)}{p(|\tgamma^*|)} \frac{{J \choose |\tgamma^*|}}{{J \choose |\gamma|}}
= p^{-c (|\gamma| - |\tgamma^*|)} {|\gamma| \choose |\tgamma^*|} {J-|\tgamma^*| \choose |\gamma| - |\tgamma^*|}^{-1}.
\label{eq:priorodds}
\end{align}
Denote by $\omega(|\gamma|)= p(\tgamma^*)/p(\gamma)$. The right-hand side in \eqref{eq:bound_pp_spurious2} is
\begin{align}
=\sum_{l= |\tgamma^*|+1}^{\bar{J}} p^{-c (l - |\tgamma^*|)} {l \choose |\tgamma^*|}  m_l
\nonumber
\end{align}
where 
\begin{align}
 m_l=
{J-|\tgamma^*| \choose l - |\tgamma^*|}^{-1}
\sum_{\gamma \in S_l}  \frac{\log \left( (n g)^{\frac{(p_\gamma - p_{\tgamma^*})}{2}} \omega(l) \right)}{(n g)^{\frac{p_\gamma - p_{\tgamma^*}}{2}}}
\nonumber
\end{align}
is an average penalization term across the ${J - |\tgamma^*| \choose l - |\tgamma^*|}$ models included in $S_l$.
Since $p_\gamma - p_{\tgamma^*}> q (|\gamma| - |\tgamma^*|)$, where $q$ is the smallest group size,
\begin{align}
\eqref{eq:bound_pp_spurious2} &\leq \sum_{l= |\tgamma^*|+1}^{\bar{J}} {l \choose |\tgamma^*|}  
\frac{\log \left((n g)^{\frac{q (l - |\tgamma^*|)}{2}} \omega(l) \right) }{p^{c (l - |\tgamma^*|)} (n g)^{\frac{q (l -|\tgamma^*|)}{2}}}
\nonumber \\
&\leq 
\log \left((n g)^{\frac{q (\bar{J} - |\tgamma^*|)}{2}} \max_{l \geq |\tgamma^*|+1} \omega(l) \right)
(p^c (n g)^{q/2})^{|\tgamma^*|} \sum_{l= |\tgamma^*|+1}^{\bar{J}} {l \choose |\tgamma^*|}  
\frac{1}{\left( p^c (n g)^{\frac{q}{2}} \right)^l}.
\label{eq:bound_pp_spurious3}
\end{align}
To complete the proof, recall that the Binomial coefficient ordinary generating function gives that
$\sum_{l=k+1}^\infty {l \choose k} b^l= b^k[ (1-b)^{-(k+1)} - 1]$ for any $b$.
Plugging in $b= p^{-c} (n g)^{-q/2}$ and $k= |\tgamma^*|$,
\begin{align}
(p^c (n g)^{q/2})^{|\tgamma^*|} \sum_{l= |\tgamma^*|+1}^{\bar{J}} {l \choose |\tgamma^*|}  
\frac{1}{\left( p^c (n g)^{\frac{q}{2}} \right)^l}
\leq
(1 - p^{-c} (n g)^{-q/2})^{-(|\tgamma^*| + 1)}  - 1.
\nonumber
\end{align}
Using that $n g \rightarrow \infty$ by assumption and the definition of the exponential function
\begin{align}
 \asymp \exp\left\{ \frac{|\tgamma^*|+1}{p^c (n g)^{q/2}} \right\} - 1 
\asymp \frac{|\tgamma^*|+1}{p^c (n g)^{q/2}}
\nonumber
\end{align}
since $\lim_{n \rightarrow \infty} [|\tgamma^*|+1] /[p^c (n g)^{q/2}]=0$ by assumption and $\lim_{z \rightarrow \infty} (e^z - 1)/z=1$.
Combining this expression with \eqref{eq:bound_pp_spurious3}
\begin{align}
\log \left((n g)^{\frac{q (\bar{J} - |\tgamma^*|)}{2}} \max_{l \geq |\tgamma^*|+1} \omega(l) \right) \frac{|\tgamma^*|+1}{p^c (n g)^{q/2}}.
\nonumber
\end{align}
Using Stirling's bounds $\sqrt{2\pi} i^{i+1/2} e^{-i} <  i! < e i^{i+1/2} e^{-i} $ gives
\begin{align}
& \omega(l) \leq \frac{e^2}{2\pi} p^{c(l - |\tgamma^*|)} (J-\bar{J})^{\bar{J}-|\tgamma^*|} \left(1 + \frac{\bar{J} - |\tgamma^*|}{J - \bar{J}} \right)^{J-|\tgamma^*|+\frac{1}{2}}.
\nonumber \\
& \log \omega(l) \leq \log \left( \frac{e^2}{2\pi} \right)
 + (\bar{J}-|\tgamma^*|) [c\log(p)  + \log(J)] + (J-|\tgamma^*|+\frac{1}{2} ) \log \left(1 + \frac{\bar{J} - |\tgamma^*|}{J - \bar{J}} \right).
\nonumber
\end{align}
Using that $l \leq \bar{J}$, $J \leq p$, and that $\lim_{n\rightarrow \infty} \bar{J}/J=0$ by assumption, for large enough $n$
\begin{align}
\log \max_{l \geq |\tgamma^*|} \omega(l) \leq  (\bar{J}-|\tgamma^*|) [c\log(p)  + \log(J) + \epsilon]
\leq  (\bar{J}-|\tgamma^*|) [(c+1)\log(p) + \epsilon]
\label{eq:bound_maxomegal}
\end{align}
for any fixed $\epsilon>0$, as we wished to prove.

\subsubsection*{\underline{Case $\phi/b''(0) \leq \sigma^2$}} Then we take $a$ to be a constant smaller than but arbitrarily close to $\phi/(b''(0)\sigma^2) \leq 1$. 
Then the right-hand side in \eqref{eq:bound_pp_spurious2} is
\begin{align}
&\leq a \log \left((n g)^{\frac{q (\bar{J} - |\tgamma^*|)}{2}} \max_{l \geq |\tgamma^*|+1} \omega(l) \right)
 \sum_{l=|\tgamma^*|+1}^{\bar{J}} \left(\frac{1}{p^{a c} (ng)^{aq/2}}\right)^{l-|\tgamma^*|} {l \choose |\tgamma^*|}^a {J - |\tgamma^*| \choose l - |\tgamma^*|}^{1-a}
\nonumber \\
&< a \log \left((n g)^{\frac{q (\bar{J} - |\tgamma^*|)}{2}} \max_{l \geq |\tgamma^*|+1} \omega(l) \right)
 \sum_{l=|\tgamma^*|+1}^{\bar{J}} \left(\frac{(J-|\tgamma^*|)^{1-a}}{p^{a c} (ng)^{aq/2}}\right)^{l-|\tgamma^*|} {l \choose |\tgamma^*|} 
\label{eq:bound_pp_spurious_overdisp}
\end{align}
since $a<1$ implies that ${l \choose |\tgamma^*|}^a < {l \choose |\tgamma^*|}$
and ${J - |\tgamma^*| \choose l - |\tgamma^*|} < (J-|\tgamma^*|)^{l-|\tgamma^*|}$.
Using the Binomial coefficient ordinary generating function gives that the sum in \eqref{eq:bound_pp_spurious_overdisp} is
\begin{align}
< \left[ \left( 1 - \frac{(J-|\tgamma^*|)^{1-a}}{p^{ac} (ng)^{aq/2}} \right)^{-(|\tgamma^*|+1)} -1 \right],
\label{eq:bound_sum_overdisp}
\end{align}
It is easy to show that the assumption that $\lim_{n\rightarrow \infty} (|\tgamma^*|+1) (ng)^{aq/2} p^{(c+1)a-1} =\infty$ implies that
$p^{ac}(ng)^{aq/2} \gg (J-|\tgamma^*|)^{1-a}$, hence the definition of the exponential function implies that \eqref{eq:bound_sum_overdisp} converges to
\begin{align}
\exp \left\{  \frac{(|\tgamma^*|+1)(J-|\tgamma^*|)^{1-a}}{p^{ac} (ng)^{aq/2}}\right\} -1 \asymp
\frac{(|\tgamma^*|+1)(J-|\tgamma^*|)^{1-a}}{p^{ac} (ng)^{aq/2}},
\nonumber
\end{align}
since $\lim_{z \rightarrow \infty} (e^z - 1)/z=1$.
Combining this expression with \eqref{eq:bound_pp_spurious_overdisp} and the bound
$\log \max_{l \geq |\tgamma^*|} \omega(l) \leq  (\bar{J}-|\tgamma^*|) [(c+1)\log(p) + \epsilon]$
given in \eqref{eq:bound_maxomegal}, we obtain
\begin{align}
(\bar{J} - |\tgamma^*|) [ \log \left((n g)^{q/2} \right) + (c+1)\log(p) + \epsilon]
\frac{(|\tgamma^*|+1)(J-|\tgamma^*|)^{1-a}}{p^{ac} (ng)^{aq/2}}.
\nonumber
\end{align}
Since $(J-|\tgamma^*|)<p$, we have $(J-|\tgamma^*|)^{1-a}/p^{ac} < 1/p^{a(c+1)-1}$, concluding the proof.

\subsection{Part (ii)}
Let $\lambda_\gamma= \frac{\mu^T\mu}{2 \max\{\sigma^2, \phi/b''(0) \} \log(\mu^T\mu)}$.
The proof strategy is to split $\tilde{p}(S^c \mid y)$ into two terms corresponding to models with smaller and larger size than $|\tgamma^*|$.
\begin{align}
E_{F_0} \left( \tilde{p}(S^c \mid y) \right)=
 E_{F_0} \left( \sum_{l= 0}^{|\tgamma^*|} \sum_{|\gamma|=l, \gamma \not\supset \tgamma^*} \tilde{p}(\gamma \mid y) \right)  
+ E_{F_0} \left( \sum_{l= |\tgamma^*|+1}^{\bar{J}} \sum_{|\gamma|=l, \gamma \not\supset \tgamma^*} \tilde{p}(\gamma \mid y) \right).
\nonumber
\end{align}
Consider the second term. Using expression \eqref{eq:bfrate_nonspur_bothcases} proven in Theorem \ref{thm:bf_ala_highdim} Part (ii),
there exists a finite $n_0= \max_{|\gamma|>|\tgamma^*|, \gamma \not\supset \tgamma^*} n_{0\gamma}$ such that
\begin{align}
E_{F_0} \left( \sum_{l= |\tgamma^*|+1}^{\bar{J}} \sum_{|\gamma|=l, \gamma \not\supset \tgamma^*} \tilde{p}(\gamma \mid y) \right)  
\leq \sum_{l= |\tgamma^*|+1}^{\bar{J}} \sum_{|\gamma|=l, \gamma \not\supset \tgamma^*}  
\left( \frac{e^{- \lambda_\gamma }}{(n g)^{\frac{(p_\gamma - p_{\tgamma^*})}{2}}} \frac{p(\gamma)}{p(\tgamma^*)} \right)^b
\label{eq:bound_pp_bignonspur}
\end{align}
for all $n \geq n_0$ and any fixed $b<1$, which we take to be arbitrarily close to 1. 
The existence of such a finite $n_0$ uniformly bounding all $n_{0 \gamma}$ follows from Condition (D4),
specifically using that $e^{\lambda_\gamma } (n g)^{\frac{(p_\gamma - p_{\tgamma^*})}{2}} p(\tgamma^*)/ p(\gamma)$
is uniformly bounded by $e^{\bar{\lambda}} (n g)^{\frac{(p_\gamma - p_{\tgamma^*})}{2}} p(\tgamma^*)/ p(\gamma)$, which under Condition (D4) converges to infinity,
and arguing as in \eqref{eq:condition_ncp_proof_large}.

Let $q= \min\{p_1,\ldots, p_J\}$ be the size of the smallest group
and $\bar{\lambda}= \min_{|\gamma|>|\tgamma^*|, \gamma \not\subset \tgamma^*} \lambda_\gamma$.
Using the expression for $p(\gamma)/p(\tgamma^*)$ 
and noting that there are ${J \choose l} - {J - |\tgamma^*| \choose l - |\tgamma^*| }$ models of size $l$ such that $\gamma \not\supset \tgamma^*$,
gives that \eqref{eq:bound_pp_bignonspur} is bounded above by
\begin{align}
 \sum_{l=|\tgamma^*|+1}^{\bar{J}} \frac{e^{-b\bar{\lambda}}}{\left(p^{b c} (ng)^{bq/2}\right)^{l-|\tgamma^*|}} {l \choose |\tgamma^*|}^b {J - |\tgamma^*| \choose l - |\tgamma^*|}^{-b} \left[ {J \choose l} - {J - |\tgamma^*| \choose l - |\tgamma^*| } \right]
\nonumber \\
\leq 
\sum_{l=|\tgamma^*|+1}^{\bar{J}} \frac{e^{-b\bar{\lambda}}}{\left(p^{b c} (ng)^{bq/2}\right)^{l-|\tgamma^*|}} {l \choose |\tgamma^*|}
\left(\frac{l - |\tgamma^*|}{J - |\tgamma^*|}\right)^{b(l - |\tgamma^*|)} J^l
\nonumber \\
< e^{-b\bar{\lambda}} J^{|\tgamma^*|} \sum_{l=|\tgamma^*|+1}^{\bar{J}} \frac{J^{l-|\tgamma^*|}}{\left(p^{b c} (ng)^{bq/2}\right)^{l-|\tgamma^*|}} 
\left(\frac{\bar{J} - |\tgamma^*|}{J - |\tgamma^*|}\right)^{b(l - |\tgamma^*|)}
{l \choose |\tgamma^*|}
\nonumber \\
= e^{-b\bar{\lambda}} J^{|\tgamma^*|} \sum_{l=|\tgamma^*|+1}^{\bar{J}} 
\left( \frac{J^{1-b} (\bar{J} - |\tgamma^*|)^b}{(1 - |\tgamma^*|/J)^b p^{b c} (ng)^{bq/2}} \right)^{l - |\tgamma^*|}
{l \choose |\tgamma^*|}
\nonumber
\end{align}
since $b<1$ implies that ${l \choose |\tgamma^*|}^b < {l \choose |\tgamma^*|}$.
Note that this expression is analogous to \eqref{eq:bound_pp_spurious_overdisp}. Therefore, we may use the Binomial coefficient ordinary generating function to carry out the sum and the exponential function to obtain a simpler asymptotically equivalent expression, to show that
\begin{align}
\eqref{eq:bound_pp_bignonspur} \leq 
e^{-b\bar{\lambda}} J^{|\tgamma^*|}
\frac{(|\tgamma^*|+1) J^{1-b} (\bar{J} - |\tgamma^*|)^b}{(1 - |\tgamma^*|/J)^b p^{b c} (ng)^{bq/2}}
\nonumber\\
\leq e^{(|\tgamma^*| + 1)\log J} \frac{(|\tgamma^*|+1) (\bar{J} - |\tgamma^*|)^b}{[e^{\bar{\lambda}} p^c (ng)^{q/2}]^b}
= \frac{(|\tgamma^*|+1) e^{(|\tgamma^*| + 1)\log J + \log{\bar{J}}}}{[e^{\bar{\lambda}} p^c (ng)^{q/2}]^b}
\label{eq:bound_pp_bibnonspur_final}
\end{align}
for large enough $n$.

Consider now the models of size $|\gamma| \leq |\tgamma^*|$. 
Let $\underline{\lambda}= \min_{|\gamma| \leq |\tgamma^*|} \lambda_\gamma/\max\{|\tgamma^*| - |\gamma|,1\}$,
$q'= \max \{p_1,\ldots,p_J\}$ the largest group size and note that
$$
\frac{e^{- \lambda_\gamma }}{(n g)^{\frac{(p_\gamma - p_{\tgamma^*})}{2}}} \leq 
\frac{e^{-\underline{\lambda}\max\{|\tgamma^*|-|\gamma|,1\}}}{(ng)^{(|\gamma|-|\tgamma^*|)/2}}.
$$
Using again Theorem \ref{thm:bf_ala_highdim} Part (ii),
there exists a finite $n_0'= \max_{|\gamma|\leq |\tgamma^*|} n_{0\gamma}$ such that
\begin{align}
 E_{F_0} \left( \sum_{l= 0}^{|\tgamma^*|} \sum_{|\gamma|=l} \tilde{p}(\gamma \mid y) \right) \leq
 \sum_{l=0}^{|\tgamma^*|} \sum_{|\gamma|=l}  
\left( \frac{e^{- \lambda_\gamma }}{(n g)^{\frac{(p_\gamma - p_{\tgamma^*})}{2}}} \frac{p(\gamma)}{p(\tgamma^*)} \right)^b
\nonumber \\
\leq e^{-\underline{\lambda}} {J \choose |\tgamma^*|} 
+ \sum_{l=0}^{|\tgamma^*|-1} \left( \frac{e^{\underline{\lambda}}}{p^{bc}(ng)^{q'/2}} \right)^{l - |\tgamma^*|}
\left( \frac{{J \choose |\tgamma^*|}}{{J \choose l}} \right)^b {J \choose l}
\label{eq:bound_pp_smallnonspur}
\end{align}
for all $n> n_0'$ and any fixed $b<1$, the right-hand side following from noting that there are ${J \choose l}$ models of size $|\gamma|=l<|\tgamma^*|$ and that $(ng)^{p_{\tgamma^*} - p_\gamma} \leq (ng)^{q'}$.
The existence of a finite $n_0'$ uniformly bounding all $n_{0 \gamma}$ follows from Condition (D5),
specifically bounding $\lambda_\gamma$ with $\underline{\lambda}$ and arguing as in \eqref{eq:condition_ncp_proof_small}.

Using that ${J \choose l}^{1-b} < J^{l(1-b)}$ and ${J \choose |\tgamma^*|} \leq J^{|\tgamma^*|}$, the geometric series, and simple algebra we obtain
\begin{align}
\eqref{eq:bound_pp_smallnonspur} 
&\leq 
e^{-\underline{\lambda} + |\tgamma^*| \log J}  +
e^{|\tgamma^*| \log J} \left[e^{-\underline{\lambda}} p^{bc}(ng)^{q'/2}\right]^{|\tgamma^*|} 
\sum_{l=0}^{|\tgamma^*|-1} \left( \frac{ e^{\underline{\lambda}} J^{(1-b)}}{\left( p^{bc}(ng)^{q'/2} \right)} \right)^l
\nonumber \\
&< e^{-\underline{\lambda} + |\tgamma^*| \log J}  + e^{|\tgamma^*| \log J} [e^{-\underline{\lambda}} p^{bc}(ng)^{q'/2}]^{|\tgamma^*| + 1} 
\nonumber \\
& < \frac{e^{|\tgamma^*| \log J}}{e^{\underline{\lambda}}}  + 
\frac{e^{|\tgamma^*| \log J}}{e^{(|\tgamma^*| + 1) [\underline{\lambda} - c\log(p) - \log(ng^{q'/2})]}}.
\nonumber
\end{align}
The proof concludes by combining this expression with \eqref{eq:bound_pp_bibnonspur_final}, and noting that $\bar{J} \leq J$.

\section{Supplementary results}
\label{sec:suppl_results}

\subsection{Computation times and posterior inference in Poisson regression}
\label{ssec:suppl_simpoisson}

\begin{figure}
\begin{center}
\begin{tabular}{cc}
$p=10$, $n \in \{100,500,1000,5000\}$ &
$n=500$, $p \in \{5,10,25,50\}$ \\
\includegraphics[width=0.5\textwidth]{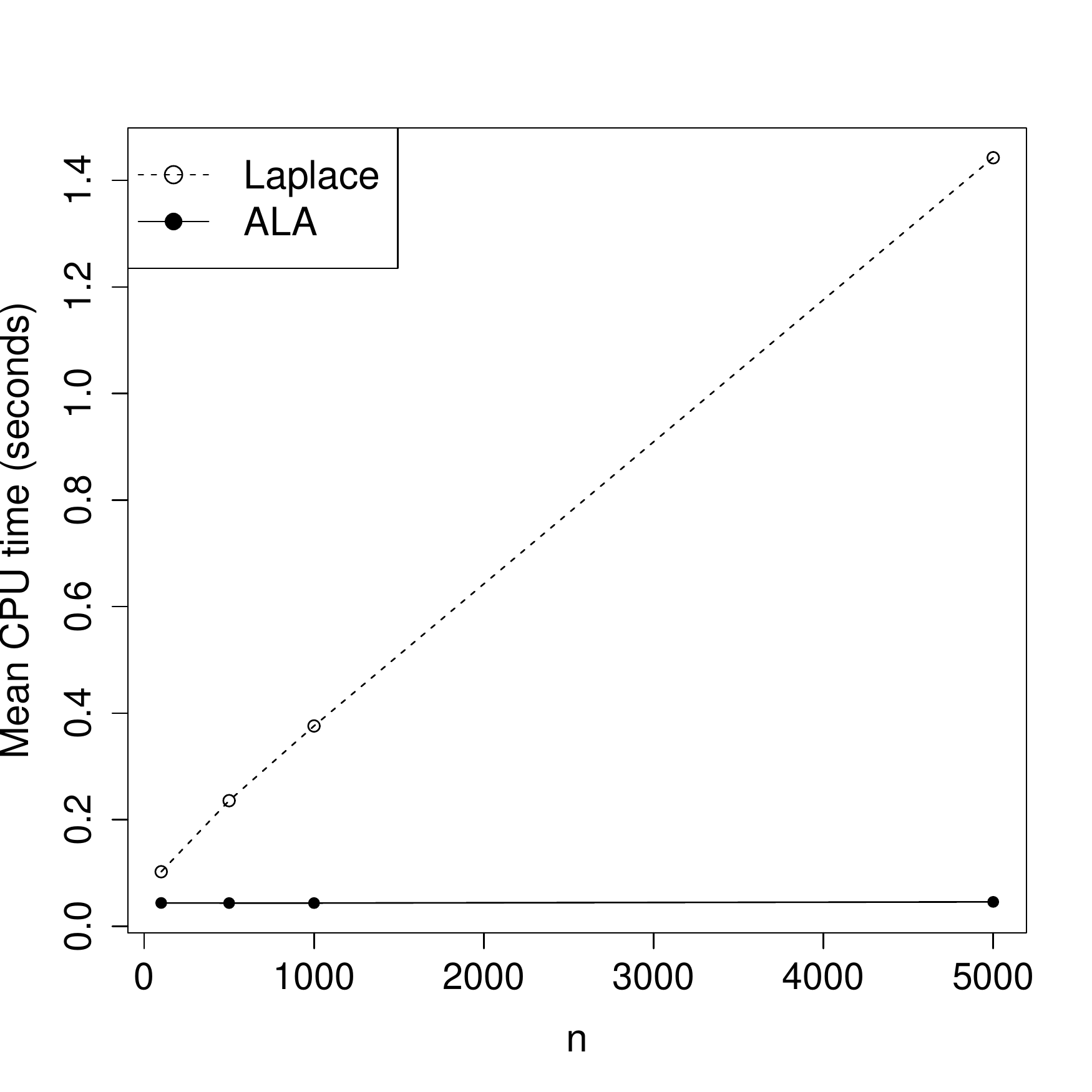} & 
\includegraphics[width=0.5\textwidth]{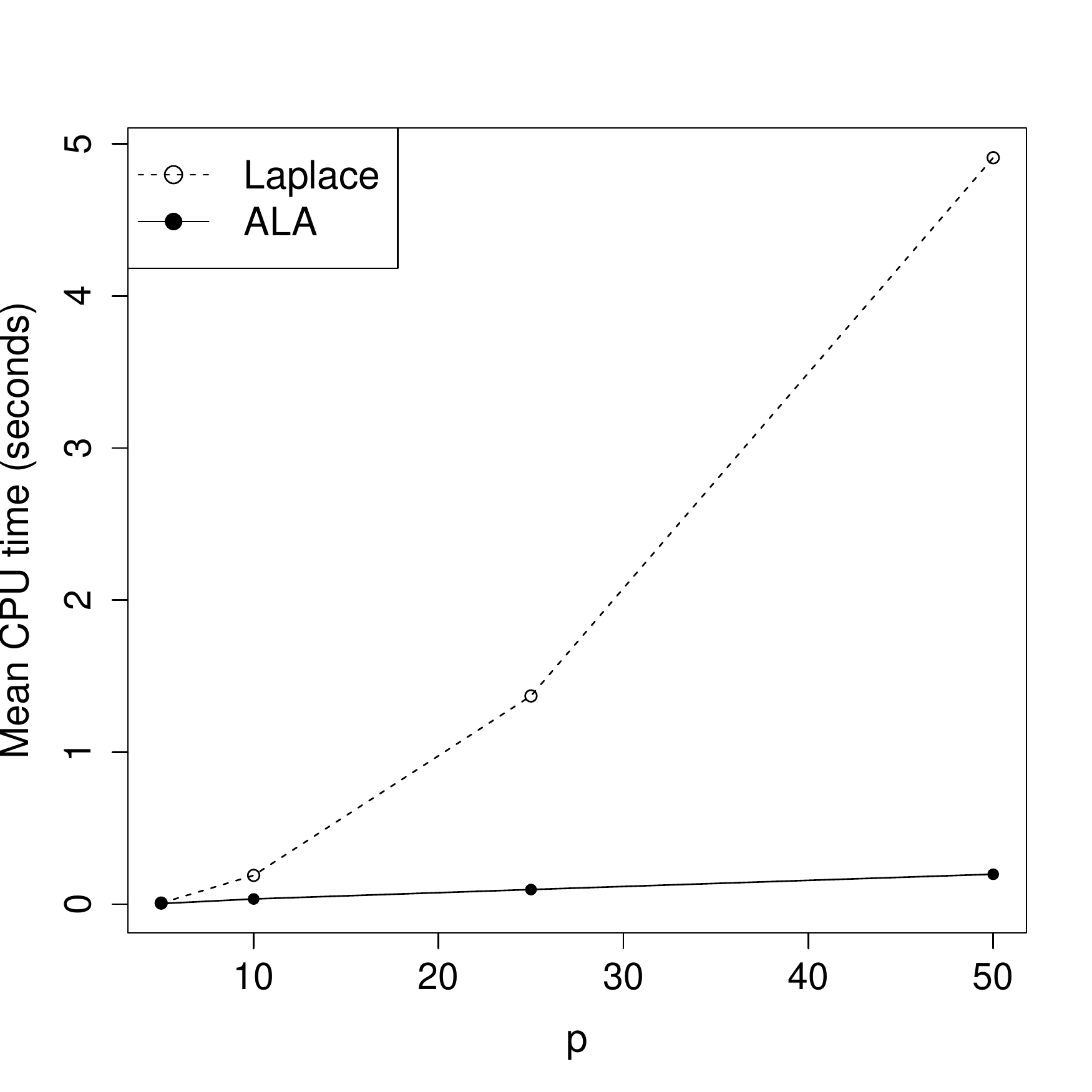} \\
\includegraphics[width=0.5\textwidth]{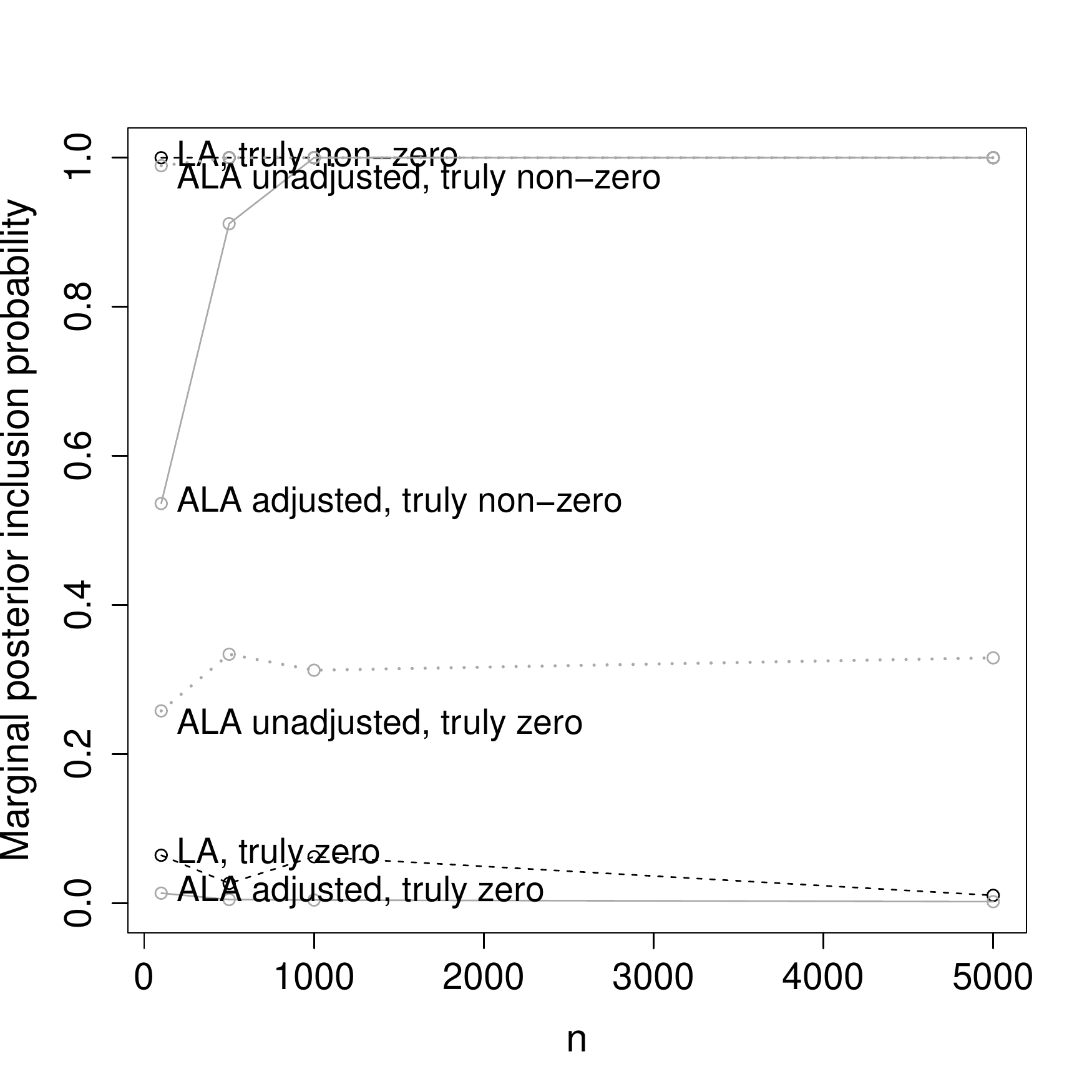} & 
\includegraphics[width=0.5\textwidth]{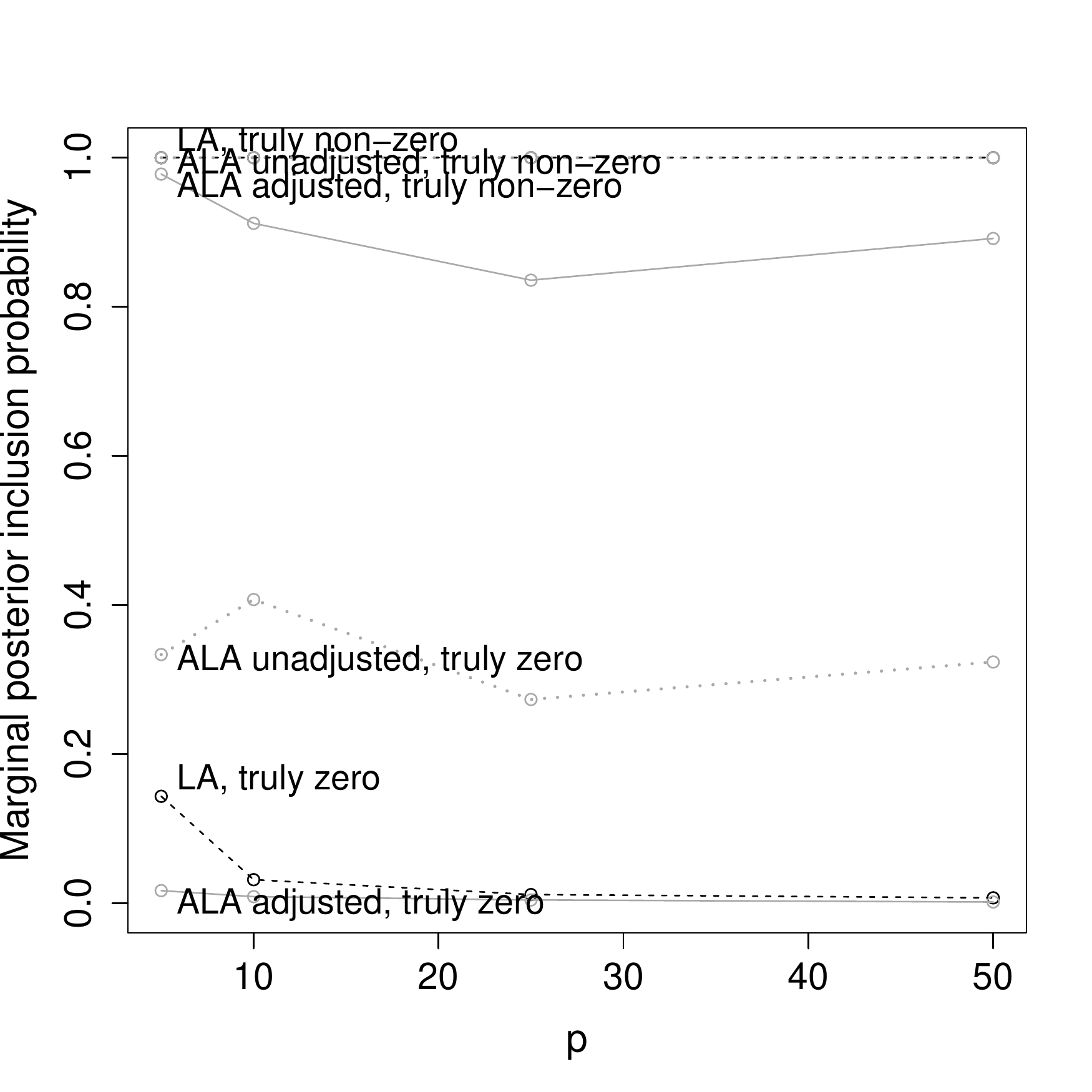} \\
\end{tabular}
\end{center}
\caption{Poisson simulation. The top panels show the average run time (seconds) in a single-core i7 processor with Ubuntu 20.04 operating system,
and the bottom panels the average posterior inclusion probabilities for variables that are truly active ($\beta_j^* \neq 0$) and truly inactive ($\beta_j^* =0$).
The left panels correspond to fixed $p=10$ and various $n$, and the right panels to fixed $n=500$ and various $p$
}
\label{fig:simpoisson}
\end{figure}

We illustrate the computation times and quality of inference in a setting where the data are truly Poisson with $\beta^*=(0,\ldots,0,0.5,1)$ and $z_i \sim N(0,\Sigma)$ with $\Sigma_{ii}=1$ and $\Sigma_{ij}=0.5$ for all $i \neq j$.
Figure \ref{fig:simpoisson} summarizes the results.
The top panels show that ALA run times are significantly more scalable as either $n$ or $p$ grow.
The bottom panels show that ALA tends to assign lower posterior inclusion probabilities $P(\beta_j \neq 0 \mid y)$ than LA, particularly for small $n$ or large $p$. For comparison, the figure also displays results for ALA where no over-dispersion adjustment is applied. This unadjusted version assigns relatively high posterior inclusion probabilities to truly inactive variables, even for large $n$.




\subsection{Combining ALA with importance sampling and variable screening}
\label{ssec:ala_importancesampling}

\begin{table}
  \begin{center}
    
\begin{tabular}{rrrrrr}   \hline
    & ALA & ALA - 1 iter & ALA - 2 iter & LA   & ALA screening + LA  \\  \hline
  $z_{1}$ & 0.0343 & 1.0000 & 1.0000 & 1.0000  & 1.000 \\ 
  $z_{2}$ & 0.0548 & 0.1000 & 0.1109 & 0.1147  & - \\ 
  $z_{3}$ & 0.0619 & 0.0904 & 0.0981 & 0.1014  & - \\ 
  $z_{4}$ & 0.0436 & 0.0683 & 0.0735 & 0.0760  & - \\ 
  $z_{5}$ & 0.0450 & 0.0719 & 0.0773 & 0.0796  & - \\ 
  $z_{6}$ & 0.0434 & 0.0722 & 0.0781 & 0.0806  & - \\ 
  $z_{7}$ & 0.0528 & 0.0771 & 0.0827 & 0.0854  & - \\ 
  $z_{8}$ & 0.0464 & 0.0739 & 0.0802 & 0.0829  & - \\ 
  $z_{9}$ & 0.9743 & 0.9978 & 0.9986 & 0.9987  & 0.9999 \\ 
  $z_{10}$ & 1.0000 & 1.0000 & 1.0000 & 1.0000 & 1.0000 \\
  \hline
\end{tabular}

\end{center}
\caption{ Marginal posterior inclusion probabilities in logistic simulation with $p=10$, $n=1000$, $\beta_1^*=2$, $\beta_9^*=0.5$, $\beta_{10}^*=1$, $\beta_2^*=\ldots,\beta_8^*=0$ }
\label{tab:is_logistic}
\end{table}

\begin{table}
  \begin{center}
    
\begin{tabular}{rrrrrr}   \hline
 & ALA & ALA - 1 iter & ALA - 2 iter & LA       & ALA screening + LA  \\  \hline
  $z_{1}$ & 0.0061 & 0.0173 & 0.0179 & 0.0188   &  - \\ 
  $z_{2}$ & 0.0066 & 0.0171 & 0.0174 & 0.0179   &  - \\ 
  $z_{3}$ & 0.0066 & 0.0132 & 0.0138 & 0.0143   &  - \\ 
  $z_{4}$ & 0.9602 & 1.0000 & 1.0000 & 1.0000   & 1.000 \\ 
  $z_{5}$ & 1.0000 & 1.0000 & 1.0000 & 1.0000   & 1.000 \\ 
  $z_{1}^2$ & 0.0049 & 0.0000 & 0.0020 & 0.0113 &  - \\ 
  $z_{2}^2$ & 0.0054 & 0.0000 & 0.0031 & 0.0099 &  - \\ 
  $z_{3}^2$ & 0.0054 & 0.0000 & 0.0078 & 0.0118 &  - \\ 
  $z_{4}^2$ & 0.9996 & 0.0000 & 0.0000 & 0.0269 & 0.1043  \\ 
  $z_{5}^2$ & 1.0000 & 0.0000 & 0.0000 & 0.0273 & 0.1059  \\ 
 \hline
\end{tabular}

\end{center}
\caption{ Marginal posterior inclusion probabilities in Poisson simulation with $p=10$, $n=1000$, $\beta_4^*=0.5$, $\beta_5^*=1$ and $\beta_j^*=0$ for $j \not\in \{4,5\}$ }
\label{tab:is_poisson}
\end{table}

\begin{figure}
\begin{center}
\begin{tabular}{cc}
\includegraphics[width=0.5\textwidth]{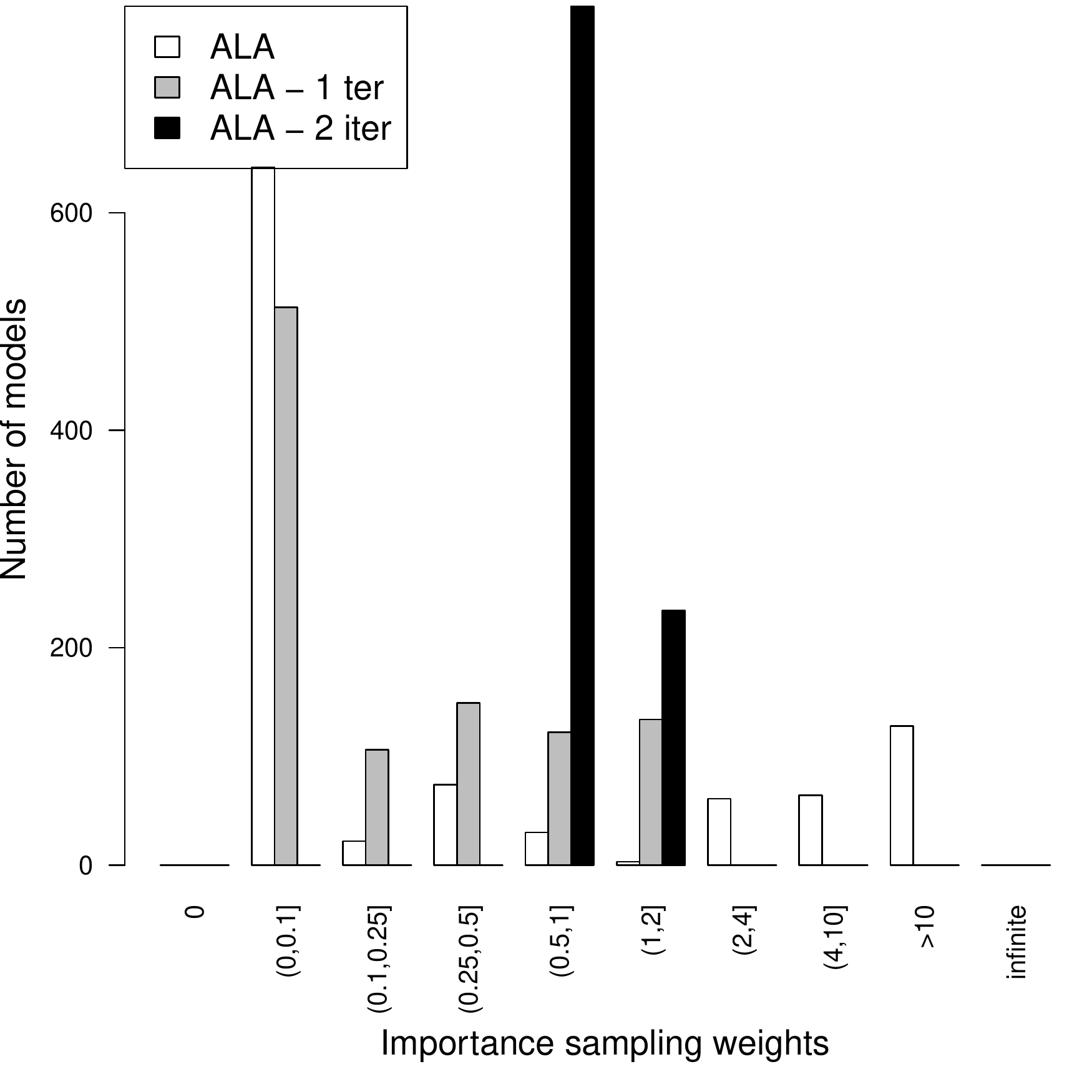} &
\includegraphics[width=0.5\textwidth]{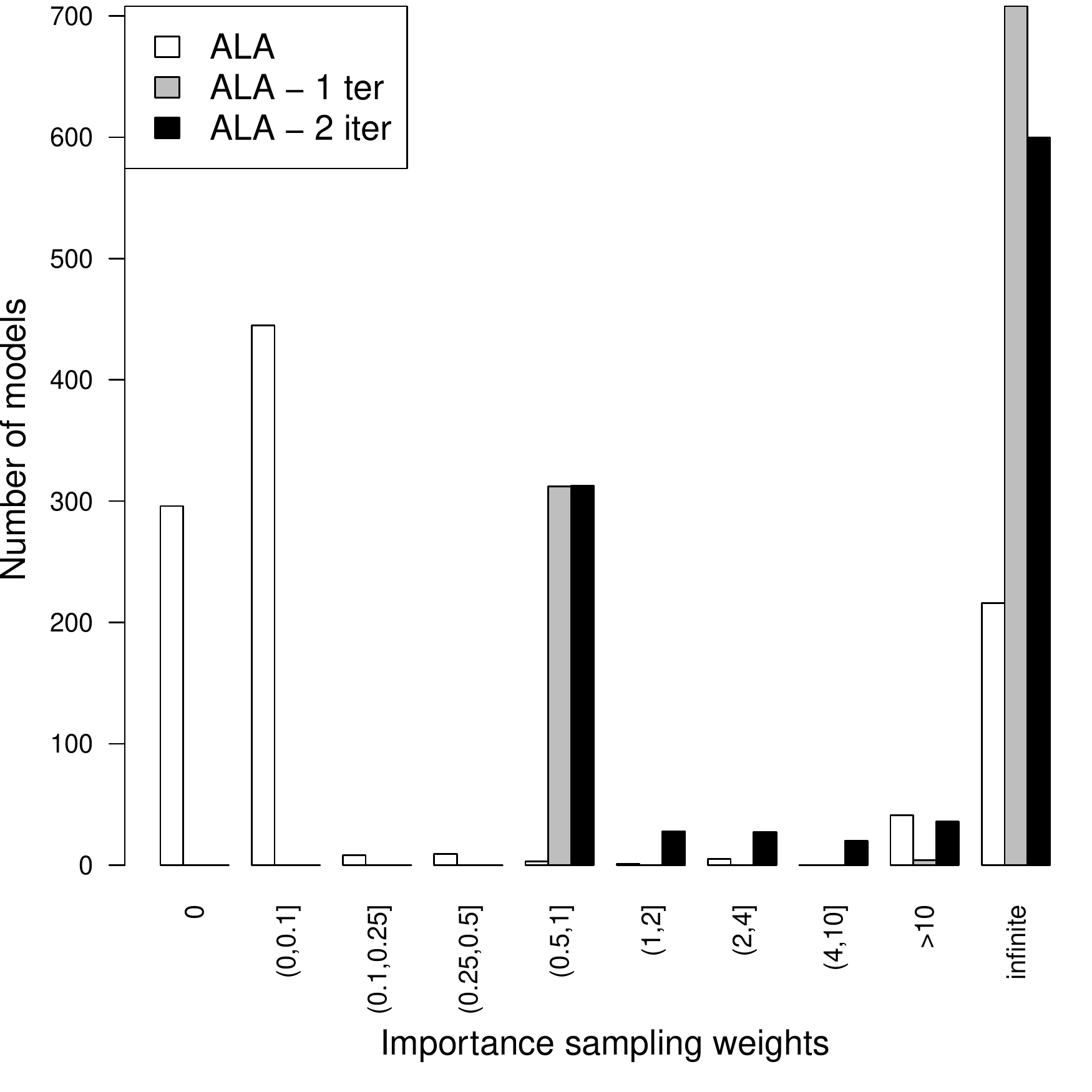} 
\end{tabular}
\end{center}
\caption{ Distribution of importance sampling weights. 
Left: logistic simulation with $p=10$, $n=1000$, $\beta_1^*=2$, $\beta_9^*=0.5$, $\beta_{10}^*=1$, $\beta_2^*=\ldots,\beta_8^*=0$.
Right: Poisson simulation with $p=10$, $n=1000$, $\beta_4^*=0.5$, $\beta_5^*=1$ and $\beta_j^*=0$ for $j \not\in \{4,5\}$ }
\label{fig:importance_sampling}
\end{figure}

This section discusses the use of ALA as a tool to pre-screen high posterior probability models. Specifically, we consider a strategy suggested by a referee based on using importance sampling to re-weight the models sampled from the ALA posterior. We illustrate that such strategy can lead to degenerate weights. We also illustrate a referee suggestion on using ALA to screen out inactive variables, and show that such strategy is effective in our examples.

Let $\tilde{p}(\gamma \mid y)$ be the posterior probability assigned by ALA and $\hat{p}(\gamma \mid y)$ that by the Laplace approximation (LA). Assuming that one has $B$ samples $\gamma^{(1)},\ldots,\gamma^{(B)}$ from $\tilde{p}(\gamma \mid y)$, importance sampling draws samples from $\hat{p}(\gamma \mid y)$ by assigning weights $w(\gamma^{(b)})= \hat{p}(\gamma \mid y) / \tilde{p}(\gamma \mid y)$ for $b=1,\ldots,B$. For this importance sampling strategy to be effective, it is necessary that the weights are bounded away from infinity. From a theoretical point of view, this is in general not true. As shown by our theory, as the sample size $n \rightarrow \infty$ the ALA and LA posteriors concentrate on two different models $\tilde{\gamma}^*$ and $\gamma^*$ respectively, hence the weight $w(\gamma^*) \stackrel{P}{\longrightarrow} \infty$.

This asymptotic argument does not rule out that in practice for finite $n$ there can be situations where importance sampling is actually effective. To illustrate this issue we consider two examples.
The first example retakes the logistic regression simulation from Section \ref{ssec:suppl_simpoisson} and illustrates a situation where importance sampling is effective. The second example is a Poisson regression where one obtains degenerate weights. In both cases we have $n=1000$, $p=10$.

Consider first the logistic regression example, where recall that $\beta_1^*=2$ for the intercept term, $\beta_9^*=0.5$, $\beta_{10}^*=1$, and $\beta_2^*=\ldots,\beta_8^*=0$. 
Figure \ref{fig:importance_sampling} (left) shows the distribution of the importance sampling weights across all $2^{10}$ models for the default ALA expanding at $\beta_0=0$, and ALA with refined $\beta_0$ obtained by taking 1 and 2 Newton--Raphson iterations.
The Newton--Raphson iterations led to a significantly improved behavior of the weights. For the default ALA no infinite weights were observed, the maximum being $\max_\gamma w(\gamma)= 64.6$.
Table \ref{tab:is_logistic} shows the posterior marginal inclusion probabilities. Interestingly, all versions of ALA and LA provided strong evidence for including the truly active covariates 9-10 and discarding the truly inactive covariates 2-8, the main discrepancy being that the default ALA excluded the intercept.

Consider next a Poisson regression setting where the design matrix has 10 covariates (plus the intercept, whose inclusion is forced into the model). Covariates 1-5 were drawn from a multivariate Gaussian with zero mean, unit variances and all pairwise correlations set to 0.5. Covariates 6-10 were the quadratic effect corresponding to covariates 1-5, e.g. $z_{i j+5}= z_{ij}^2$ for $j=1,\ldots,5$. 
We simulated $n=1000$ observations from a Poisson data-generating model where only Covariates 4-5 had a linear effect, and no quadratic effects were present. Specifically, the Poisson log-expectation was $z_i^T \beta^*$, where $\beta_4^*=0.5$, $\beta_5^*=1$ and all remaining $\beta_j^*=0$.
Figure \ref{fig:importance_sampling} (right) shows the distribution of the importance sampling weights across all $2^{10}$ models for the three ALA versions described above. In all cases there is a significant number of models with practically infinite weight (up to numerical precision).
Interestingly, the marginal posterior inclusion probabilities in Table \ref{tab:is_poisson} reveal that, despite the weight degeneracy, ALA was effective at detecting both truly active covariates 4-5, and it also discarded several truly spurious covariates.

These findings suggest that one may use ALA as a tool to screen out unpromising covariates, and then using LA to obtain posterior probabilities on the remaining ones. We tested this strategy in the logistic and Poisson regression examples. In both cases we screened out parameters with ALA marginal posterior inclusion probability $\tilde{p}(\gamma_j = 1 \mid y) < 0.5$. The last columns in Tables \ref{tab:is_logistic}-\ref{tab:is_poisson} that the use of screening followed by LA calculations selected the correct model that only includes truly active covariates, although in the  Poisson case the inclusion probabilities for two truly inactive parameters were somewhat inflated.

Altogether, these results suggest that there is promise in using ALA as a tool to pre-screen active variables. In fact, Theorem \ref{thm:sparse_linproj} gives conditions where such screening has theoretical guarantees of working asymptotically.
Regarding the use of importance sampling, while also potentially interesting, our examples show that designing effective strategies require a more careful study that is left as future work.

\subsection{Group constraints in linear regression with categorical predictors}
\label{ssec:suppl_simgaussian}

\begin{figure}
\begin{center}
\begin{tabular}{cc}
Categorical $x_{i1}$ ($\beta_1^*=0.3$, $J=50$, $p=51$) & Categorical $x_{i1}$ ($\beta_1^*=0.6$, $J=50$, $p=51$) \\
\includegraphics[width=0.5\textwidth]{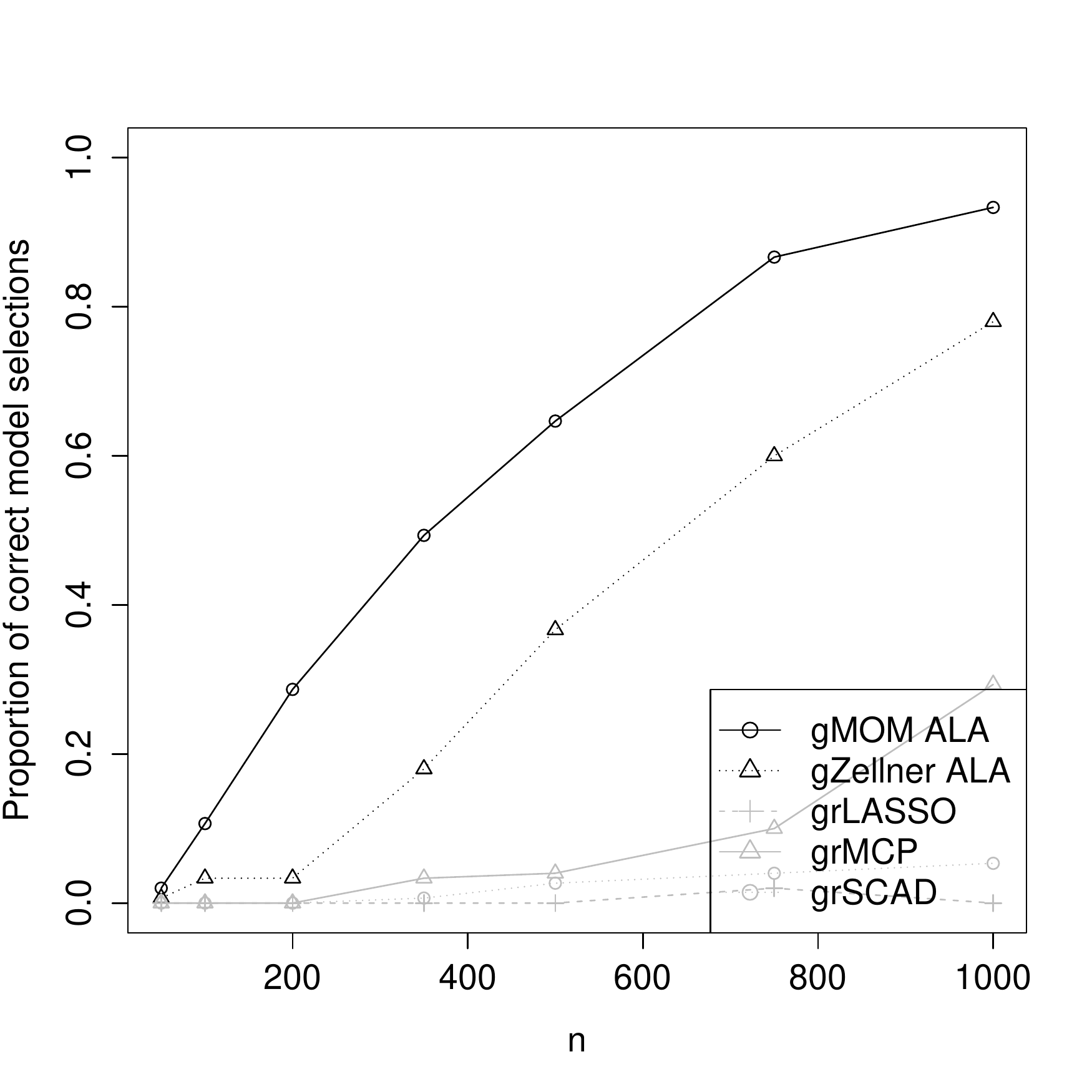} &
\includegraphics[width=0.5\textwidth]{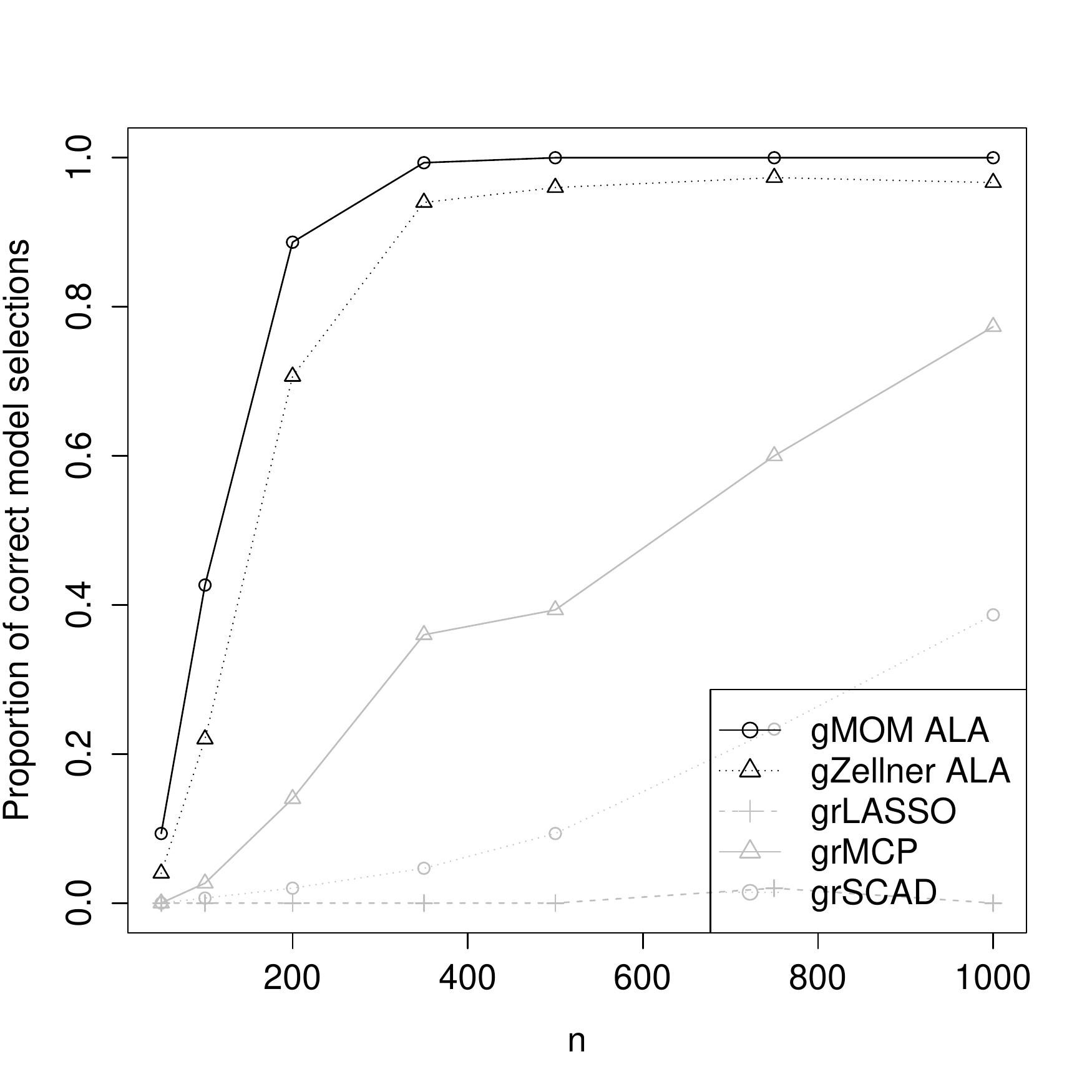} \\
Categorical $x_{i1}$ ($\beta_1^*=0.3$, $J=5$, $p=6$) & Categorical $x_{i1}$ ($\beta_1^*=0.6$, $J=5$, $p=6$) \\
\includegraphics[width=0.5\textwidth]{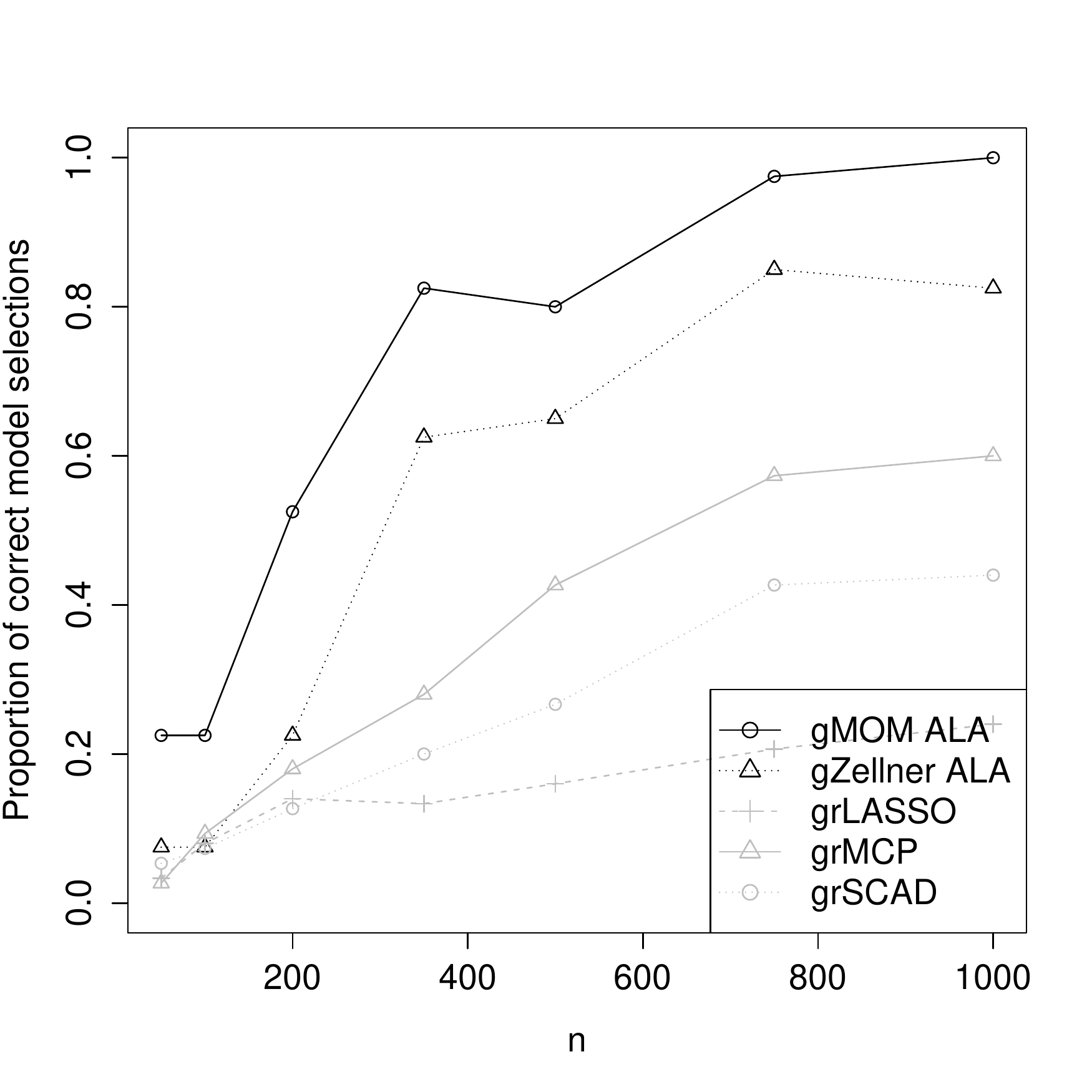} &
\includegraphics[width=0.5\textwidth]{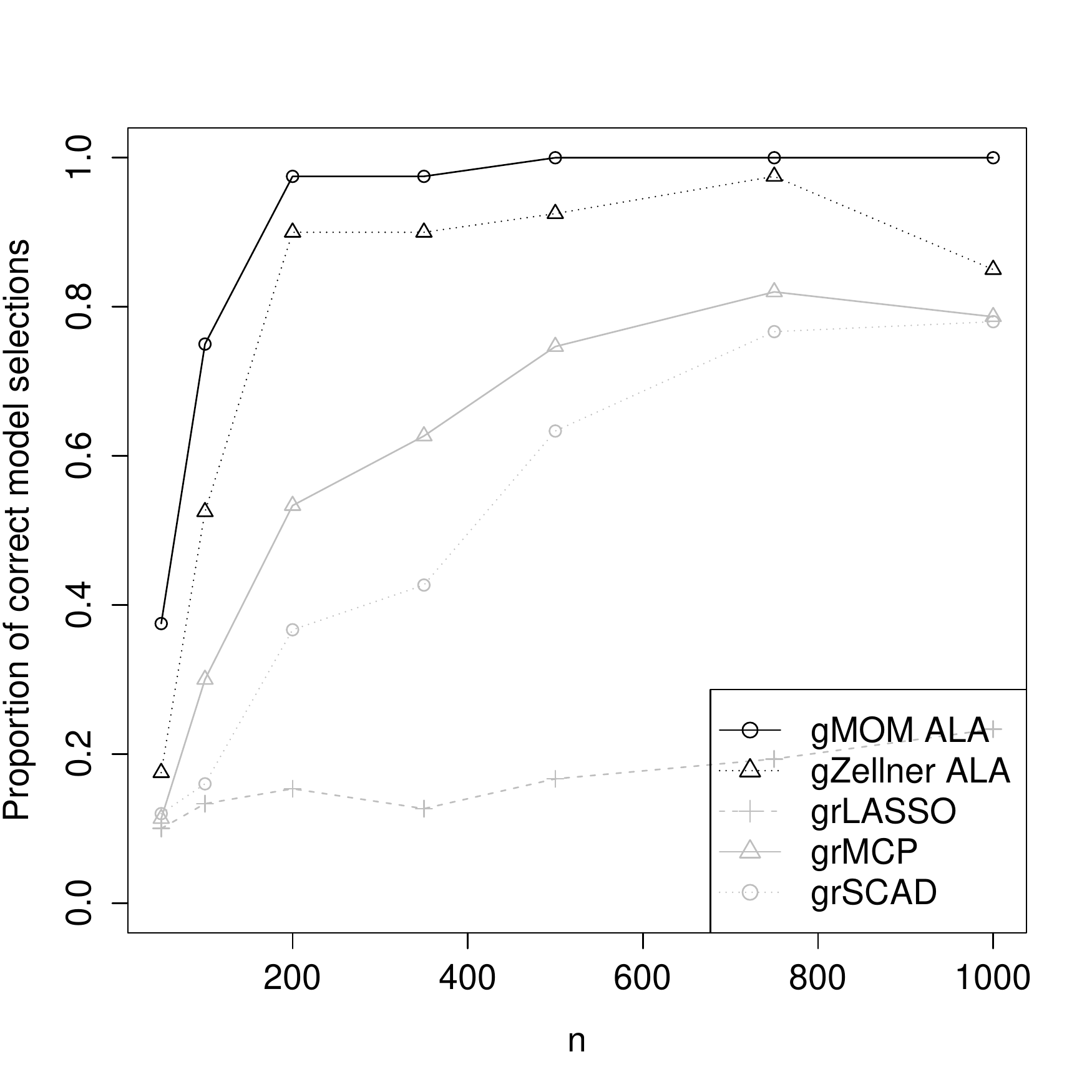}
\end{tabular}
\end{center}
  \caption{Proportion of correct model selections in Gaussian simulations with groups defined by a categorical $x_{i1} \in \{1,2,3\}$ and true effects $(-\beta_1^*,\beta_1^*)$}
\label{fig:simlm_discrete}
\end{figure}

We present a simulation example where one wishes to incorporate group constraints. We report averaged results across 150 simulations.
Continuous covariates were generated from a multivariate Normal with zero mean, unit variances and 0.5 pairwise correlations.
The regularization parameter for grLASSO, grMCP and grSCAD was set via 10-fold cross-validation.

In the example there are $J$ groups. The first group corresponds to a categorical covariate $x_{i1} \in \{1,2,3\}$, coded in the design matrix via two binary indicators $(z_{i1}, z_{i2}) \in \{0,1\}^2$, and the remaining $J-1$ groups to individual continuous covariates $x_{i2},\ldots,x_{ip}$. Hence, the design matrix $Z$ is $n \times (J+1)$, where $z_{ij+1}= x_{ij}$ for $j=2,\ldots,J$.
The data generating truth was
\begin{equation*}
  y_i = - \beta_1^* z_{i1} + \beta_1^* z_{i2} + z_{i3} 0.5 + \epsilon_i
\end{equation*}
where $\epsilon= (\epsilon_1,\ldots,\epsilon_n) \sim \mathcal{N}(0,I)$.
We considered $n \in [50,1000]$, $J \in \{5, 50\}$ groups and $\beta_1^* \in \{0.3,0.6\}$.

The top panels in Figure \ref{fig:simlm_discrete} show the proportion of correct model selections in the $J=50$ covariates case, and the bottom panels for the $J=5$ case.
The proportion of correct selections for ALA-based inference under gMOM priors was higher than for exact calculations under the gZellner prior, which was in turn higher than for the three considered penalized likelihood methods.



\begin{figure}
\begin{center}
\begin{tabular}{cc}
$n=50$, $x_i \sim \mathcal{N}(0,I)$ & $n=50$, $x_i \sim \mathcal{N}(0,V)$ \\
\includegraphics[width=0.45\textwidth, height=0.44\textwidth]{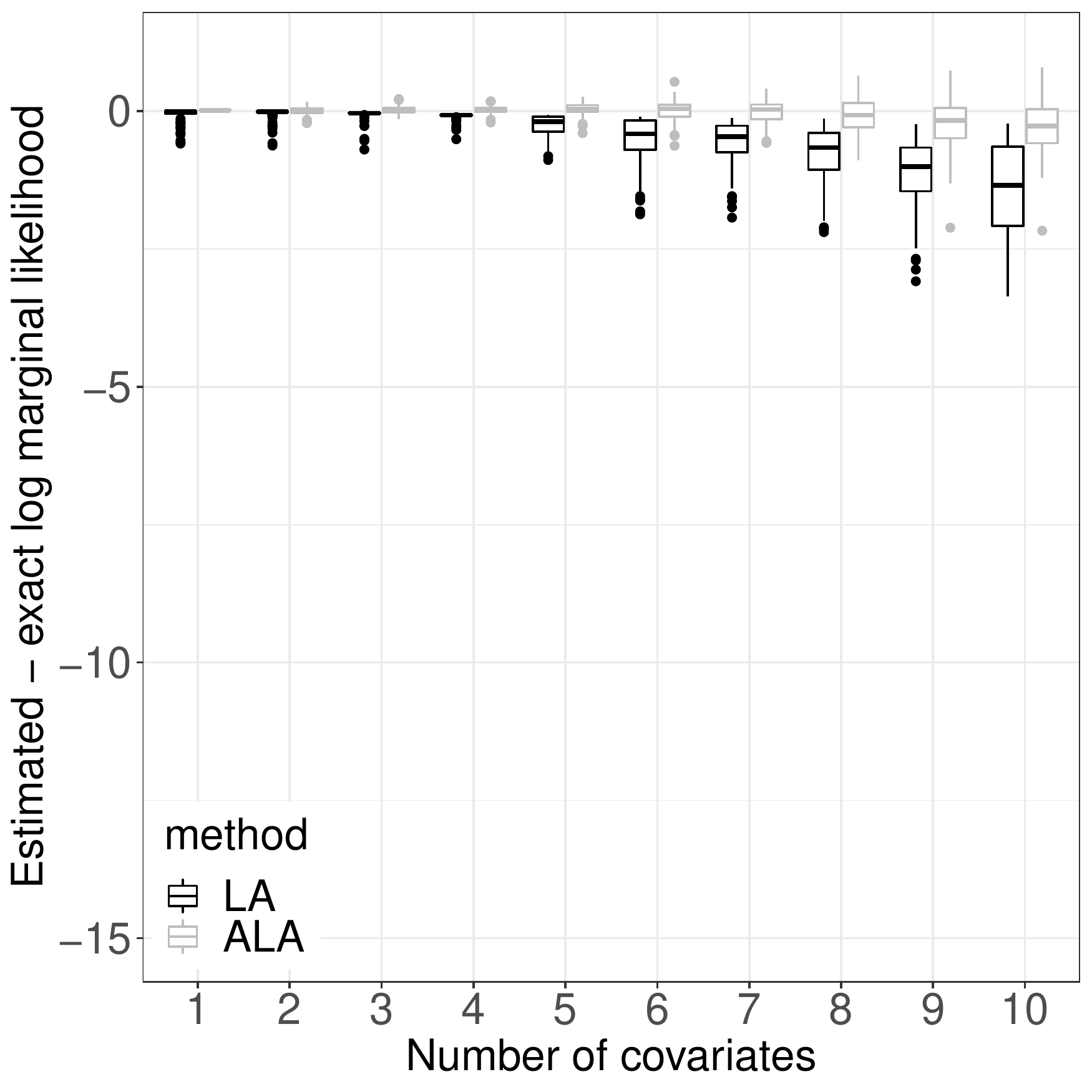} &
\includegraphics[width=0.45\textwidth, height=0.44\textwidth]{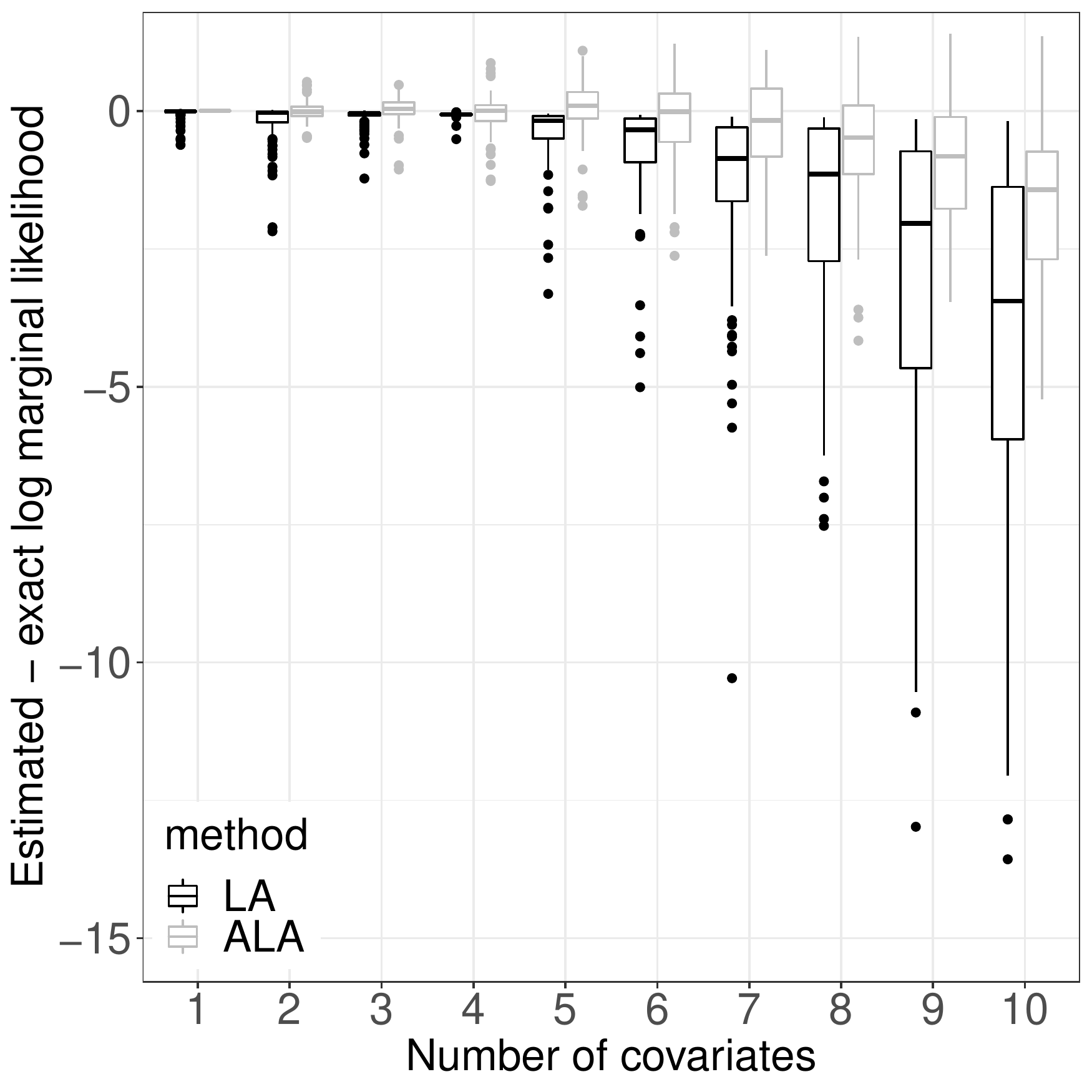} \\
$n=200$, $x_i \sim \mathcal{N}(0,I)$ & $n=200$, $x_i \sim \mathcal{N}(0,V)$ \\
\includegraphics[width=0.45\textwidth, height=0.44\textwidth]{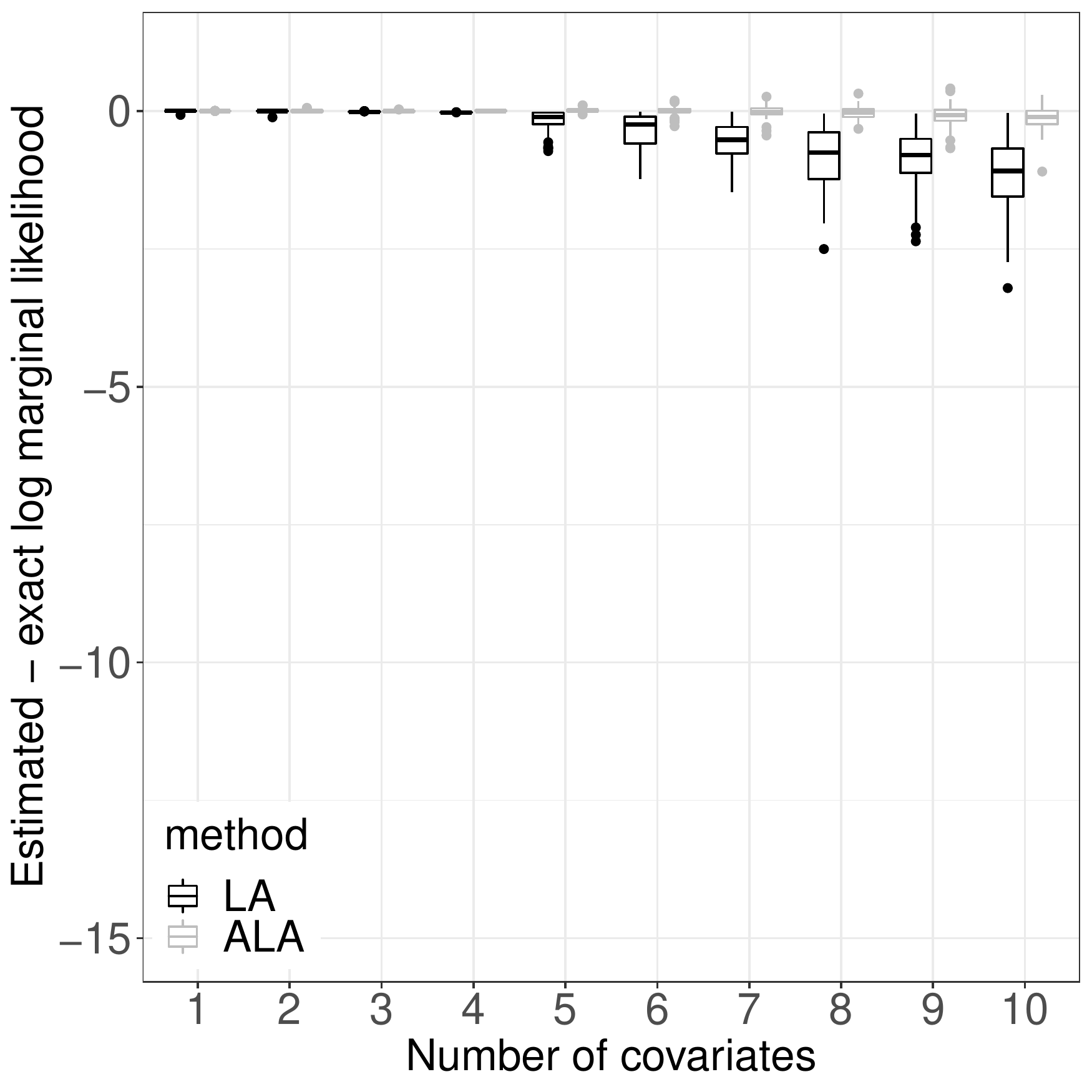} &
\includegraphics[width=0.45\textwidth, height=0.44\textwidth]{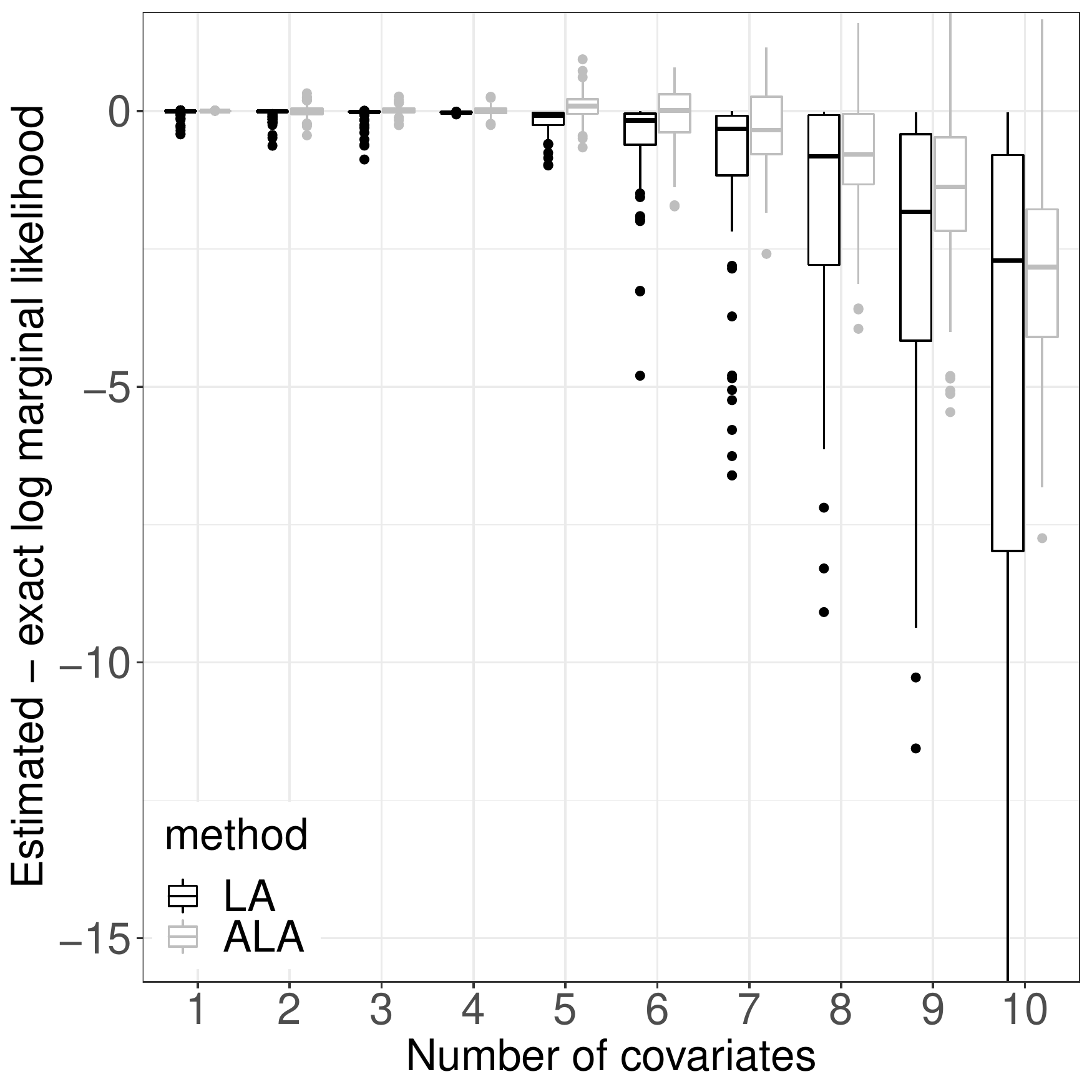} \\
$n=500$, $x_i \sim \mathcal{N}(0,I)$ & $n=500$, $x_i \sim \mathcal{N}(0,V)$ \\
\includegraphics[width=0.45\textwidth, height=0.44\textwidth]{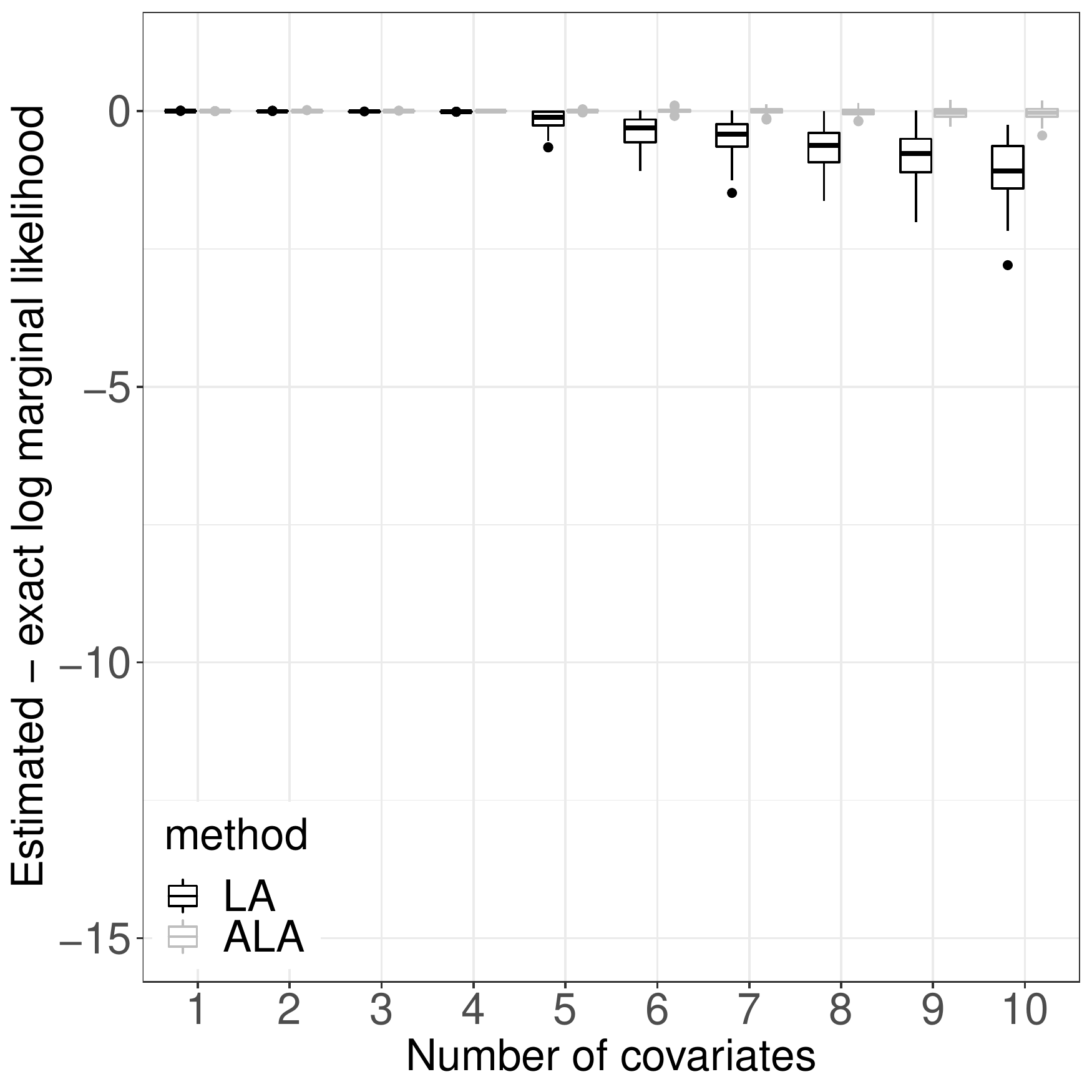} &
\includegraphics[width=0.45\textwidth, height=0.44\textwidth]{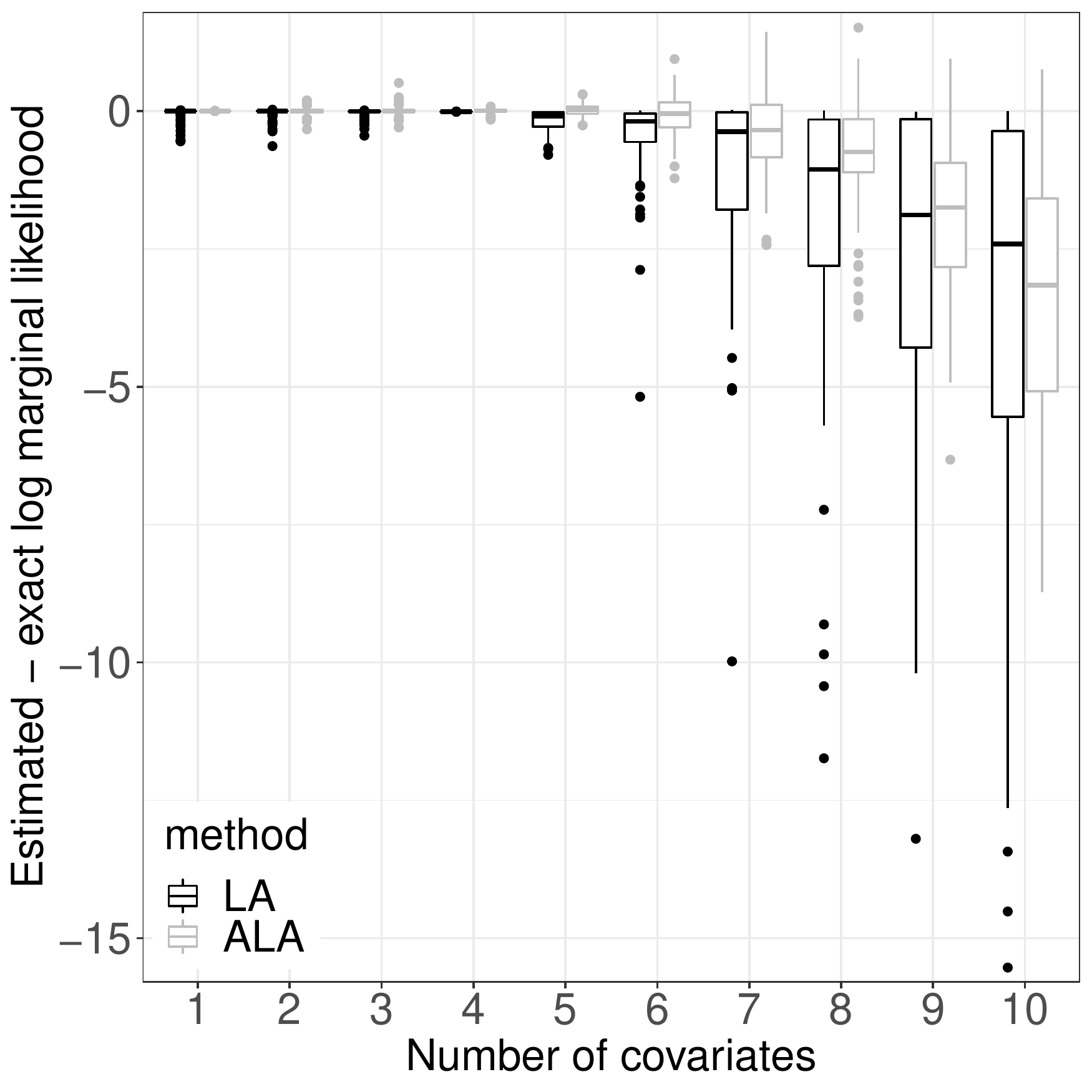} 
\end{tabular}
\end{center}
\caption{Simulated linear regression, gMOM prior.
Mean error $\log (\tilde{p}^N(y \mid \gamma) / p^N(y \mid \gamma))$ for $x_i \sim \mathcal{N}(0,I)$ and $x_i \sim \mathcal{N}(0,V)$, random non-diagonal $V$
}
\label{fig:ala_la_error_boxplot}
\end{figure}

\subsection{Poverty line data}
\label{ssec:povertyline_suppl}

The listing below provides the names of the covariates in the poverty line example, the name of the corresponding variables in the dataset, the codes corresponding to each category, and a basic descriptive summary providing the number of individuals in each category. 

Table \ref{tab:poverty_maineffects_alamodel_mle} displays the maximum likelihood parameter estimates for the model selected by ALA (marginal posterior inclusion probability $>0.5$) in the main effects analysis.
Table \ref{tab:poverty_selected} displays the estimated marginal posterior inclusion probabilities by LA and ALA for each group of variables corresponding to main effects and interactions.
Table \ref{tab:poverty_alamodel_mle} provides the maximum likelihood estimates for variables selected by ALA.

For a discussion of the variables selected by ALA and LA, and also GLASSO in the main effects analysis, see Section \ref{ssec:salary_data}.
Regarding GLASSO in the interactions analysis, it selected 4 main effects and 16 interaction terms. For 1 out of the 16 interactions (hispanic vs citizenship) the two corresponding main effects were also selected.
The main effects were hispanic, race, citizenship and occupation.
The interactions were female vs. marital, race and difficulty, marital status vs age, citizenship and hours worked, hispanic vs. citizenship and worker class,
age vs citizen, class worker vs age, moving state and hours worked,
and hours worked vs education, age, firmsize and moving state.

\scriptsize
\begin{verbatim}
-------------------------------------------------------------------------------
Female (female)

    0     1 
45798 43957 

-------------------------------------------------------------------------------
Hispanic origin (hispanic)

    0     1 
72001 17754 

-------------------------------------------------------------------------------
Marital status (marital)

    Divorced      Married NeverMarried    Separated      Widowed 
        9088        54068        23419         1986         1194 

-------------------------------------------------------------------------------
Education level (edu)

   Adv     CG    HSD No HSD  ProfD     SC 
  9245  21021  26054   7438    966  25031 

No HSD: no high-school diploma; HSD: high-school diploma; ProfD: professional 
degree; SC: some college; CG: college graduate; Adv: advanced degree (MSc/PhD)

-------------------------------------------------------------------------------
Age (age)

  18   19   20   21   22   23   24   25   26   27   28   29   30   31   32   33 
 313  578  863 1009 1253 1510 1737 1949 1907 2083 2178 2304 2252 2260 2217 2259 
  34   35   36   37   38   39   40   41   42   43   44   45   46   47   48   49 
2283 2379 2264 2388 2470 2419 2540 2413 2277 2276 2346 2389 2355 2357 2386 2300 
  50   51   52   53   54   55   56   57   58   59   60   61   62   63   64   65 
2206 2163 2040 1999 2077 1977 1832 1656 1635 1537 1374 1292 1084 1004  764  601 

-------------------------------------------------------------------------------

Race (race)

                    Asian                     Black Hawaiian/Pacific Islander 
                     6388                     11403                       533 
          Native American                     White 
                     1179                     70252 

-------------------------------------------------------------------------------
Citizenship status (citizen)

Born abroad, US parents                 Born US                 Citizen 
                    805                   71581                    7964 
            Not citizen 
                   9405 

-------------------------------------------------------------------------------
Nativity status (nativity)

      Foreign born    Foreign parents  One native parent Two native parents 
             18770               4451               2872              63573 
           Unknown 
                89 

-------------------------------------------------------------------------------
Occupation (occ)

 architect/engineer   arts/sports/media business operations      computer/maths 
               1644                1362                2399                3318 
       construction           education          extraction             farming 
               5072                5508                 143                 559 
            finance                food              health        installation 
               2332                3537                7866                3596 
              legal         maintenance          management              office 
                926                3460                8681               13031 
      personal care          production          protective               sales 
               2427                6326                1990                7310 
            science      social service          technician      transportation 
               1002                1788                 392                5086 

-------------------------------------------------------------------------------
Firm size (firmsize)

          0    1 (1-24)   2 (25-99) 3 (100-499) 4 (500-999)   5 (>1000) 
       1424       23023        9071       12191        5102       38944 

-------------------------------------------------------------------------------
Presence of any impairment/difficulty (difficulty)

    0     1 
87380  2375 

-------------------------------------------------------------------------------
Class of worker (classworker)

Government employee               Other       Self-employed         Wage/salary 
              15615                1455                4816               67869 

-------------------------------------------------------------------------------
Moved state (movedstate)

From out US     From US          No 
        238        8820       80697 

-------------------------------------------------------------------------------
Number of hours worked / week (hoursworked)

   35    36    37    38    39    40 
 4505  1814   765  1680   237 80754 
-------------------------------------------------------------------------------
\end{verbatim}
\normalsize

\small
\begin{longtable}{rrc} \hline
 & MLE & 95\% CI  \\ \hline
  age & -0.32 & (-0.364,-0.281) \\ 
  female & 0.33 & (0.248,0.41) \\ 
  difficulty & 0.34 & (0.154,0.533) \\ 
  hispanic & 0.25 & (0.149,0.344) \\ 
  hours worked & -0.10 & (-0.125,-0.07) \\ 
\hline
 Marital status & & \\
  Divorced & 0 & \\
  Married & -0.76 & (-0.871,-0.644) \\ 
  NeverMarried & -0.31 & (-0.43,-0.184) \\ 
  Separated & 0.24 & (0.061,0.424) \\ 
  Widowed & -0.21 & (-0.525,0.096) \\
\hline
 Education & & \\
No HSD & 0 & \\
  HSD & -0.59 & (-0.688,-0.493) \\ 
  ProfD & -1.83 & (-2.552,-1.117) \\ 
  SC & -0.96 & (-1.077,-0.85) \\ 
  CG & -1.44 & (-1.585,-1.294) \\ 
  Adv & -1.99 & (-2.251,-1.731) \\ 
\hline
Race & & \\
Asian & 0 & \\
  Black & 0.53 & (0.347,0.707) \\ 
  Hawaiian/Pacific Islander & -0.01 & (-0.446,0.431) \\ 
  Native American & 0.71 & (0.436,0.984) \\ 
  White & -0.03 & (-0.202,0.14) \\ 
\hline
Citizenship & & \\
Born abroad, US parents & 0 & \\
  Born US & 0.61 & (0.093,1.135) \\ 
  Citizen & 0.77 & (0.24,1.304) \\ 
  Not citizen & 1.15 & (0.62,1.671) \\ 
\hline
Occupation & & \\
architect/Engineer & 0 & \\
arts/sports/media & 0.68 & (0.017,1.342) \\ 
business operations & -0.15 & (-0.84,0.549) \\ 
computer/maths & -0.08 & (-0.753,0.593) \\ 
construction & 0.88 & (0.283,1.466) \\ 
education & 0.99 & (0.381,1.593) \\ 
extraction & 0.74 & (-0.311,1.786) \\ 
farming & 1.30 & (0.668,1.941) \\ 
finance & -0.14 & (-0.838,0.568) \\ 
food & 1.37 & (0.782,1.96) \\ 
health & 0.65 & (0.056,1.241) \\ 
installation & 0.59 & (-0.016,1.201) \\ 
legal & -0.85 & (-1.923,0.219) \\ 
maintenance & 1.19 & (0.596,1.781) \\ 
management & 0.27 & (-0.329,0.868) \\ 
office & 0.38 & (-0.208,0.969) \\ 
personal care & 1.02 & (0.418,1.617) \\ 
production & 0.74 & (0.153,1.334) \\ 
protective & 0.62 & (-0.03,1.275) \\ 
sales & 0.87 & (0.283,1.459) \\ 
science & 0.54 & (-0.235,1.319) \\ 
social service & 0.74 & (0.085,1.395) \\ 
technician & 0.39 & (-0.563,1.346) \\ 
transportation & 0.86 & (0.271,1.454) \\ 
\hline
Firm size & & \\
0 & 0 & \\
 1-24 & -0.37 & (-1.209,0.46) \\ 
 25-99 & -0.63 & (-1.474,0.213) \\ 
 100-499 & -0.85 & (-1.694,-0.006) \\ 
 500-999 & -0.73 & (-1.581,0.12) \\ 
 $>$1000 & -0.77 & (-1.614,0.065) \\ 
\hline
Worker class & & \\
Wage/salary & 0 & \\
Self-employed & 0.73 & (0.591,0.863) \\ 
Gov employee & -0.44 & (-0.589,-0.283) \\ 
Other & 2.08 & (1.247,2.905) \\ 
\hline
Moved state & & \\
From out US & 0 & \\
From US     & -0.84 & (-1.224,-0.451) \\ 
No          & -1.35 & (-1.733,-0.973) \\ 
   \hline
\caption{Model selected by ALA in poverty line example. Point estimates and 95\% confidence intervals from a maximum likelihood fit}
\label{tab:poverty_maineffects_alamodel_mle}
\end{longtable}
\normalsize

\begin{table}
\begin{center}
\begin{tabular}{lrr}
   & $P(\gamma_j=1 \mid y)$ (LA) & $\tilde{P}(\gamma_j=1 \mid y)$ (ALA) \\ \hline
  female & 1.000 & 1.000 \\ 
  hispanic & 1.000 & 1.000 \\ 
  marital & 1.000 & 1.000 \\ 
  education & 1.000 & 1.000 \\ 
  age & 1.000 & 1.000 \\ 
  race & 1.000 & 1.000 \\ 
  citizen & 1.000 & 1.000 \\ 
  occupation & 1.000 & 1.000 \\ 
  firmsize & 1.000 & 1.000 \\ 
  classworker & 1.000 & 1.000 \\ 
  movedstate & 1.000 & 1.000 \\ 
  hoursworked & 1.000 & 1.000 \\ 
  female:marital & 1.000 & 1.000 \\ 
  female:edu & 0.982 & 0.996 \\ 
  difficulty & 0.979 & 0.936 \\ 
  hispanic:marital & 1.000 & 0.530 \\ 
  female:hispanic & 0.501 & 0.989 \\ 
  maritalMarried:edu & 0.000 & 1.000 \\ 
  hispanic:movedstate & 0.000 & 1.000 \\ 
  maritalMarried:movedstate & 0.000 & 1.000 \\ 
  edu:movedstate & 0.000 & 1.000 \\ 
  edu:firmsize & 0.990 & 0.000 \\ 
  female:race & 0.211 & 0.000 \\ 
  hispanic:firmsize & 0.087 & 0.000 \\ 
  hispanic:race & 0.018 & 0.000 \\ 
  race:movedstate & 0.000 & 0.009 \\ 
  female:movedstate & 0.000 & 0.007 \\ 
  female:firmsize & 0.005 & 0.000 \\ 
  nativity & 0.001 & 0.000 \\ 
   \hline
\end{tabular}
\end{center}
\caption{Poverty data. Terms with either LA or ALA marginal posterior inclusion probability $>0.001$}
\label{tab:poverty_selected}
\end{table}

\small
\begin{longtable}{rrc} \hline
 & MLE & 95\% CI  \\ \hline
  female & 0.93 & (0.655,1.202) \\ 
  hispanic & 0.49 & (-0.29,1.272) \\ 
  maritalMarried & -0.08 & (-1.534,1.368) \\ 
  maritalNeverMarried & -0.25 & (-1.698,1.194) \\ 
  maritalSeparated & -14.19 & (-1166.616,1138.246) \\ 
  maritalWidowed & -0.14 & (-1.13,0.842) \\ 
  eduHSD & -0.33 & (-1.32,0.661) \\ 
  eduProfD & -10.74 & (-200.424,178.944) \\ 
  eduSC & 0.67 & (-0.524,1.854) \\ 
  eduCG & -0.05 & (-1.115,1.01) \\ 
  eduAdv & -0.19 & (-1.645,1.257) \\ 
  age & -0.36 & (-0.406,-0.324) \\ 
  raceBlack & 0.32 & (0.142,0.493) \\ 
  raceHawaiian/Pacific Islander & -0.03 & (-0.452,0.39) \\ 
  raceNative American & 0.49 & (0.221,0.754) \\ 
  raceWhite & -0.15 & (-0.322,0.011) \\ 
  citizenBorn US & 0.50 & (0.005,0.993) \\ 
  citizenCitizen & 0.62 & (0.112,1.121) \\ 
  citizenNot citizen & 0.99 & (0.496,1.493) \\ 
  occarts/sports/media & 1.04 & (0.387,1.703) \\ 
  occbusiness operations & 0.11 & (-0.578,0.803) \\ 
  occcomputer/maths & 0.04 & (-0.625,0.709) \\ 
  occconstruction & 1.29 & (0.705,1.885) \\ 
  occeducation & 1.12 & (0.519,1.724) \\ 
  occextraction & 0.92 & (-0.095,1.929) \\ 
  occfarming & 1.63 & (0.996,2.26) \\ 
  occfinance & 0.16 & (-0.54,0.855) \\ 
  occfood & 1.61 & (1.025,2.2) \\ 
  occhealth & 0.91 & (0.315,1.497) \\ 
  occinstallation & 0.89 & (0.281,1.493) \\ 
  occlegal & -0.45 & (-1.507,0.612) \\ 
  occmaintenance & 1.56 & (0.975,2.155) \\ 
  occmanagement & 0.66 & (0.064,1.255) \\ 
  occoffice & 0.65 & (0.062,1.235) \\ 
  occpersonal care & 1.51 & (0.918,2.11) \\ 
  occproduction & 1.04 & (0.454,1.631) \\ 
  occprotective & 0.59 & (-0.053,1.236) \\ 
  occsales & 1.24 & (0.658,1.829) \\ 
  occscience & 0.62 & (-0.152,1.391) \\ 
  occsocial service & 0.94 & (0.289,1.593) \\ 
  occtechnician & 0.51 & (-0.434,1.461) \\ 
  occtransportation & 1.19 & (0.597,1.776) \\ 
  difficulty1 & 0.34 & (0.161,0.528) \\ 
  movedstateFrom US & -1.03 & (-2.637,0.571) \\ 
  movedstateNo & -1.62 & (-3.198,-0.036) \\ 
  hoursworked & -0.11 & (-0.138,-0.084) \\ 
  female:hispanic & -0.18 & (-0.34,-0.017) \\ 
  female:maritalMarried & -1.19 & (-1.431,-0.944) \\ 
  female:maritalNeverMarried & 0.02 & (-0.212,0.26) \\ 
  female:maritalSeparated & 0.01 & (-0.386,0.402) \\ 
  female:maritalWidowed & 0.04 & (-0.781,0.857) \\ 
  hispanic:maritalMarried & 0.46 & (0.18,0.729) \\ 
  hispanic:maritalNeverMarried & -0.04 & (-0.304,0.229) \\ 
  hispanic:maritalSeparated & 0.27 & (-0.126,0.665) \\ 
  hispanic:maritalWidowed & 0.33 & (-0.384,1.034) \\ 
  female:eduHSD & -0.08 & (-0.272,0.114) \\ 
  female:eduProfD & 0.24 & (-1.261,1.732) \\ 
  female:eduSC & -0.15 & (-0.374,0.071) \\ 
  female:eduCG & -0.66 & (-0.933,-0.38) \\ 
  female:eduAdv & -0.86 & (-1.389,-0.338) \\ 
  maritalMarried:eduHSD & -0.30 & (-0.64,0.046) \\ 
  maritalNeverMarried:eduHSD & -0.19 & (-0.526,0.154) \\ 
  maritalSeparated:eduHSD & -0.13 & (-0.609,0.346) \\ 
  maritalWidowed:eduHSD & -0.30 & (-1.114,0.52) \\ 
  maritalMarried:eduProfD & 10.53 & (-179.127,200.195) \\ 
  maritalNeverMarried:eduProfD & 8.68 & (-180.987,198.35) \\ 
  maritalSeparated:eduProfD & -1.23 & (-486.004,483.542) \\ 
  maritalWidowed:eduProfD & 1.33 & (-746.496,749.153) \\ 
  maritalMarried:eduSC & -0.78 & (-1.145,-0.41) \\ 
  maritalNeverMarried:eduSC & -0.52 & (-0.869,-0.162) \\ 
  maritalSeparated:eduSC & -0.17 & (-0.691,0.345) \\ 
  maritalWidowed:eduSC & -0.37 & (-1.24,0.508) \\ 
  maritalMarried:eduCG & -0.64 & (-1.095,-0.182) \\ 
  maritalNeverMarried:eduCG & -0.47 & (-0.91,-0.027) \\ 
  maritalSeparated:eduCG & -0.55 & (-1.353,0.262) \\ 
  maritalWidowed:eduCG & 0.00 & (-1.141,1.147) \\ 
  maritalMarried:eduAdv & -0.89 & (-1.658,-0.127) \\ 
  maritalNeverMarried:eduAdv & -0.55 & (-1.306,0.205) \\ 
  maritalSeparated:eduAdv & -1.39 & (-3.487,0.715) \\ 
  maritalWidowed:eduAdv & -0.10 & (-2.267,2.057) \\ 
  hispanic:movedstateFrom US & -0.49 & (-1.259,0.269) \\ 
  hispanic:movedstateNo & -0.32 & (-1.063,0.423) \\ 
  maritalMarried:movedstateFrom US & -0.01 & (-1.449,1.427) \\ 
  maritalNeverMarried:movedstateFrom US & 0.39 & (-1.041,1.814) \\ 
  maritalSeparated:movedstateFrom US & 14.68 & (-1137.747,1167.115) \\ 
  maritalWidowed:movedstateFrom US & 0.19 & (-0.768,1.14) \\ 
  maritalMarried:movedstateNo & 0.21 & (-1.211,1.623) \\ 
  maritalNeverMarried:movedstateNo & 0.37 & (-1.042,1.782) \\ 
  maritalSeparated:movedstateNo & 14.53 & (-1137.9,1166.962) \\ 
  eduHSD:movedstateFrom US & -0.02 & (-0.985,0.947) \\ 
  eduProfD:movedstateFrom US & -0.41 & (-3.5,2.671) \\ 
  eduSC:movedstateFrom US & -1.09 & (-2.26,0.08) \\ 
  eduCG:movedstateFrom US & -0.53 & (-1.549,0.491) \\ 
  eduAdv:movedstateFrom US & -0.67 & (-2.041,0.702) \\ 
  eduHSD:movedstateNo & -0.04 & (-0.977,0.908) \\ 
  eduProfD:movedstateNo & -1.56 & (-4.357,1.238) \\ 
  eduSC:movedstateNo & -1.09 & (-2.24,0.053) \\ 
  eduCG:movedstateNo & -0.61 & (-1.594,0.373) \\ 
  eduAdv:movedstateNo & -0.87 & (-2.154,0.419) \\ 
   \hline
\caption{Model selected by ALA in poverty line example. Point estimates and 95\% confidence intervals from a maximum likelihood fit}
\label{tab:poverty_alamodel_mle}
\end{longtable}
\normalsize

\subsection{Colon cancer survival}
\label{ssec:coloncancer_suppl}

\begin{table}
\begin{center}
\begin{tabular}{|c|ccc|c|} \hline
         & \multicolumn{3}{c|}{Selected variables} & Concordance Index \\
         & Original & Original + Fake & Overlap & \\
gMOM ALA     & 3 &  3  & 3  & 0.64 \\
gZellner ALA &11 & 12  & 1  & 0.66 \\
Cox LASSO    &11 & 14  & 9  & 0.69 \\
AFT LASSO    & 6 & 1   & 1  & 0.51 \\
\hline
\end{tabular}
\end{center}
\caption{Colon cancer survival data. Number of selected variables within original $p=175$, original + 50 fake variables, overlap between these two sets, and concordance index (leave-one-out cross-validation)}
\label{tab:colon}
\end{table}

\cite{calon:2012} showed a strong association between the average expression of 172 genes related to fibroblasts 
and that of growth factor TGFB were associated with lower colon cancer survival (time until recurrence), in human patients in cancer stages 1-3.
We perform a deeper analysis of their data where we seek to identify which individual genes are associated with survival, in addition to the clinical variable tumor stage, for a total of $p=175$.
The data records survival times for $n=260$ patients. Out of these only 50 were observed, and the remaining 210 were censored, posing a challenging inference problem.

Table \ref{tab:cputime_survival} reports run times for the ALA- and LA-based analyses, for illustration both under a uniform prior on the model space, and under the Beta-Binomial prior. The ALA brought significant speed-ups, e.g. from 55 hours to 3.3-11 minutes under the uniform model prior.
Given that the computational exercise under the uniform prior is particularly considerable, we focus on this analysis and the ALA-based results, as the cost of the LA is impractical in this setting.

In empirical data it is hard to evaluate what method performs best in terms of model selection. To explore this issue, we first assessed the predictive ability of each method via the leave-one-out cross-validated concordance index \citep{harrell:1996}. Briefly, higher values of the index indicate a higher proportion of pairs of predicted survival times that matched the order of the observed survival times.
Cox-LASSO, gMOM-ALA and gZellner-ALA achieved similar indexes (Table \ref{tab:colon}), whereas for AFT-LASSO it was lower.

Next, to acknowledge that our goal is not prediction but model selection, we added 50 fake genes that do not truly have an effect on survival but are correlated with the original genes. Specifically, we selected the first 50 genes in the original data, and each of the 50 fake genes was obtained by adding standard Gaussian noise to the original genes. The average correlation between each fake versus original gene was 0.70. We then run all methods on the combined original data plus the 50 fake genes.
The results are in Table \ref{tab:colon}.
Interestingly, gMOM-ALA returned the same highest posterior probability model in both datasets,
which included genes ESM1, GAS1 and PDPN, agreeing with current theory on non-local priors helping control false discoveries.
In contrast the remaining methods returned a different result, in particular for gZellner-ALA and AFT-LASSO the overlap with the original analysis was poor.

To assess the biological plausibility of the results, according to {genecards.org} \citep{stelzer:2016}
ESM1 is related to endothelium disorders, growth factor receptor binding and gastric cancer networks,
GAS1 plays a role in growth and tumor suppression, and PDPN mediates effects on cell migration and adhesion, all of these are processes potentially associated with tumor growth and metastasis.
In fact ESM1 was also selected by gZellner-ALA and Cox-LASSO, and GAS1 by AFT-LASSO.
We remark that these results offer no guarantee that the gMOM-ALA solution matches better the unknown biological truth than other methods, but they illustrate that the use of ALA led to significant computational savings and its combination with non-local priors can help reduce false positives.

\bibliographystyle{plainnat}
\bibliography{references}

\end{document}